\documentclass[11pt]{amsart}
\usepackage[margin=1in]{geometry}
\usepackage{amsfonts, float, mathtools}
\usepackage{amssymb}
\usepackage[dvips]{graphics}
\usepackage{epsfig}
\usepackage{listings}
\usepackage{euscript}
\usepackage{color}
\usepackage{cite}
\usepackage{subcaption}

\usepackage[all]{xy}
\usepackage{multirow}
\usepackage{tikz}
\usetikzlibrary{matrix,arrows,decorations.pathmorphing,positioning,calc}
\definecolor{Green}{rgb}{0,.8,.4}
\definecolor{Plum}{rgb}{.5,0,1}
\definecolor{cyan}{rgb}{0.50,.9,0.9}
\definecolor{lightblue}{rgb}{0.3, 0.6, .7}
\definecolor{ocre}{RGB}{243,102,25}
\definecolor{mygray}{RGB}{243,243,244}
\definecolor{burgundy}{rgb}{0.5, 0.0, 0.13}

\usepackage{fancyvrb}
\usepackage{bbm}

\newtheorem{thm}{Theorem}[section]
\newtheorem{cor}[thm]{Corollary}
\newtheorem{lemma}[thm]{Lemma}

\theoremstyle{definition}

\theoremstyle{remark}

\numberwithin{equation}{section}

\def\idty{{\mathchoice {\mathrm{1\mskip-4mu l}} {\mathrm{1\mskip-4mu l}} %
		{\mathrm{1\mskip-4.5mu l}} {\mathrm{1\mskip-5mu l}}}}
\newcommand{\vgg}{\Vec{\gamma}}

\newcommand{\bR}{{\mathbb R}}

\newcommand{\Ir}{{\mathbb Z}}
\newcommand{\bC}{{\mathbb C}}

\newcommand{\bN}{{\mathbb N}}
\newcommand{\bZ}{{\mathbb Z}}
\newcommand{\cA}{{\mathcal A}}

\newcommand{\cG}{{\mathcal G}}
\newcommand{\cH}{{\mathcal H}}

\newcommand{\supp}{\operatorname{supp}}

\newcommand{\cL}{{\mathcal L}}
\newcommand{\cP}{{\mathcal P}}

\newcommand{\cE}{{\mathcal E}}
\newcommand{\cW}{{\mathcal W}}

\newcommand{\caA}{{\mathcal A}}
\newcommand{\caB}{{\mathcal B}}
\newcommand{\caC}{{\mathcal C}}
\newcommand{\caD}{{\mathcal D}}
\newcommand{\caE}{{\mathcal E}}

\newcommand{\caG}{{\mathcal G}}
\newcommand{\caH}{{\mathcal H}}

\newcommand{\caK}{{\mathcal K}}
\newcommand{\caL}{{\mathcal L}}

\newcommand{\caN}{{\mathcal N}}

\newcommand{\caP}{{\mathcal P}}

\newcommand{\caS}{{\mathcal S}}
\newcommand{\caT}{{\mathcal T}}

\newcommand{\caV}{{\mathcal V}}
\newcommand{\caW}{{\mathcal W}}

\newcommand{\bbH}{{\mathbb H}}

\newcommand{\bbZ}{{\mathbb Z}}

\newcommand{\braket}[2]{\left\langle #1 , #2\right\rangle}

\newcommand{\id}{{\rm id}}

\newcommand{\N}{\mathbb{N}}

\newcommand{\spec}{\mathop{\rm spec}}

\newcommand{\eq}[1]{(\ref{#1})}

\newcommand{\be}{\begin{equation}}
	\newcommand{\ee}{\end{equation}}
\newcommand{\bea}{\begin{eqnarray}}
	\newcommand{\eea}{\end{eqnarray}}
\newcommand{\beann}{\begin{eqnarray*}}
	\newcommand{\eeann}{\end{eqnarray*}}

\newcommand{\Rl}{\bR}

\title[LTQO for AKLT models on hexagonal and Lieb lattices]{Local Topological Quantum Order\\ and Spectral Gap Stability for the AKLT Models\\ on the Hexagonal and Lieb Lattices}

\author[T.~Jackson]{Thomas Jackson} 
\address{\textnormal{(Thomas Jackson)} Department of Mathematical Sciences, United Arab Emirates University, AL Ain, UAE}

\author[B.~Nachtergaele]{Bruno Nachtergaele} 
\address{\textnormal{(Bruno Nachtergaele)} Department of Mathematics and Center for Quantum Mathematics and Physics, University of California,  Davis, CA  95616-8633, USA}
\thanks{Corresponding author: Bruno Nachtergaele, email: bnachtergaele@ucdavis.edu}

\author[A.~Young]{Amanda Young}
\address{\textnormal{(Amanda Young)} Department of Mathematics, University of Illinois Urbana-Champaign, Urbana, IL 61801, USA}
\thanks{\copyright 2026 by the authors. This article may be reproduced, in its entirety, for non-commercial purposes.}

\date{May 4, 2026}

\begin{document}

	\begin{abstract}
We prove that the ground state of the AKLT models on the hexagonal lattice and the Lieb lattice satisfy the local topological quantum order (LTQO) condition. This will be a consequence of proving that the finite volume ground states are indistinguishable from a unique infinite volume ground state. Concretely, we identify a sequence of increasing and absorbing finite volumes for which any finite volume ground state expectation is well approximated by the infinite volume state with error decaying at a uniform exponential rate in the distance between the support of the observable and boundary of the finite volume. As a corollary to the LTQO property, we obtain that the spectral gap above the ground state in these models is stable under general small perturbations of sufficient decay. We prove these results by a detailed analysis of the polymer representation of the ground states state derived  by Kennedy, Lieb and Tasaki (1988) with the necessary modifications required for proving the strong form of ground state indistinguishability needed for LTQO.
	\end{abstract}
\maketitle

	\begin{center}
		\emph{Dedicated to the memory of Huzihiro Araki}
	\end{center}	

	\section{Introduction and Main Results}\label{sec:Introduction}
	
	\subsection{Introduction}
	
	The mathematical study of gapped ground states of quantum spin systems has yielded a broad variety of fascinating results over the past two decades.  We speak about a gapped ground state phase when there is a positive gap in the spectrum of the Hamiltonian of the infinite system that is stable under a broad class of gently perturbed Hamiltonians. These gapped ground state phases, which in their simplest form are described by a unique ground state without ordinary long-range order, may exhibit so-called  topological orders. Characterizing and mathematically classifying the possible toplogical orders has been a very active and fruitful area of research. For a selection of mathematical results of this type see \cite{naaijkens:2013,cha:2020,ogata:2021a,ogata:2021e,bourne:2021,kapustin:2021,kapustin:2022,jones:2025,bols:2025}. 
	
	While the spectral gap question is typically rather trivial for concrete commuting Hamiltonians, even for small perturbations of classical models the question whether a model is gapped or gapless is, in general, undecidable \cite{cubitt:2022}. For non-commuting interactions, even frustration-free ones, the question has been addressed in only a few examples in two or more dimensions: PVBS models \cite{bachmann:2015,bishop:2016a,bishop:2019}, AKLT models on decorated lattices \cite{abdul-rahman:2020,pomata:2020,lucia:2023}, and the AKLT model on the hexagonal lattice \cite{lemm:2020a,pomata:2020}.
	
	Given the existence of a positive gap, the next question to address is its stability under gentle perturbations of the Hamiltonian. The goal is to show that the gap remains open after the perturbation of the model by a broad class of interactions as long they are sufficiently small.  This typically involves proving that the perturbation is relatively (form) bounded with respect to the unperturbed one after a suitable spectrum-preserving transformation. Strategies of varying generality to achieve such bounds have been developed. For example, the first spectral gap stability result for an AKLT model relied on a cluster expansion  \cite{yarotsky:2006}, as do some later results \cite{de-roeck:2019,bjornberg:2021}. In \cite{frohlich:2020,del-vecchio:2021,del-vecchio:2022} the authors introduce and apply a transformation they call Lie-Schwinger block diagonalization to prove relative boundedness. The most general results have been obtained using the Bravyi-Hastings-Michalakis (BHM) strategy \cite{bravyi:2010,bravyi:2011,michalakis:2013,nachtergaele:2022}. The BHM strategy is sufficiently general to apply to AKLT models including the two-dimensional AKLT models we consider in this work. The crucial property one needs to apply this approach  is called Local Topological Quantum Order (LTQO) of the ground states which can alternatively be stated as an indistinguishability property of finite-volume ground states (see below). In this paper we prove this property for the AKLT models on the hexagonal lattice and on the Lieb lattice. 
	
	Apart from being an essential ingredient in proofs of spectral gap stability, LTQO is a structural property relevant for characterizing the quantum code provided by  the ground state of Kitaev's quantum double and the Levin-Wen string-net models \cite{kitaev:2006, levin:2005, cui:2020} and a closely related condition is employed in the construction of the boundary algebra describing topological order \cite{jones:2025}.
	
	AKLT models, which were introduced by Affleck, Kennedy, Lieb, and Tasaki in \cite{affleck:1987a,affleck:1988}, are frustration-free models with $SU(2)$ invariant nearest neighbor interactions. For their general definition see Section \ref{sec:main_results}.
	For one-dimensional AKLT models (defined on $\Ir$ or a more general quasi-one-dimensional lattice), generally speaking, the ground states, the spectral gap, and spectral gap stability, can be analyzed systematically. For example, if such a model has a finite number of infinite-volume gapped ground states, these states satisfy LTQO and the gap above those ground states is stable \cite[Appendix B]{nachtergaele:2022}.
	
	The main objective of this paper is to prove the LTQO property for the hexagonal and Lieb lattices. The AKLT model on the hexagonal lattice has a spin 3/2 degree of freedom at each site and a nearest neighbor interaction given by the orthogonal projection on the total spin 3 states for each pair of nearest neighbor spin. The Lieb lattice is a decorated version of the regular square lattice. It can be viewed as a square lattice with one vertex added at each of the edges (see Figure \ref{fig:lambda2}). The corresponding AKLT model has a spin 2 at each degree 4 site and a spin 1 at each site of degree 2. The interaction is nearest neighbor and given by the spin 3 projection for each nearest neighbor pair of the Lieb lattice.

	In \cite{affleck:1988} it was conjectured that the AKLT models in regular higher-dimensional lattices would have a unique ground state in the thermodynamic limit and a spectral gap if the degree of the vertices is at most 4. This includes the hexagonal and square lattices. For higher degree lattices one expects gapless ground states and N\'eel order. This dependence on the degree is also supported by results for the AKLT model on trees \cite{fannes:1992b,jackson:2025}. AKLT models on so-called decorated lattices are also of interest. For example, their ground state can serve as a resource for universal quantum computation \cite{wei:2011,wei:2014}. In this paper we study the models on the hexagonal lattice and the square lattice with regular decoration consisting of $m$ vertices inserted at each edge. We prove LTQO for all these models except for the square lattice, for which we have to assume $m\geq 1$. The undecorated square lattice is expected to have a unique gapped ground state as suggested by the exponential decay of the spin-spin correlation function proved in \cite{kennedy:1988}.  We note that there exist several families
	of $SU(N)$-invariant  generalizations of the AKLT models. A particular family introduced in  \cite{affleck:1988} involves alternating conjugate symmetric representations (with the spin dimension determined by the degree of the lattice site). This is well-defined for bipartite graphs, and has a similar cluster expansion as the ones used in this paper. LTQO results of the present paper can be generalized to those models \cite{jackson:2026}.
	
	We now briefly summarize prior results about the ground state gap and the LTQO property of these and some related models.
	
	LTQO is an essential property for all existing spectral gap stability results that are applicable to the AKLT models, including the models on the hexagonal and Lieb lattices. The existence of a uniform gap for the hexagonal lattice was originally conjectured in \cite{affleck:1988}, and there is strong evidence to support this claim. Kennedy, Lieb, and Tasaki proved that the frustration-free ground state $\omega$ has exponential decay of correlations \cite{kennedy:1988}, a property that generically occurs together with a non-vanishing gap in the spectrum above the ground state. AKLT models on decorated hexagonal lattices have also been shown to be uniformly gapped \cite{abdul-rahman:2020,lucia:2023}, and convincing corroboration of a uniform gap for the (undecorated) hexagonal lattice model was independently obtained by combining analytical methods with sophisticated numerical techniques in two almost simultaneously appearing papers, \cite{lemm:2020a, pomata:2020}. The results in \cite{pomata:2019} also indicate a gap for the AKLT model on the Lieb lattice.
	
	LTQO and, as a consequence, spectral gap stability has previously been shown in \cite{lucia:2023} for the AKLT on the decorated hexagonal lattice with decoration number five or more. Here we show LTQO for the AKLT model on the hexagonal lattice with any number of decoration, including the standard, undecorated lattice. We also show LTQO for the AKLT model on the decorated square lattice with decoration number $m\geq 1$, which includes the Lieb lattice. Our proofs rely very much on the approach taken in \cite{lucia:2023} and the polymer representation of \cite{kennedy:1988}.  The main difference consists of the way we use the convergence criteria of Koteck\'y and Preiss \cite{kotecky:1986} and Ueltschi \cite{ueltschi:2004}. Since the arguments in \cite{pomata:2019} show the AKLT model on the Lieb lattice to be gapped, we also obtain that the gap is stable. Furthermore, we use the cluster expansion techniques to show exponential decay of the two-point correlation for the same class of models.

	\subsection{Main Results}\label{sec:main_results} 
We consider AKLT models defined on either the hexagonal or square lattice, denoted $\bbH$ and $\bZ^2$, respectively, as well as $m$-decorated versions of the lattice. Letting $\Gamma\in\{\bbH,\bZ^2\}$ and $\Lambda=(\caV_\Lambda, \caE_\Lambda)$ any subgraph of $\Gamma$, the $m$-decorated volume $\Lambda^{(m)}$, $m\in \bN$, is the graph obtained from adding  $m$-vertices to every edge of $\Lambda$. In particular, we call $\Gamma^{(m)}$ the $m$-decorated lattice. We will also write $\Lambda^{(0)}$ to denote the original undecorated subgraph when convenient.

To define the AKLT model of interest, for each vertex $v\in\caV_{\Gamma^{(m)}}$, we take the onsite state space to be a $2s_v+1$-dimensional complex Hilbert space, denoted $\caH_v$, where
\[
s_v = \deg_{\Gamma^{(m)}}(v)/2.
\]
As such, vertices belonging to the undecorated lattice have spin $3/2$ or $2,$ depending on if $\Gamma = \bbH$ or $\Gamma=\bZ^2$, and all vertices decorating $\Gamma$ are spin $1$. For any edge $e=(v,w)\in\caE_{\Gamma^{(m)}}$ the nearest neighbor interaction, denoted $P_e$, is the orthogonal projection onto the subspace of total spin $s_v+s_w$ contained in $\caH_v\otimes\caH_w$. Thus, the local Hamiltonian for any finite  subgraph $\Lambda = (\caV_{\Lambda}, \caE_{\Lambda})\subseteq \Gamma^{(m)}$ is
\begin{equation}\label{Hamiltonian}
H_{\Lambda}=\sum_{e\in \caE_{\Lambda}} P_e \in \caA_{\Lambda}
\end{equation}
where
\[\caA_{\Lambda}=\caB(\caH_{\Lambda}), \quad \caH_{\Lambda}=\bigotimes_{x\in\caV_{\Lambda}}\caH_x\] 
are the algebra of local observables and Hilbert space of states associated with $\Lambda$, respectively. Here, we use that $\caA_X\subseteq \caA_{\Lambda}$ via the map $A\mapsto A\otimes\idty_{\Lambda\setminus X}$ for any $X\subseteq \Lambda$. For completeness, recall that the support of an observable $A\in \caA_{\Lambda}$, denoted $\supp(A)$, is the smallest $X\subseteq \Lambda$ such that $A\in\caA_X$. We note that while \eqref{Hamiltonian} holds for any finite subgraph $\Lambda\subseteq \Gamma^{(m)}$, the volumes considered in this work are always obtained by decorating subgraphs of the original lattice $\Gamma$.

The aim of this work is to establish that the expectations of finite volume ground states for these AKLT models are indistinguishable when the observable is supported sufficiently far from the boundary of the finite volume, see Theorem~\ref{thm:gs_indistinguishability} below. AKLT models are well-known to be frustration-free, meaning that for any finite graph $\Lambda\subseteq \Gamma^{(m)}$, $H_\Lambda\geq 0$ and its ground state space is given by
\[
\caG_\Lambda = \ker(H_\Lambda).
\]
We will denote the orthogonal projection onto the ground state space by $G_{\Lambda}$.

The indistinguishability result will follow from showing that any finite volume ground state is well approximated by the same infinite volume ground state,
\[\omega^{(m)}:\caA_{\Gamma^{(m)}}\to \bC, \quad \caA_{\Gamma^{(m)}} =\overline{\caA_{\Gamma^{(m)}}^{\rm loc}}^{\|\cdot\|}, \quad \caA_{\Gamma^{(m)}}^{\rm loc} = \bigcup_{\substack{\Lambda\subset \Gamma^{(m)} \\ |\caV_\Lambda|<\infty}}\caA_\Lambda\,.
\]
The state $\omega^{(m)}$ of interest is frustration-free, which in the infinite volume context means
\[
\omega^{(m)}(P_e) = 0 \quad \forall e\in \caE_{\Gamma^{(m)}}\,.
\]
Such states are known to exist as they are the weak-$*$ limits of any sequence of normalized finite volume ground states, $\psi_{N}\in \caG_{\Lambda_N}$, $\Lambda_N\uparrow \Gamma^{(m)}$.

For gapped AKLT models, one expects that there is a unique frustration-free ground state, $\omega^{(m)}$. The uniqueness of this state was proved in \cite{kennedy:1988} for the case of the hexagonal lattice (i.e. $m=0$), which was subsequently generalized to decorated hexagonal models with $m\geq 5$ in \cite{lucia:2024}. It was moreover shown in \cite{kennedy:1988} that $\omega^{(0)}$ has exponential decay of correlations. The uniqueness of the frustration-free ground state in the case of $m$-decorated square lattices, $m\geq 1$, and $m$-decorated hexagonal lattices with $1\leq m \leq 4$ will actually follow from the proof of our indistinguishability result, i.e. Theorem~\ref{thm:gs_indistinguishability}. The uniqueness of the frustration-free state for the square lattice is still an open question. If this state is shown or known to be unique, then it immediately follows that
\begin{equation}\label{fv_gs_limit}
	\lim_{N\to\infty}\braket{\psi_N}{A\psi_{N}} = \omega^{(m)}(A)\qquad \forall A\in\caA_{\Gamma_N^{(m)}}^{\rm loc},
\end{equation}
for any sequence of normalized ground states $\psi_{N}\in \caG_{\Lambda_N}$ such that $\Lambda_N\uparrow \Gamma^{(m)}$. The novelty of Theorem~\ref{thm:gs_indistinguishability} is establishing that the finite volume ground states converge exponentially fast to $\omega^{(m)}$ in the distance between the support of the observable to the boundary of $\Lambda_N$, see \eqref{eq:indistinguishability}.

The main result is stated with respect to a specific sequence of increasing and absorbing finite volumes $\Lambda_N\subseteq \Gamma$. In the case that $\Gamma$ is the hexagonal lattice, these volumes are defined as follows. For any vertex $\tilde{x}\in\tilde{\Gamma}$, the dual lattice of $\Gamma$, let $h_{\tilde{x}}\subseteq \Gamma$ denote the finite subgraph obtained from the union of the hexagon centered at $\tilde{x}$ with all edges $e\in \caE_{\Gamma}$ connected to this hexagon, see Figure~\ref{fig:lambda2}. Then, for a fixed vertex $\tilde{0}\in\tilde{\Gamma}$ define
\begin{equation}\label{Lambda_N}
	\Lambda_N = \bigcup_{\tilde{x}\in b_{\tilde{0}}(N-1)} h_{\tilde{x}} \subseteq\bbH 
\end{equation}
where $b_{\tilde{0}}(N)\subseteq \tilde{\Gamma}$ is the closed ball of radius $N\in\bN$ centered at $\tilde{0}$ with respect to the graph distance in $\tilde{\Gamma},$ see {Figure \ref{fig:lambda2}}. 
When $\Gamma$ is the square lattice, we define 
\begin{equation}
	\Lambda_N = \bigcup_{x\in b_0^\infty(N-1)}b_x^1(1)\subseteq \bZ^2 \label{Lambda_N_Square}
\end{equation} 
where $b_x^p(N)$ is the closed ball of radius $N\in\bN$ at $x\in \bZ^2$ with respect to the standard $p$-norm, see again Figure~\ref{fig:lambda2}. In the case of either lattice $\Gamma$, we let $\Lambda_N^{(m)}\subseteq \Gamma^{(m)}$, $m\in \bN_0$, be the volume given by decorating each edge of $\Lambda_N\subseteq \Gamma$ with $m$ additional particles.

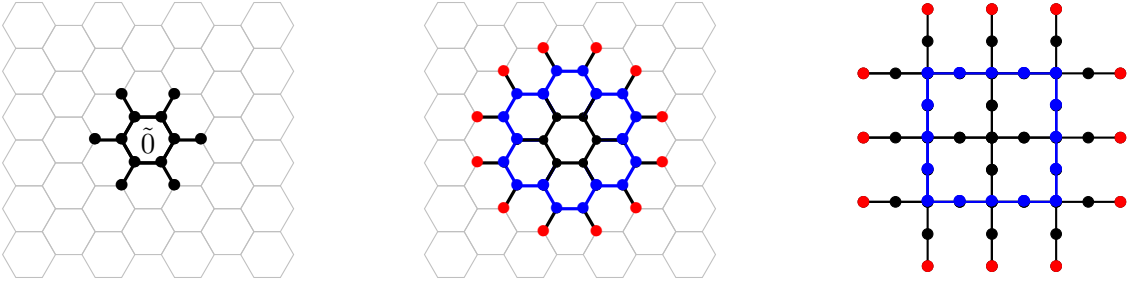
\begin{figure}
	\begin{subfigure}{0.3\textwidth}
		\begin{center}
			\begin{tikzpicture}
				\def\n{.35};
				\newcommand{\hexcoord}[2]
				{[shift=(60:#1),shift=(120:#1),shift=(0:#2),shift=(-60:#2),shift=(0:#2),shift=(60:#2)]}
				\foreach \x in {3,...,7}
				\foreach \y in {1,...,3}{
					\draw[color=black!25]\hexcoord{\x*\n}{\y*\n}
					(0:\n)--(60:\n)--(120:\n)--(180:\n)--(-120:\n)--(-60:\n)--cycle;
					\draw[color=black!25,shift=(-60:\n), shift=(0:\n)]\hexcoord{\x*\n}{\y*\n}
					(0:\n)--(60:\n)--(120:\n)--(180:\n)--(-120:\n)--(-60:\n)--cycle;
				}
				\foreach \x in {3,...,8}{
					\draw[color=black!25,shift=(-60:\n), shift=(0:\n)]\hexcoord{\x*\n}{0}
					(0:\n)--(60:\n)--(120:\n)--(180:\n)--(-120:\n)--(-60:\n)--cycle;}
				\foreach \y in {1,...,3}{\draw[color=black!25,shift=(-60:\n), shift=(0:\n)]\hexcoord{8*\n}{\y*\n}
					(0:\n)--(60:\n)--(120:\n)--(180:\n)--(-120:\n)--(-60:\n)--cycle;
				}
				
				\foreach \k in {3}
				\foreach \x in {0,60,120,180,240,300}
				\foreach \y in {240,300}
				{\draw[shift=(\x+120:\n), shift=(\x+60:\n), very thick, color=black]\hexcoord{5*\n}{5*\n-\k*\n}
					(\x+\y:\n)--(\x+\y+60:\n);
					\draw[shift=(\x+180:\n), shift=(\x:\n), fill=black]\hexcoord{5*\n}{2*\n} (\x:2*\n) node[black]{$\bullet$};
					\draw[shift=(\x:\n), shift=(\x:\n), fill=black]\hexcoord{5*\n}{2*\n} (\x:-1*\n) node[black]{$\bullet$};
				}
				
				\foreach \k in {3}
				\foreach \x in {0,60,120,180,240,300}
				\foreach \y in {240}
				{\draw[shift=(\x+120:\n), shift=(\x+60:\n), very thick, color=black]\hexcoord{5*\n}{5*\n-\k*\n}
					(\x+\y:\n)--(\x+\y+60:\n);
					\draw[shift=(\x+120:\n), shift=(\x+60:\n), fill=red]\hexcoord{5*\n}{5*\n-\k*\n} ({\n*(cos(\x+\y)},{\n*(sin(\x+\y)});
				}
				
				\draw\hexcoord{5*\n}{2*\n} (60:\n)--(120:\n) node[midway, below, yshift=0pt] {$\tilde{0}$};
			\end{tikzpicture}
			
		\end{center}
	\end{subfigure}\hspace{.5cm}
	\begin{subfigure}{0.3\textwidth}
		\begin{center}
			\begin{tikzpicture}
				\def\n{.35};
				\newcommand{\hexcoord}[2]
				{[shift=(60:#1),shift=(120:#1),shift=(0:#2),shift=(-60:#2),shift=(0:#2),shift=(60:#2)]}
				\foreach \x in {3,...,7}
				\foreach \y in {4,...,6}{
					\draw[color=black!25]\hexcoord{\x*\n}{\y*\n}
					(0:\n)--(60:\n)--(120:\n)--(180:\n)--(-120:\n)--(-60:\n)--cycle;
					\draw[color=black!25,shift=(-60:\n), shift=(0:\n)]\hexcoord{\x*\n}{\y*\n}
					(0:\n)--(60:\n)--(120:\n)--(180:\n)--(-120:\n)--(-60:\n)--cycle;
				}
				\foreach \x in {3,...,7}{
					\draw[color=black!25,shift=(-60:\n), shift=(0:\n)]\hexcoord{\x*\n}{3*\n}
					(0:\n)--(60:\n)--(120:\n)--(180:\n)--(-120:\n)--(-60:\n)--cycle;}
				\foreach \y in {3,...,6}{\draw[color=black!25,shift=(-60:\n), shift=(0:\n)]\hexcoord{8*\n}{\y*\n}
					(0:\n)--(60:\n)--(120:\n)--(180:\n)--(-120:\n)--(-60:\n)--cycle;
				}
				
				\foreach \k in {0}
				\foreach \x in {0,60,120,180,240,300}
				\foreach \y in {0,60,120,180,300}
				{\draw[shift=(\x+120:\n), shift=(\x+60:\n), very thick, color=blue]\hexcoord{5*\n}{5*\n-\k*\n}
					(\x+\y:\n)--(\x+\y+60:\n);
					\draw[shift=(\x+120:\n), shift=(\x+60:\n), fill=black]\hexcoord{5*\n}{5*\n-\k*\n} ({\n*(cos(\x+\y)},{\n*(sin(\x+\y)}) circle (1.6pt);
				}
				\foreach \k in {0}
				\foreach \x in {0,60,120,180,240,300}
				\foreach \y in {240,300}
				{\draw[shift=(\x+120:\n), shift=(\x+60:\n), very thick, color=black]\hexcoord{5*\n}{5*\n-\k*\n}
					(\x+\y:\n)--(\x+\y+60:\n);
				}
				
				\foreach \x in {0,60,120,180,240,300}
				\foreach \y in {60,120}{
					\draw[shift=(\x+120:\n), shift=(\x+60:\n), very thick]\hexcoord{5*\n}{5*\n} (\x+\y:\n)--(\x+\y:2*\n);
					\draw[shift=(\x+120:\n), shift=(\x+60:\n), fill=black]\hexcoord{5*\n}{5*\n} (\x+\y:2*\n) node[red]{$\bullet$};
					\draw[shift=(\x+120:\n), shift=(\x+60:\n), fill=black]\hexcoord{5*\n}{5*\n} (\x+\y:\n) node[blue]{$\bullet$};
					\draw[shift=(\x+180:\n), shift=(\x:\n), fill=black]\hexcoord{5*\n}{5*\n} (\x:2*\n) node[blue]{$\bullet$};}
			\end{tikzpicture}
			
		\end{center}
		
	\end{subfigure}\hspace{.5cm}
	\begin{subfigure}{0.3\textwidth}
		\begin{center}
			\begin{tikzpicture}
				\def\n{0.85}
				\foreach \i in {2,3} {
					\foreach \j in {2,...,4} {
						\filldraw[black] (\i*\n, \j*\n+1*\n) circle (1 pt);
						\filldraw[black] (\i*\n+1*\n, \j*\n) circle (1pt);
						\filldraw[black] (\i*\n-1*\n, \j*\n) circle (1pt);
						\filldraw[black] (\i*\n, \j*\n-1*\n) circle (1pt);
						
						\ifnum\i<5
						\filldraw[black] (\i*\n+0.5*\n, \j*\n) circle (1pt);
						\filldraw[black] (\i*\n-.5*\n, \j*\n) circle (1pt);
						\fi
						
						\ifnum\j<5
						\filldraw[black] (\i*\n, \j*\n+0.5*\n) circle (1pt);
						\filldraw[black] (\i*\n, \j*\n-0.5*\n) circle (1pt);
						\fi
					}						
				}
				\foreach \i in {2,3} {
					\foreach \j in {1,...,4} {
						\draw[thick] (\i*\n,\j*\n) -- (\i*\n,\j*\n+1*\n);
					}
				}
				\foreach \i in {1,2,3} {
					\foreach \j in {2,...,4} {
						\draw[thick] (\i*\n,\j*\n) -- (\i*\n+1*\n,\j*\n);
					}
				}
				\foreach \i in {3,4} {
					\foreach \j in {2,...,4} {
						\filldraw[black] (\i*\n, \j*\n+1*\n) circle (1pt);
						\filldraw[black] (\i*\n+1*\n, \j*\n) circle (1pt);
						\filldraw[black] (\i*\n-1*\n, \j*\n) circle (1pt);
						\filldraw[black] (\i*\n, \j*\n-1*\n) circle (1pt);
						
						\ifnum\i<5
						\filldraw[black] (\i*\n+0.5*\n, \j*\n) circle (1pt);
						\filldraw[black] (\i*\n-.5*\n, \j*\n) circle (1pt);
						\fi
						
						\ifnum\j<5
						\filldraw[black] (\i*\n, \j*\n+0.5*\n) circle (1pt);
						\filldraw[black] (\i*\n, \j*\n-0.5*\n) circle (1pt);
						\fi
					}
				}
				
				\foreach \i in {3,4} {
					\foreach \j in {1,...,4} {
						\draw[thick] (\i*\n,\j*\n) -- (\i*\n,\j*\n+1*\n);
					}
				}
				\foreach \i in {1,...,4} {
					\foreach \j in {2,...,4} {
						\draw[thick] (\i*\n,\j*\n) -- (\i*\n+1*\n,\j*\n);
					}
				}
				\foreach \i in {2,...,4} {
					\foreach \j in {2,3,4} {
						\filldraw[black] (\i*\n, \j*\n+1*\n) circle (2pt);
						\filldraw[black] (\i*\n+1*\n, \j*\n) circle (2pt);
						\filldraw[black] (\i*\n-1*\n, \j*\n) circle (2pt);
						\filldraw[black] (\i*\n, \j*\n-1*\n) circle (2pt);
						
						\ifnum\i<5
						\filldraw[black] (\i*\n+0.5*\n, \j*\n) circle (2pt);
						\filldraw[black] (\i*\n-.5*\n, \j*\n) circle (2pt);
						\fi
						
						\ifnum\j<5
						\filldraw[black] (\i*\n, \j*\n+0.5*\n) circle (2pt);
						\filldraw[black] (\i*\n, \j*\n-0.5*\n) circle (2pt);
						\fi
					}
				}
				\foreach \i in {2,3,4} {
					\foreach \j in {0,4} {
						\filldraw[red] (\i*\n, \j*\n+1*\n) circle (2pt);
					}
				}
				\foreach \i in {1,5} {
					\foreach \j in {1,2,3} {
						\filldraw[red] (\i*\n, \j*\n+1*\n) circle (2pt);
					}
				}
				\foreach \i in {2,3} {
					\foreach \j in {2,4} {
						\filldraw[blue,thick] (\i*\n,\j*\n) -- (\i*\n+1*\n,\j*\n);
						\filldraw[blue,thick] (2*\n,2*\n) -- (2*\n,4*\n);
						\filldraw[blue,thick] (4*\n,2*\n) -- (4*\n,4*\n);
						\filldraw[blue] (\i*\n, \j*\n) node[blue]{$\bullet$};
						\filldraw[blue] (2*\n, 2*\n) node[blue]{$\bullet$};
						\filldraw[blue] (4*\n, 4*\n) node[blue]{$\bullet$};
						\filldraw[blue] (\i*\n+\n, \j*\n) node[blue]{$\bullet$};
						\filldraw[blue] (2*\n, 4*\n) node[blue]{$\bullet$};
						\filldraw[blue] (4*\n, 2*\n) node[blue]{$\bullet$};
						\filldraw[blue] (2*\n, 3*\n) node[blue]{$\bullet$};
						\filldraw[blue] (4*\n, 3*\n) node[blue]{$\bullet$};
						\filldraw[blue] (2*\n, 3.5*\n) node[blue]{$\bullet$};
						\filldraw[blue] (4*\n, 3.5*\n) node[blue]{$\bullet$};
						\filldraw[blue] (2*\n, 2.5*\n) node[blue]{$\bullet$};
						\filldraw[blue] (4*\n, 2.5*\n) node[blue]{$\bullet$};\filldraw[blue] (2.5*\n, 4*\n) node[blue]{$\bullet$};																																			\filldraw[blue] (3.5*\n, 4*\n) node[blue]{$\bullet$};
						\filldraw[blue] (2.5*\n, 2*\n) node[blue]{$\bullet$};																																			\filldraw[blue] (3.5*\n, 2*\n) node[blue]{$\bullet$};
						
					}
				}
			\end{tikzpicture}
			
		\end{center}
		\label{fig:lb}
	\end{subfigure}\hspace{.5cm}
	\caption{Illustration of $h_{\tilde{0}}\subset\bbH$ (left), $\Lambda_{2}\subset\bbH$ (middle), and $\Lambda^{(1)}_2\subset\Gamma^{(1)}$ with $\Gamma=\bbZ^2$ (right). The red vertices in the middle and right figures indicate the boundary $\partial\Lambda_2$, while the blue edges and vertices in the middle and right figures denote $\gamma^{(2,m)}$ is the smallest polymer containing $\mathring{\Lambda}_2^{(m)}$.}\label{fig:lambda2}
\end{figure}

We now state our main result, which holds for all decoration parameters $m\geq m_\Gamma$, where
\[
m_\Gamma = \begin{cases}
	0, & \Gamma = \bbH\\
	1, & \Gamma = \bZ^2
\end{cases}
\]

\begin{thm}[Ground State Indistinguishability]\label{thm:gs_indistinguishability} Let $\Gamma\in\{\bbH,\bZ^2\}$, and $\Lambda_N^{(m)}$ the corresponding $m$-decorated sequence from \eqref{Lambda_N}-\eqref{Lambda_N_Square}. For all $m\geq m_\Gamma$, there is a unique frustration-free AKLT ground state $\omega^{(m)}:\caA_{\Gamma^{(m)}}\to \bC$. Moreover, there are $N_\Gamma, K_\Gamma\in\bN$ so that if $N>K+N_\Gamma>K_\Gamma+N_\Gamma$, then
	\begin{equation} \label{eq:indistinguishability}
		\left|\braket{\Psi}{A \Psi} - \omega^{(m)}(A)\right| \leq 2\|A\|F_m(N,K)e^{F_m(N,K)}
	\end{equation}
	for any $A\in \caA_{\Lambda_{K-1}^{(m)}}$ and normalized $\Psi \in \caG_{\Lambda_N^{(m)}}$, where 
	\[F_{m}(N,K) =(m+1)L_\Gamma(K)e^{-\eta_\Gamma(m+1)(N-K)}\]
	where $L_\Gamma:\bR\to\bR$, and the constant $\eta_\Gamma>0$ can be chosen as shown in \eq{linearL} and \eq{eta}.
\end{thm}
The explicit values one can take in the above result are
\begin{equation}\label{eta}
K_\Gamma = \begin{cases}
	25, & \Gamma = \bbH \\
	2, & \Gamma = \bZ^2
\end{cases},\qquad
N_\Gamma = \begin{cases}
	52, & \Gamma = \bbH \\
	4, & \Gamma = \bZ^2
\end{cases},\qquad
\eta_\Gamma = \begin{cases}
	.0172, & \Gamma = \bbH \\
	.046, & \Gamma = \bZ^2
\end{cases}
\end{equation}
and the polynomials are
\begin{equation}\label{linearL}
L_\Gamma(x) = 
\begin{cases}
	1.8(2x+1), & \Gamma = \bbH \\
	1.36x, & \Gamma = \bZ^2
\end{cases}\,.
\end{equation}

	To show this bound implies LTQO, let $\gamma^{(K)}\subseteq\Gamma$ denote the smallest loop containing the interior vertices of $\Lambda_K$, and $\gamma^{(K,m)}\subseteq \Gamma^{(m)}$ its $m$-decorated counterpart, see Figure~\ref{fig:lambda2}. The notion of interior used here will be made precise in Section~\ref{sec:notation}. For any $m\in\bN_0$, the number of edges in $\gamma^{(K,m)}$ satisfies $|\gamma^{(K,m)}| = (m+1)|\gamma^{(K)}|$, and the number of edges in the undecorated loop is
	\[
	|\gamma^{(K)}| = \begin{cases}
		6(2K-1), & \Gamma=\bbH \\
		8(K-1), & \Gamma = \bbZ^2
	\end{cases}
	\]
	 We note that $\gamma^{(K,m)}$ also encapsulates $\Lambda_{K-1}^{(m)}$ and its length can be considered as the surface area for this volume. If the value of $N$ in Theorem~\ref{thm:gs_indistinguishability} additionally satisfies $N\geq K + \eta_\Gamma^{-1}\ln(K)$, the bound in \eqref{eq:indistinguishability} implies 
	\begin{equation}\label{for_LTQO}
	\left|\braket{\Psi}{A \Psi} - \omega(A)\right| \leq C_\Gamma|\gamma^{(K,m)}|\|A\|e^{-\eta_\Gamma(N-K)} \qquad \forall A\in \caA_{\Lambda_{K-1}^{(m)}}
	\end{equation} 
	where 
	\begin{equation}\label{LTQO_constant}
	C_\Gamma = \frac{2L_\Gamma(K_\Gamma)}{|\gamma^{(K_\Gamma)}|}\exp\left(\frac{(m_\Gamma+1)L_\Gamma(K_\Gamma)}{K_\Gamma^{m_\Gamma+1}}\right)\,.
	\end{equation}
	Evaluating this expressions gives $C_\bbH\approx 24.5615$ and $C_{\bbZ^2}\approx2.4951$.
	
	\begin{cor}[LTQO]\label{cor:LTQO} Let $\Lambda_N^{(m)}$ denote the sequence from either \eqref{Lambda_N} or \eqref{Lambda_N_Square} depending on $\Gamma\in\{\bbH,\bZ^2\}$. Suppose that $m\geq m_\Gamma$, $K\geq K_\Gamma$ and $N \geq K+\max\{N_\Gamma, \, \eta_\Gamma^{-1}\ln(K)\}$. Then the projection of any $A\in \caA_{\Lambda_{K-1}^{(m)}}$ into the ground state space $\caG_{\Lambda_N^{(m)}}$ satisfies
		\begin{equation}\label{LTQO}
		\|G_{\Lambda_N^{(m)}}AG_{\Lambda_N^{(m)}}-\omega^{(m)}(A)G_{\Lambda_N^{(m)}}\| \leq 2C_\Gamma|\gamma^{(K,m)}|\|A\|e^{-\eta_\Gamma(m+1)(N-K)}
		\end{equation}
	where $\eta_\Gamma$ and $C_\Gamma$ are as in \eqref{eta} and \eqref{LTQO_constant}, respectively.
	\end{cor}
	Since the AKLT model on these graphs is translation-invariant, the above also holds for any translation of $\Lambda_N^{(m)}$. We give the proof of Corollary~\ref{cor:LTQO} assuming that Theorem~\ref{thm:gs_indistinguishability} holds.
	\begin{proof}
		If $A\in \caA_{\Lambda_{K-1}}$ is (anti-)hermitian, then by Theorem~\ref{thm:gs_indistinguishability} and \eqref{for_LTQO}, 
		\begin{align}
	\|G_{\Lambda_N^{(m)}}AG_{\Lambda_N^{(m)}}-\omega^{(m)}(A)G_{\Lambda_N^{(m)}}\| & = \sup_{\substack{\psi\in \caG_{\Lambda_N^{(m)}}\\ \|\psi\|=1}}	\left|\braket{\psi}{A \psi} - \omega^{(m)}(A)\right| \nonumber \\ 
	&\leq C_\Gamma|\gamma^{(K,m)}|\|A\|e^{-\eta_\Gamma(m+1)(N-K)}\label{sa_case}\,.
		\end{align}
	The result for any $A\in \caA_{\Lambda_{K-1}^{(m)}}$ follows from writing $A = \frac{A+A^*}{2}+ \frac{A-A^*}{2}$ and applying \eqref{sa_case} twice.
	\end{proof}

For completeness we now state a spectral gap stability result which, due to the LTQO property we prove, follows from known results. Gap stability results utilizing LTQO fall into two different categories depending on the type of spectral gap the model exhibits. A model is said to be uniformly gapped if there exists an increasing and absorbing sequence of finite volumes for which the spectral gaps of the associated local Hamiltonians are all uniformly bounded from below by the same nonzero constant. Alternatively, one can consider the spectral gap question for a ground state of the infinite volume system \cite{nachtergaele:2024}. Frustration-free models (which includes AKLT models) that exhibit LTQO necessarily have a unique infinite volume frustration-free ground state. 
The state is then called a gapped ground state if the GNS Hamiltonian has a gap in its spectrum above the ground state. When the lattice for the model has no boundary, the GNS gap is also referred to as the bulk gap. As such, the existence of the bulk gap immediately follows if an associated sequence of finite volume models has a uniform gap. However, in certain cases, the bulk gap can be positive even if the finite-volume gaps of the limiting sequence tends to zero. The LTQO condition one verifies for stability is the same regardless of the type of gap stability one wishes to prove.

For the AKLT model on the hexagonal and Lieb lattices the stability of the uniform gap follows from Corollary~\ref{cor:LTQO} given that the model is uniformly gapped across the sequence $\Lambda_N$. A similar statement for the bulk gap can be obtained by making small adjustments to the statement of \cite[Theorem~2.2]{lucia:2023}. 

Let $H_{\Lambda_N^{(m)}}$ be the AKLT Hamiltonian for the finite volume $\Lambda_N^{(m)}$ as defined in the sentence following \eq{Lambda_N_Square}. Suppose $V_{n,x}\in \caA_{b^1_x(n)}$, $n\geq 1, x\in \Gamma^{(m)}$ for the respective model, self-adjoint, and define 
\be\label{perturbed}
H^V_{\Lambda_N^{(m)}}(s)  = H_{\Lambda_N^{(m)}} + s \sum_{\substack{n\geq 1, x\in \Gamma\\ b^1_x(n)\subset \Lambda_N}} V_{n,x},\quad s\in\Rl.
\ee
	
\begin{cor}[Spectral Gap Stability\cite{nachtergaele:2022}]\label{cor:stability}
Consider the AKLT Hamiltonians $H_{\Lambda_N^{(m)}}$ with $m\geq 0$ for the decorated or undecorated  hexagonal lattice and $m\geq 1$ for the decorated square lattice. Assume there exists $\gamma>0$, and $N_0\geq 1$ such that $\spec H_{\Lambda_N^{(m)}}\cap (0,\gamma)=\emptyset$, for all $N\geq N_0$, and that there are constants $a>0$, and $\theta \in (0,1]$ such that
$$
\Vert V_{n,x}\Vert \leq e^{-an^\theta}, \mbox{ for all } n\geq 1, \; x\in\Gamma^{(m)}\,.
$$
Then, for all $\gamma_0\in (0,\gamma)$, there exists $s(\gamma_0)>0$ and a positive, decreasing sequence $\epsilon_N\to 0$, such that 
$$
\spec H^V_{\Lambda_N^{(m)}}(s) \subseteq [-\epsilon_N,\epsilon_N]\cup [\gamma_0+\epsilon_N,\infty) =\emptyset ,\mbox{ for all } N \mbox{ large enough, and  } |s| < s(\gamma_0).
$$
\end{cor}	

The precise details of the stability argument were explicitly worked out for the $m$-decorated hexagonal model in the case that $m\geq 5$ in \cite{lucia:2024}, and apply up to mild modifications for the models considered here given the LTQO result.		

\subsection{Summary of Sections}
The rest of this work is organized as follows. We conclude Section~\ref{sec:Introduction} by introducing some notation and definitions that will be used throughout the work. 

In Section~\ref{sec:AKLT_GSS} we review the Weyl representation for the AKLT model on any graph $G$ and how it is used to describe the ground state space for any local Hamiltonian associated to a finite subgraph $\Lambda\subseteq G$. We use the specific form of the ground state expectations in this representation to define two maps, the bulk-boundary map and bulk state, and outline how these states will be used to prove Theorem~\ref{thm:gs_indistinguishability} for the two cases of interest. The main objective will be to write a portion of the integral defining the bulk-boundary map or bulk state in terms of a polymer representation that has a convergent cluster expansion. 

Section~\ref{sec:hex_polymers} is focused on establishing the necessary polymer representation for the (un)decorated hexagonal lattice. We then turn to reviewing a key cluster expansion convergence criterion due to Koteck\'y and Preiss \cite{kotecky:1986} which was later generalized by Ueltschi \cite{ueltschi:2004} in Section~\ref{sec:KP}, and show that this criterion is satisfied by the polymer representation for the hexagonal model in Section~\ref{sec:Cluster_convergence}. The proof of the cluster expansion convergence criterion for the hexagonal model follows closely the approach used in \cite{kennedy:1988}, with differences stemming from choices that result in having a larger set of polymers in our representation than what was used in the previous work. Using the cluster expansion, we then prove Theorem~\ref{thm:gs_indistinguishability} in the case that $\Gamma=\bbH$ in Section~\ref{sec:proof}. 

The case that $\Gamma=\bbZ^2$ is the content of Section~\ref{sec:lieb}, where we establish the polymer representation and its convergent cluster expansion. We then show that a slight modification of the proof from Section~\ref{sec:proof} also then holds in this context, which completes the proof of Theorem~\ref{thm:gs_indistinguishability}.

Finally, in Appendix~\ref{sec:correlation_decay} we show that the convergence cluster expansion implies that the infinite volume ground state exhibits exponential decay of two-point correlations, and in Appendix~\ref{sec:computer_code} give the computer code we used for a number of counting arguments in the proof that the cluster expansion convergence criterion holds for both the hexagonal and square models.

\subsection{Notation}\label{sec:notation}
	\begin{figure}
	\begin{subfigure}{0.4\textwidth}
		\begin{tikzpicture}
			\def\n{.3};
			\newcommand{\hexcoord}[2]
			{[shift=(60:#1),shift=(120:#1),shift=(0:#2),shift=(-60:#2),shift=(0:#2),shift=(60:#2)]}
			\foreach \x in {-6,...,6}
			\foreach \y in {-3,...,2}{
				\draw[color=black!25]\hexcoord{\x*\n}{\y*\n}
				(0:\n)--(60:\n)--(120:\n)--(180:\n)--(-120:\n)--(-60:\n)--cycle;
				\draw[color=black!25,shift=(-60:\n), shift=(0:\n)]\hexcoord{\x*\n}{\y*\n}
				(0:\n)--(60:\n)--(120:\n)--(180:\n)--(-120:\n)--(-60:\n)--cycle;
			}
			\draw[color=black!25,shift=(-60:\n), shift=(0:\n)]\hexcoord{6.5*\n}{2.5*\n}
			(0:\n)--(60:\n)--(120:\n)--(180:\n)--(-120:\n)--(-60:\n)--cycle;
			\foreach \x in {-5.5,-4.5,-3.5,-2.5,-1.5,-0.5,5.5,4.5,3.5,2.5,1.5,0.5}
			\foreach \y in {2.5}{
				\draw[color=black!25]\hexcoord{\x*\n}{\y*\n}
				(0:\n)--(60:\n)--(120:\n)--(180:\n)--(-120:\n)--(-60:\n)--cycle;
				\draw[color=black!25,shift=(-60:\n), shift=(0:\n)]\hexcoord{\x*\n}{\y*\n}
				(0:\n)--(60:\n)--(120:\n)--(180:\n)--(-120:\n)--(-60:\n)--cycle;
			}
			\foreach \j in {-2,-1,0,1,2}
			\foreach \k in {-2}
			\foreach \x in {0,60,120,180,240,300}
			\foreach \y in {240,300}
			\draw[shift=(\x+120:\n), shift=(\x+60:\n), very thick, color=black]\hexcoord{\j*\n}{\k*\n}
			(\x+\y:\n)--(\x+\y+60:\n);
			\foreach \j in {-3,-2,2,3}
			\foreach \k in {-1}
			\foreach \x in {0,60,120,180,240,300}
			\foreach \y in {240,300}
			\draw[shift=(\x+120:\n), shift=(\x+60:\n), very thick, color=black]\hexcoord{\j*\n}{\k*\n}
			(\x+\y:\n)--(\x+\y+60:\n);
			\foreach \j in {-4,-3,3,4}
			\foreach \k in {0}
			\foreach \x in {0,60,120,180,240,300}
			\foreach \y in {240,300}
			\draw[shift=(\x+120:\n), shift=(\x+60:\n), very thick, color=black]\hexcoord{\j*\n}{\k*\n}
			(\x+\y:\n)--(\x+\y+60:\n);
			\foreach \j in {-2,-1,0,1,2}
			\foreach \k in {2}
			\foreach \x in {0,60,120,180,240,300}
			\foreach \y in {240,300}
			\draw[shift=(\x+120:\n), shift=(\x+60:\n), very thick, color=black]\hexcoord{\j*\n}{\k*\n}
			(\x+\y:\n)--(\x+\y+60:\n);
			\foreach \j in {-3,-2,2,3}
			\foreach \k in {1}
			\foreach \x in {0,60,120,180,240,300}
			\foreach \y in {240,300}
			\draw[shift=(\x+120:\n), shift=(\x+60:\n), very thick, color=black]\hexcoord{\j*\n}{\k*\n}
			(\x+\y:\n)--(\x+\y+60:\n);
			\foreach \j in {-3.5,-2.5,2.5,3.5}
			\foreach \k in {-.5}
			\foreach \x in {0,60,120,180,240,300}
			\foreach \y in {240,300}
			\draw[shift=(\x+120:\n), shift=(\x+60:\n), very thick, color=black]\hexcoord{\j*\n}{\k*\n}
			(\x+\y:\n)--(\x+\y+60:\n);
			\foreach \j in {-2.5,-1.5,-.5,0.5,1.5,2.5}
			\foreach \k in {-1.5}
			\foreach \x in {0,60,120,180,240,300}
			\foreach \y in {240,300}
			\draw[shift=(\x+120:\n), shift=(\x+60:\n), very thick, color=black]\hexcoord{\j*\n}{\k*\n}
			(\x+\y:\n)--(\x+\y+60:\n);
			\foreach \j in {-3.5,3.5}
			\foreach \k in {0.5}
			\foreach \x in {0,60,120,180,240,300}
			\foreach \y in {240,300}
			\draw[shift=(\x+120:\n), shift=(\x+60:\n), very thick, color=black]\hexcoord{\j*\n}{\k*\n}
			(\x+\y:\n)--(\x+\y+60:\n);
			\foreach \j in {-2.5,-1.5,-.5,0.5,1.5,2.5}
			\foreach \k in {1.5}
			\foreach \x in {0,60,120,180,240,300}
			\foreach \y in {240,300}
			\draw[shift=(\x+120:\n), shift=(\x+60:\n), very thick, color=black]\hexcoord{\j*\n}{\k*\n}
			(\x+\y:\n)--(\x+\y+60:\n);
			\foreach \j in {-3.5,-2.5,2.5,3.5}
			\foreach \k in {.5}
			\foreach \x in {0,60,120,180,240,300}
			\foreach \y in {240,300}
			\draw[shift=(\x+120:\n), shift=(\x+60:\n), very thick, color=black]\hexcoord{\j*\n}{\k*\n}
			(\x+\y:\n)--(\x+\y+60:\n);
			\foreach \x in {0,60,120,180,240,300}
			\foreach \y in {60,120}{
				\draw[shift=(\x+120:\n), shift=(\x+60:\n), fill=black]\hexcoord{0*\n}{0*\n} (\x+\y:1*\n) node[red]{$\bullet$};}
			\foreach \x in {-3,-2,-1,0,1}
			\foreach \y in {0}{\draw[shift=(120:\n), shift=(60:\n), fill=black]\hexcoord{\x*\n}{\y*\n} (0:8*\n) node[red]{$\bullet$};}
			\foreach \x in {-3,-2,-1,0,1}
			\foreach \y in {0}{\draw[shift=(120:\n), shift=(60:\n), fill=black]\hexcoord{\x*\n}{\y*\n} (180:8*\n) node[red]{$\bullet$};}
			\foreach \x in {-8,-2}
			\foreach \y in {-1}{\draw[shift=(120:\n), shift=(60:\n), fill=black]\hexcoord{\x*\n}{\y*\n} (120:8*\n) node[red]{$\bullet$};}
			\foreach \x in {-9,-1}
			\foreach \y in {0}{\draw[shift=(120:\n), shift=(60:\n), fill=black]\hexcoord{\x*\n}{\y*\n} (120:8*\n) node[red]{$\bullet$};}
			\foreach \x in {-10,0}
			\foreach \y in {1}{\draw[shift=(120:\n), shift=(60:\n), fill=black]\hexcoord{\x*\n}{\y*\n} (120:8*\n) node[red]{$\bullet$};}
			\foreach \x in {0,6}
			\foreach \y in {1}{\draw[shift=(120:\n), shift=(60:\n), fill=black]\hexcoord{\x*\n}{\y*\n} (300:8*\n) node[red]{$\bullet$};}
			\foreach \x in {-1,7}
			\foreach \y in {0}{\draw[shift=(120:\n), shift=(60:\n), fill=black]\hexcoord{\x*\n}{\y*\n} (300:8*\n) node[red]{$\bullet$};}
			\foreach \x in {-2,8}
			\foreach \y in {-1}{\draw[shift=(120:\n), shift=(60:\n), fill=black]\hexcoord{\x*\n}{\y*\n} (300:8*\n) node[red]{$\bullet$};}

			\foreach \x in {-8.5,-1.5}
			\foreach \y in {-.5}{\draw[shift=(120:\n), shift=(60:\n), fill=black]\hexcoord{\x*\n}{\y*\n} (120:8*\n) node[red]{$\bullet$};}
			\foreach \x in {-9.5,-.5}
			\foreach \y in {0.5}{\draw[shift=(120:\n), shift=(60:\n), fill=black]\hexcoord{\x*\n}{\y*\n} (120:8*\n) node[red]{$\bullet$};}
			\foreach \x in {-0.5,6.5}
			\foreach \y in {0.5}{\draw[shift=(120:\n), shift=(60:\n), fill=black]\hexcoord{\x*\n}{\y*\n} (300:8*\n) node[red]{$\bullet$};}
			\foreach \x in {-1.5,7.5}
			\foreach \y in {-0.5}{\draw[shift=(120:\n), shift=(60:\n), fill=black]\hexcoord{\x*\n}{\y*\n} (300:8*\n) node[red]{$\bullet$};}
		\end{tikzpicture}
		\caption{$\Lambda_{5,2}\subset\bbH$ with boundary vertices $\partial\Lambda_{5,2}$ in red.}\label{fig:lambdank}
	\end{subfigure}\hspace{30pt}
	\begin{subfigure}{0.4\textwidth}
		\begin{tikzpicture}[scale=0.6]
			\def\N{5}
			
			\foreach \i in {2,3} {
				\foreach \j in {2,...,10} {
					\filldraw[black] (\i, \j+1) circle (2pt);
					\filldraw[black] (\i+1, \j) circle (2pt);
					\filldraw[black] (\i-1, \j) circle (2pt);
					\filldraw[black] (\i, \j-1) circle (2pt);
					
					\ifnum\i<21
					\filldraw[black] (\i+0.5, \j) circle (2pt);
					\filldraw[black] (\i-.5, \j) circle (2pt);
					\fi
					
					\ifnum\j<21
					\filldraw[black] (\i, \j+0.5) circle (2pt);
					\filldraw[black] (\i, \j-0.5) circle (2pt);
					\fi
				}
			}
			
			\foreach \i in {2,3} {
				\foreach \j in {1,...,10} {
					\draw[thick] (\i,\j) -- (\i,\j+1);
				}
			}
			\foreach \i in {1,2,3} {
				\foreach \j in {2,...,10} {
					\draw[thick] (\i,\j) -- (\i+1,\j);
				}
			}
			\foreach \i in {9,10} {
				\foreach \j in {2,...,10} {
					\filldraw[black] (\i, \j+1) circle (2pt);
					\filldraw[black] (\i+1, \j) circle (2pt);
					\filldraw[black] (\i-1, \j) circle (2pt);
					\filldraw[black] (\i, \j-1) circle (2pt);
					
					\ifnum\i<21
					\filldraw[black] (\i+0.5, \j) circle (2pt);
					\filldraw[black] (\i-.5, \j) circle (2pt);
					\fi
					
					\ifnum\j<21
					\filldraw[black] (\i, \j+0.5) circle (2pt);
					\filldraw[black] (\i, \j-0.5) circle (2pt);
					\fi
				}
			}

			\foreach \i in {9,10} {
				\foreach \j in {1,...,10} {
					\draw[thick] (\i,\j) -- (\i,\j+1);
				}
			}
			\foreach \i in {8,9,10} {
				\foreach \j in {2,...,10} {
					\draw[thick] (\i,\j) -- (\i+1,\j);
				}
			}
			\foreach \i in {2,...,10} {
				\foreach \j in {2,3,9,10} {
					\filldraw[black] (\i, \j+1) circle (2pt);
					\filldraw[black] (\i+1, \j) circle (2pt);
					\filldraw[black] (\i-1, \j) circle (2pt);
					\filldraw[black] (\i, \j-1) circle (2pt);
					
					\ifnum\i<21
					\filldraw[black] (\i+0.5, \j) circle (2pt);
					\filldraw[black] (\i-.5, \j) circle (2pt);
					\fi
					
					\ifnum\j<21
					\filldraw[black] (\i, \j+0.5) circle (2pt);
					\filldraw[black] (\i, \j-0.5) circle (2pt);
					\fi
				}
			}
			
			\foreach \i in {2,...,10} {
				\foreach \j in {1,2,3,8,9,10} {
					\draw[thick] (\i,\j) -- (\i,\j+1);
				}
			}
			\foreach \i in {2,...,9} {
				\foreach \j in {2,3,9,10} {
					\draw[thick] (\i,\j) -- (\i+1,\j);
				}
			}
			\foreach \i in{1,11}{
				\foreach \j in {2,...,10}{\filldraw[red] (\i, \j) circle (3pt);
					\filldraw[red] (\j, \i) circle (3pt);}}
			\foreach \i in{4,8}{
				\foreach \j in {4,...,8}{\filldraw[red] (\i, \j) circle (3pt);
					\filldraw[red] (\j, \i) circle (3pt);}}
		\end{tikzpicture}
		\caption{$\Lambda^{(1)}_{5,3}$ for the Lieb lattice with boundary vertices $\partial\Lambda^{(1)}_{5,3}$ in red.}\label{fig:lambda_nk_lieb}
	\end{subfigure}
	\caption{}				
\end{figure}

For ease of reference, we finish this section by stating some basic set notation that will be frequently used throughout this work. First, for either lattice $\Gamma\in\{\bbH,\bZ^2\}$ and any subgraph $\Lambda=(\caV_\Lambda,\caE_\Lambda)\subseteq \Gamma^{(m)}$, we define the degree of a vertex $x\in\caV_\Lambda$ in the graph $\Lambda$ by
\[
\deg_\Lambda(x) = |\{y\in \caV_\Lambda : (x,y)\in \caE_\Lambda\}|.
\] 
and define the boundary $\partial \Lambda$ as the set of all vertices in $\Lambda$ that share an edge with a vertex outside of $\Lambda$, i.e.
\begin{equation}\label{boundary}
	\partial \Lambda = \left\{x\in \caV_{\Lambda}: \deg_\Lambda(x) \neq \deg_{\Gamma^{(m)}}(x)\right\}\,.
\end{equation}
We note that the sequences defined in \eqref{Lambda_N} and \eqref{Lambda_N_Square} are such that $\partial\Lambda_N^{(m)}=\partial\Lambda_N$ for all $m\geq 0$, which will be used throughout this work.

For any $\Lambda\subseteq \Gamma$ we denote by $\mathring{\Lambda}=(\caV_{\mathring{\Lambda}},\caE_{\mathring{\Lambda}})$ the interior of $\Lambda$, which is the subgraph obtained from removing all boundary vertices from $\Lambda$, i.e.
\begin{equation}\label{interior}
\caV_{\mathring{\Lambda}}=\caV_\Lambda\setminus\partial\Lambda,\qquad \caE_{\mathring{\Lambda}}=\{(x,y)\in\caE_\Lambda: x,y\in \caV_{\mathring{\Lambda}}\}\,.
\end{equation}
Given that $\Lambda\subseteq \Gamma$, the undecorated lattice, we note that the operations of taking the interior and $m$-decorating the volume do not commute. For the specific sequence of volumes defined in  \eqref{Lambda_N} or \eqref{Lambda_N_Square}, we will always want to consider the sequence of volumes defined decorating the interior $\mathring{\Lambda}_N$  with $m$-vertices additional vertices, $m\in\bN_0$, and denote these volumes by $\mathring{\Lambda}_N^{(m)}$. One can easily check that these volumes are different from, but contained in, the interior of $\Lambda_N^{(m)}$, see Figure~\ref{fig:lambda2}. We note that these volumes satisfy
\[
|\caE_{\mathring{\Lambda}_N^{(m)}}|=(m+1)|\caE_{\mathring{\Lambda}_N}|
\]
 and, recalling the notation from Corollary~\ref{cor:LTQO}, that $\gamma^{(N,m)}$ is the unique, shortest loop encompassing $\mathring{\Lambda}_N^{(m)}$.

For the proof of Theorem~\ref{thm:gs_indistinguishability} we will need to consider a subgraph of $\Lambda_N^{(m)}$ that is disjoint from $\Lambda_{K-1}^{(m)}$. In the case of $\Gamma=\bbH$, this volume is defined for any choice of parameters $1 \leq K \leq N-2$ by
\begin{equation}\label{LambdaNK}
\Lambda_{N,K} = \bigcup_{\substack{\tilde{x}\in\tilde{\Gamma}:\\K+1\leq\tilde{d}(\tilde{x},\tilde{0})\leq N-1}}h_{\tilde{x}}\subseteq \Lambda_N
\end{equation}
where $\tilde{d}$ is the graph distance in the dual lattice $\tilde{\Gamma}$, see Figure~\ref{fig:lambdank}. We additionally set $\Lambda_{N,0}=\Lambda_N$ for all $N\geq 1$. The edge and vertex sets of these volumes satisfy 
\begin{equation}\label{Lambda_NK_properties}
	\caE_{\Lambda_{N,K}}=\caE_{\Lambda_N}\setminus \caE_{\mathring{\Lambda}_K},\quad\quad \caV_{\Lambda_{N,K}}\cap\caV_{\mathring{\Lambda}_K} = \partial \mathring{\Lambda}_K.
\end{equation}
The boundary of $\Lambda_{N,K}$ consists of all vertices with degree one and partition into two sets, which we will refer to as the \emph{inner boundary}, $\partial\mathring{\Lambda}_K$, and \emph{outer boundary}, $\partial\Lambda_N$:
\begin{equation}\label{LambdaNK_bdy}
	\partial\Lambda_{N,K} = \{x\in \Lambda_{N,K}: \deg_{\Lambda_{N,K}}(x) = 1\}=\partial\Lambda_N \cup \partial\mathring{\Lambda}_K.
\end{equation}

When $\Gamma=\bbZ^2$, we define
	\begin{equation}\label{liebvols}
	\Lambda_{N,K}=\bigcup_{\substack{x\in\Gamma \\ K\leq d_\infty(x,0)\leq N-1}} b_x^1(1)
	\end{equation}
	where again the boundary can be divided into its inner and outer components,  $\partial\mathring{\Lambda}_K$ and $\partial\Lambda_N$, respectively, but this time the boundary has four vertices of degree two, see \ref{fig:lambda_nk_lieb}. 
	
	For either choice of $\Gamma$, $\Lambda_{N,K}^{(m)}\subseteq \Gamma^{(m)}$ will denote the $m$ decorated version of $\Lambda_{N,K}$.

\section{Description of the Ground States}\label{sec:AKLT_GSS}

	\subsection{Ground States and Expectations of AKLT Models}\label{sec:gs}
As in \cite{kennedy:1988}, we use the Weyl representation of the Lie algebra $\mathfrak{su}(2)$ acting on homogeneous polynomials of two variables to give a description of the ground state space of the AKLT Hamiltonian. As we will consider several values of spin, we introduce the Hilbert space representing a particle of general spin $s\in \bN/2$. Namely,
\[
\caH^{(s)} = \left\{\sum_{k=0}^{2s}\lambda_k u^kv^{2s-k}: \lambda_k\in \bC\right\}
\]
where $u=\cos(\theta/2)e^{i\phi/2}$ and $v=\sin(\theta/2)e^{-i\phi/2}$ are functions of $0\leq\phi<2\pi,$ $0\leq\theta<\pi$, and the inner product for all $\Psi,\,\Phi\in\caH^{(s)}$ is
\begin{align}
	\braket{\Psi}{\Phi}=\int d\Omega \overline{\Psi(\theta,\phi)}\Phi(\theta,\phi), \qquad d\Omega=\frac{1}{4\pi}\sin(\theta)d\phi d\theta\,.
\end{align}
Here, we note that $d\Omega$ is the measure associated with the spherical variable 
\begin{equation}\label{spherical_coord}
	\Omega=(\sin(\theta)\cos(\phi),\sin(\theta)\sin(\phi),\cos(\theta))
\end{equation}
and $\caH^{(s)}\subseteq L^2(d\Omega)$ for all $s\in \bN/2$, the latter of which is independent of $s$. The spin-$s$ matrices in $\caB(\caH^{(s)})$ given in terms of the Weyl representation are
\begin{equation}\label{weyl}
	S^3 = \frac{1}{2}\left(v\partial_{v}-u\partial_{u}\right),\qquad S^-=u\partial_{v},\qquad S^+=v\partial_{u}\,.
\end{equation}

To described the class AKLT models considered in this work in the context of the Weyl representation, let $\Gamma=(\caV_\Gamma,\caE_\Gamma)$ be any simple connected graph. The local Hilbert space for any $x\in \caV_\Gamma$ is given by $\caH_x:=\caH^{(s_x)}$ where $s_x=\deg_\Gamma(x)/2$. Let $u_x, \, v_x, \, \theta_x,\, \phi_x,$ $\Omega_x$ and $d\Omega_x$ denote the variables associated to $x$. Then, for any finite subgraph $\Lambda\subseteq \Gamma$, the inner product on $\caH_\Lambda=\bigotimes_{x\in\caV_\Lambda}\caH_x$ is 
\begin{equation}\label{hs_ip}
	\langle\Phi,\Psi\rangle=\int d\Omega^{\Lambda} \overline{\Phi(\theta,\phi)}\Psi(\theta,\phi)\qquad \forall \Phi,\Psi\in\caH_{\Lambda}
\end{equation}
where $d\Omega^{\Lambda}$ is the product measure associated with $\{d\Omega_x:x\in\caV_\Lambda\}$, and $\Phi(\theta,\phi)$ and $\Psi(\theta,\phi)$ are functions of the variables $\theta=(\theta_x)_{x\in \caV_\Lambda}$ and $\phi=(\phi_x)_{x\in\caV_\Lambda}$. As described in Section~\ref{sec:main_results}, the AKLT interaction $P_e$ any edge $e=(x,y)\in \caE_\Gamma$,  is the orthogonal projection onto the subspace of maximal total spin $s_x+s_y$. Considering \eqref{weyl}, this is the orthogonal projection onto the subspace 
\begin{equation}\label{total_spin}
	{\rm span}\left\{(u_x\partial_{v_x}+u_y\partial_{v_y})^kv_x^{2s_x}v_y^{2s_y}: 0 \leq k \leq s_x+s_y\right\}\cong \cH^{(s_x+s_y)}\,.
\end{equation}
The ground state space of the AKLT model for any finite subgraph $\Lambda \subseteq \Gamma$ is then given by the following result.

\begin{thm}[Ground State Space\cite{kennedy:1988}] \label{thm:AKLT_GSS}
	Let $\Gamma=(\caV_\Gamma,\,\caE_\Gamma)$ be any simple connected graph. For any finite subgraph $\Lambda=(\caV_\Lambda,\caE_\Lambda)\subseteq \Gamma$, the ground state space of $H_\Lambda=\sum_{e\in\caE_\Lambda}P_e$ is
	\begin{equation}\label{gss}
		\caG_\Lambda = \left\{\Psi(f) = f \cdot\prod_{(x,y)\in \caE_\Lambda}(u_xv_y-v_yu_x): f\in \caH_{\partial \Lambda}^{\rm gss}\right\}
	\end{equation}
	where
	\begin{equation}
		\caH_{\partial\Lambda}^{\mathrm{gss}}=\bigotimes_{x\in\partial\Lambda}\caH^{(d_x'/2)} , \qquad d_x' =\deg_\Gamma(x)-\deg_\Lambda(x)
	\end{equation}
In particular, this holds for any finite subgraph of the $m$-decorated square or hexagonal lattice, with $m\in \bZ_{\geq 0}$.
\end{thm}

This result follows from first proving that for any edge $e=(x,y)\in \Gamma$, 
\[
\mathrm{ker}P_{e}=\left\{(v_xu_y-u_xv_y)g\in\caH_x\otimes\caH_y:g\in \cH^{(s_x-1/2)}\otimes\caH^{(s_y-1/2)}\right\}.    
\]
Since $\caH_x=\cH^{(s_x)}$ is the set of all homogeneous polynomials of degree $s_x=\deg_\Gamma(x)/2$ in two variables, it follows that every interaction $P_e$, $e=(x,y)\in\caE_\Lambda$, is simultaneously minimized by any vector $\psi\in \caH_\Lambda$ that has a singlet $u_xv_y-v_xu_y$ across each edge. Since the polynomials over $\mathbb{C}$ form a unique factorization domain, this produces \eqref{gss}. For complete details, see \cite{kennedy:1988}.

It was shown \cite{arovas:1988} that for any finite subgraph, $\Lambda$, of a simple connected graph, $\Gamma$, and any $A\in \caA_{\Lambda}$, there exists a polynomial $A({\bf\Omega})$ in the variables $(u_x,v_x,\overline{u_x}, \overline{v_x})_{x\in \supp(A)}$, called the \emph{symbol of $A$}, such that the matrix element of $A$ for any vectors $\Psi, \Phi\in \caH_\Lambda$ is
\begin{equation}\label{matrix_elt}
	\braket{\Psi}{A\Phi}=\int d\Omega^{\Lambda}\overline{\Psi(\theta,\phi)}\Phi(\theta,\phi)A({\bf\Omega})\,.
\end{equation}
This is a consequence of using that $\{\partial_{u}^k\partial_v^lu^{k+j}v^{l-j}: k,l\in\bN_0, \, -k\leq j\leq l\}$ is a spanning set of $\caB(\caH^{(s)})$ for any $s$, and then applying the identity
\[
\braket{ \Phi} {\partial_{u}^k\partial_v^lu^{k+j}v^{l-j}\Psi}=C_{k,l}\braket{ u^kv^l\Phi}{u^{k+j}v^{l-j}\Psi}, \quad C_{k,l}=\frac{(2s+l+k+1)!}{(2s+1)!} \,.
\]
See \cite{kennedy:1988,lucia:2023} for more details. We note that the symbol $A({\bf\Omega})$ is not necessarily unique, but that our calculations will not depend on the particular choice. 

An immediate consequence of \eqref{matrix_elt} and Theorem~\ref{thm:AKLT_GSS}, is that the expectation of an observable $A\in \caA_{\Lambda}$ in any ground state $\Psi(f)\in \caG_{\Lambda}$ is given by
\begin{equation}\label{gs_exp_f}
	\langle \Psi(f),A\Psi(f)\rangle=\int d\Omega^{\Lambda}\prod_{(x,y)\in\caE_\Lambda}|u_xv_y-v_xu_y|^2|f|^2A({\bf\Omega}) = \int d\rho_{\Lambda}|f|^2 A({\bf\Omega})
\end{equation} 
where one replaces $|u_xv_y-v_xu_y|^2=\frac{1}{2}(1-\Omega_x\cdot\Omega_y)$ and defines
\begin{equation}\label{rho}
d\rho_{\Lambda}=\rho_{\Lambda}d\Omega^{\Lambda},\qquad \rho_{\Lambda}=2^{-|\caE_\Lambda|}\prod_{(x,y)\in\caE_\Lambda}(1-\Omega_x\cdot\Omega_y)\,.
\end{equation}

Theorem~\ref{thm:AKLT_GSS} establishes that any two AKLT ground states on a finite volume $\Lambda$ differ only by the choice of boundary polynomial. In the case that there is a unique limiting infinite volume ground state one understands intuitively that, for any observable $A$ supported sufficiently far from the boundary of $\Lambda$, the expectation in any finite volume state on $\Lambda$ should be roughly that of the infinite volume state. Theorem~\ref{thm:gs_indistinguishability} makes explicit upper bounds on how different these expectations can be from one another, and is a consequence of analyzing two different maps, called the \emph{bulk-boundary map} and \emph{bulk state}.
\subsection{The Bulk-Boundary Map and Bulk State for an Increasing and Absorbing Sequence of Volumes}  The goal of this section is to introduce these maps in a relatively general setting, which will in particular, be applicable both the hexagonal and square lattices as well as any decorated version of these lattices. 

For the remainder of this section, we assume that $\Gamma=(\caV_\Gamma,\caE_\Gamma)$ is any infinite, simple connected graph, and that $\Lambda=(\caV_\Lambda,\caE_\Lambda)$ and $\Lambda'=(\caV_{\Lambda'},\caE_{\Lambda'})$ are two finite subgraphs of $\Gamma$ for which
\begin{enumerate}
	\item[(i)] $\Lambda'\subseteq \mathring{\Lambda}$.
	\item[(ii)] $|\partial\Lambda\cap \{x,y\}|\leq 1$  for any edge $(x,y)\in \caE_\Lambda$. 
\end{enumerate}
Above, $\mathring{\Lambda}$ denotes the interior of $\Lambda$. The first condition guarantees that the boundaries of $\Lambda'$ and $\Lambda$ are disjoint, and the second condition guarantees that only one vertex across any edge can belong to the boundary, $\partial\Lambda$. We note that, in particular, these properties are satisfied if one takes $\Lambda=\Lambda_N^{(m)}$ and $\Lambda'=\mathring{\Lambda}_K^{(m)}$ for any $K<N$, where these are the $m$-decorated sequences defined in \eqref{Lambda_N}-\eqref{Lambda_N_Square}, see also \eqref{interior} and the text that follows.

Lemma~\ref{lem:bulk_bdy_map} below, which is a generalization of \cite[Lemma 4.2]{lucia:2023}, shows that for any $\Lambda$ and $\Lambda'$ as above, the ground state expectation for any $A\in \cA_{\Lambda'}$ in any state $\Psi(f)\in \cG_{\Lambda}$  can be written in terms of $|f|^2$ and \emph{the bulk-boundary map} $\omega_\Lambda$ defined by
\begin{align}
	\omega_\Lambda(A;\Omega^{\partial\Lambda}) & :=Z_\Lambda(A;\Omega^{\partial\Lambda})/Z_\Lambda(\Omega^{\partial\Lambda})\label{bb_map}\\  Z_\Lambda(A;\Omega^{\partial\Lambda})&:=\int d\Omega^{\mathring{\Lambda}}\rho_{\Lambda}A(\Omega),
	\quad Z_\Lambda(\Omega^{\partial\Lambda}) = Z_\Lambda(\idty;\Omega^{\partial\Lambda}) \nonumber
\end{align}
We note that condition (i) guarantees that $\omega_\Lambda(A,;\Omega^{\partial\Lambda})$ is a polynomial in the boundary variables $\Omega^{\partial\Lambda}=(\Omega_x: x\in\partial\Lambda)$. In particular, this is a continuous function of these variables.
\begin{lemma}\label{lem:bulk_bdy_map} Suppose that $\Gamma$ is an infinite, simple connected graph and $\Lambda,\Lambda'\subseteq \Gamma$ are any two finite volumes satisfying properties (i)-(ii) above. For any $A\in \caA_{\Lambda'}$ and $\Psi(f)\in\caG_{\Lambda}$, one has
	\begin{equation}
		\langle \Psi(f),A\Psi(f)\rangle=\int d\rho_{\Lambda}|f|^2\omega_\Lambda(A;\Omega^{\partial\Lambda})\,.
	\end{equation}
\end{lemma}
\begin{proof} 
This is a consequence of recognizing that since $\Lambda'\subseteq \mathring{\Lambda}$,
\[\langle \Psi(f),A\Psi(f)\rangle
=\int d\Omega^{\Lambda}|f|^2\rho_{\Lambda}A(\Omega)
=\int d\Omega^{\partial\Lambda}|f|^2\int d\Omega^{\mathring{\Lambda}}\rho_{\Lambda}\frac{\int d\Omega^{\mathring{\Lambda}}\rho_{\Lambda}A(\Omega)}{\int d\Omega^{\mathring{\Lambda}}\rho_{\Lambda}}\]
and verifying that $Z_\Lambda(\Omega^{\partial\Lambda})=\int d\Omega^{\mathring{\Lambda}}\rho_{\Lambda}>0$ for all values of the boundary variables $(\Omega_x:x\in\partial\Lambda)$. To prove the latter, $x\in \partial\mathring{\Lambda}$ if and only if there exists $y\in\partial\Lambda$ such that $(x,y)\in\caE_{\Lambda}$. As a consequence, for any fixed choice of the boundary variables $\Omega^{\partial\Lambda}=(\Omega_x: x\in\partial)$,
\[
Z_N(\Omega^{\partial\Lambda})= \|\Psi(g_{\Omega^{\partial\Lambda}})\|^2\quad \text{where}\quad g_{\Omega^{\partial\Lambda}}=\prod_{\substack{(x,y)\in\caE_{\Lambda}:\\ x\in\partial\Lambda}}(u_xv_y-u_yv_x)\in\cH_{\partial\mathring{\Lambda}}^{\mathrm{gss}}\,.
\]
By condition (ii), $g_{\Omega^{\partial\Lambda}}$ is a nonzero polynomial for any choice of $\Omega^{\partial\Lambda}$. Therefore, by Theorem~\ref{thm:AKLT_GSS}, $\Psi(g_{\Omega^{\partial\Lambda_N}})\neq 0$ as desired.
\end{proof}

The definition for the \emph{bulk state} $\omega_\Lambda^{\rm bulk}:\caA_{\Lambda}\to \bC$ comes from averaging $Z_\Lambda(A;\Omega^{\partial\Lambda})$ over the boundary variables. Since all ground states $\Psi(f)\in \caG_{\Lambda}$ only differ by the choice of boundary polynomial, in the case that one expects there to be a unique infinite volume ground state, $\omega_\Lambda^{\rm bulk}(A)$ should produce a good approximation of $\braket{\Psi(f)}{A\Psi(f)}$ for any normalized $\Psi(f)\in \caG_{\Lambda}$ and observable $A\in \caA_{\Lambda'}$ when $d(\Lambda',\partial\Lambda)>>1$. Concretely, for any $A\in\caA_{\Lambda}$,
\begin{align}
	\omega_\Lambda^{\rm bulk}(A) &:= Z_\Lambda^{\rm bulk}(A)/Z_\Lambda^{\rm bulk} \label{bulk_state} \\
	Z^{\mathrm{bulk}}_\Lambda(A)&:=\int d\rho_{\Lambda}A(\Omega)=\int d\Omega^{\partial\Lambda}Z_\Lambda(A;\Omega^{\partial\Lambda}), \quad Z^{\mathrm{bulk}}_\Lambda=Z^{\mathrm{bulk}}_\Lambda(\mathbbm{1}) \nonumber\,.
\end{align}

While it is not a priori clear whether $\omega_\Lambda^{\rm bulk}$ is a state on $\caA_{\Lambda}$, it is a ground state of $H_{\Lambda'}$. This can be seen from applying Theorem~\ref{thm:AKLT_GSS} where one replaces the infinite graph $\Gamma$ with $\tilde{\Gamma}=\Lambda$. In this case $\omega_\Lambda^{\rm bulk}$ is the unique ground state on $\tilde{\Gamma}$. By frustration-freeness, this is also a ground state of $H_{\Lambda'}$, and by property (i), this Hamiltonian is unchanged from replacing $\Gamma$ with $\tilde{\Gamma}$. 

In the case that $A\in \caA_{\Lambda'}$ with $\Lambda'\subset \Lambda$ satisfying properties (i)-(ii) we also define $\Lambda\setminus\Lambda'\subseteq \Lambda$ to be the minimal subgraph whose edge set is $\caE_\Lambda\setminus\caE_{\Lambda'}$. This is motivated by partitioning the product in \eqref{rho} based on which edges are contained in $\Lambda'$. Since $A(\Omega)$ only depends on the variables $x\in \Lambda'$, we can write the integrals in \eqref{bb_map} and \eqref{bulk_state} as
\[
	Z_\Lambda(A;\Omega^{\partial\Lambda}) = \int d\rho_{\Lambda'} A({\bf \Omega})\Phi_{\Lambda,\Lambda'}({\bf \Omega}), \qquad Z_\Lambda^{{\rm bulk}}(A) = \int d\rho_{\Lambda'} A({\bf \Omega})\Phi_{\Lambda,\Lambda'}^{{\rm bulk}}({\bf \Omega})
\]
where
\begin{align}
	\Phi_{\Lambda,\Lambda'}({\bf \Omega}) & = 2^{-|\caE_{\Lambda\setminus\Lambda'}|}\int d\Omega^{(\Lambda\setminus\Lambda')\setminus\partial(\Lambda\setminus\Lambda')}\prod_{(x,y)\in \caE_{\Lambda\setminus\Lambda'}}(1-\Omega_x\cdot\Omega_y) \label{Phi_bbmap_general}\\
	\Phi_{\Lambda,\Lambda'}^{{\rm bulk}}({\bf \Omega}) & = 2^{-|\caE_{\Lambda\setminus\Lambda'}|}\int d\Omega^{(\Lambda\setminus\Lambda')\setminus\partial\Lambda'}\prod_{(x,y)\in \caE_{\Lambda\setminus\Lambda'}}(1-\Omega_x\cdot\Omega_y)\,. \label{Phi_bulk_general}
\end{align}
Here, we note that property (i) guarantees that $\partial(\Lambda\setminus\Lambda')=\partial\Lambda\cup\partial\Lambda'$. Above, the integral in \eqref{Phi_bbmap_general} is taken with respect to the product measure $\Omega_x$ associated to the interior sites $x$ of $\Lambda\setminus\Lambda'$, and \eqref{Phi_bulk_general} is taken with respect to the product measure associated to the sites of $\Lambda\setminus\Lambda'$ that are not boundary sites of $\Lambda'$. 

\subsection{Outline for Proof of Theorem~\ref{thm:gs_indistinguishability}}
We are now in a position to outline the proof of Theorem~\ref{thm:gs_indistinguishability}. To this end, given the $m$-decorated sequence $\Lambda_N^{(m)}$ as in \eqref{Lambda_N}, in the case that $\Gamma=\bbH$, or \eqref{Lambda_N_Square}, in the case that $\Gamma=\bZ^2$, we define for any $A\in \caA_{\Lambda_N^{(m)}}$,
\begin{equation}\label{mdec_bbmap_bulk}
\omega_N^{(m)}(A;\Omega^{\partial\Lambda_N}) = \omega_{\Lambda_N^{(m)}} (A;\Omega^{\partial\Lambda_N}), \qquad \omega_N^{\mathrm{bulk},(m)}(A) = \omega_{\Lambda_N^{(m)}}^{\mathrm{bulk}} (A)
\end{equation}
and analogously label $Z_N^{(m)}(A;\Omega^{\partial\Lambda_N})$, $Z_N^{(m)}(\Omega^{\partial\Lambda_N})$, $Z^{\mathrm{bulk},(m)}_N(A)$ and $Z^{\mathrm{bulk},(m)}_N$ the pieces defining $\omega_N^{(m)}$ and $\omega_N^{\mathrm{bulk},(m)}$, respectively. Given the previous observations, any weak-$*$ limit of $\omega_N^{{\rm bulk},(m)}$ is a frustration-free ground state of the AKLT model on $\Gamma^{(m)}$. If $\omega_N^{{\rm bulk},(m)}(A)$ is Cauchy for all $A\in \caA_{\Gamma}^{\rm loc}$, then there is only one such weak-$*$ limit and
\[
\lim_{N\to\infty}\omega_N^{{\rm bulk},(m)}(A) = \omega^{(m)}(A)\qquad \forall\, A\in \caA_{\Gamma^{(m)}}^{\rm loc}\,.
\]
We will prove that this is the case for the specific AKLT models covered in Theorem~\ref{thm:gs_indistinguishability}, and it is with respect to this infinite volume ground state that Theorem~\ref{thm:gs_indistinguishability} is stated. It follows immediately from the bound in \eqref{eq:indistinguishability} that this must be the only weak-* limit of any sequence of IAS finite volume ground states, and hence the unique infinite-volume, frustration-free ground state.

Then the strategy for proving Theorem~\ref{thm:gs_indistinguishability} is as follows. For any normalized $\Psi(f)\in \caG_{\Lambda_N^{(m)}}$ and fixed observable $A\in \caA_{\mathring{\Lambda}_K^{(m)}}$, we apply Lemma~\ref{lem:bulk_bdy_map} and bound
\begin{align}
	\left|\braket{\Psi(f)}{A\Psi(f)} - \omega^{(m)}(A)\right| \leq & \left|\braket{\Psi(f)}{A\Psi(f)} - \omega_N^{{\rm bulk},(m)}(A)\right| + \left|\omega_N^{{\rm bulk},(m)}(A) - \omega^{(m)}(A)\right| \nonumber\\
	= & \left|\int d\rho_{\Lambda_N^{(m)}}|f|^2[\omega_N^{(m)}(A;\Omega^{\partial\Lambda_N})-\omega^{\mathrm{bulk},(m)}_N(A)]\right| \nonumber\\
	& \;+  \left|\omega_N^{{\rm bulk},(m)}(A) - \omega^{(m)}(A)\right| \nonumber\\
	 \leq & \sup_{\Omega^{\partial\Lambda_N}}\left|\omega_N^{(m)}(A;\Omega^{\partial\Lambda_N})-\omega^{\mathrm{bulk},(m)}_N(A)\right|  \label{two_bounds} \\
	 & \;+  \lim_{M\to\infty}\left|\omega_N^{{\rm bulk},(m)}(A) - \omega_M^{{\rm bulk},(m)}(A)\right| \nonumber
\end{align}
where we used that $d\rho_{\Lambda_N^{(m)}}|f|^2$ is a probability measure and the supremum is taken over all possible values of the boundary variables. Thus, to prove Theorem~\ref{thm:gs_indistinguishability}, one needs to determine appropriate upper bounds on the two quantities from \eqref{two_bounds}. 

To this end, since the symbol $A({\bf \Omega})$ in \eqref{bb_map} and \eqref{bulk_state} only depend on the variables $x\in \caV_{\mathring{\Lambda}_K^{(m)}}$, one can write for any $A\in \caA_{\mathring{\Lambda}_K^{(m)}}$, $K<N$,
\begin{equation}\label{Z_to_Phi}
	Z_N^{(m)}(A;\Omega^{\partial\Lambda_N}) = \int d\rho_{\mathring{\Lambda}_K^{(m)}} A({\bf \Omega})\Phi_{N,K}^{(m)}({\bf \Omega}),\qquad
	Z_N^{{\rm bulk},(m)}(A) = \int d\rho_{\mathring{\Lambda}_K^{(m)}} A({\bf \Omega})\Phi_{N,K}^{{\rm bulk},(m)}({\bf \Omega})
\end{equation}
where, given \eqref{Phi_bbmap_general}-\eqref{Phi_bulk_general}, 
\[\Phi_{N,K}^{(m)}({\bf \Omega}) = \Phi_{\Lambda_N^{(m)},\mathring{\Lambda}_K^{(m)}}({\bf \Omega})\quad \text{and} \quad\Phi_{N,K}^{{\rm bulk},(m)}({\bf \Omega}) = \Phi_{\Lambda_N^{(m)},\mathring{\Lambda}_K^{(m)}}^{\rm bulk}({\bf \Omega}).\] 
Recalling that $\Lambda_{N,K}^{(m)}$ is the $m$-decorated version of $\Lambda_{N,K}$ defined as in \eqref{LambdaNK} if $\Gamma=\bbH$ and \eqref{liebvols} if $\Gamma=\bbZ^2$. Moreover, since in either case, $\partial\Lambda_{N,K} = \partial\Lambda_N\cup\partial\mathring{\Lambda}_K$, and this is also the boundary of the $m$-decorated volumes $\Lambda_{N,K}^{(m)}$, these can be written as
\begin{align}
	\Phi_{N,K}^{(m)}({\bf \Omega}) & = 2^{-|\caE_{\Lambda_{N,K}^{(m)}}|}\int d\Omega^{\Lambda_{N,K}^{(m)}\setminus\partial\Lambda_{N,K}}\prod_{(x,y)\in \caE_{\Lambda_{N,K}^{(m)}}}(1-\Omega_x\cdot\Omega_y) \label{Phi_bbmap}\\
	\Phi_{N,K}^{{\rm bulk},(m)}({\bf \Omega}) & = 2^{-|\caE_{\Lambda_{N,K}^{(m)}}|}\int d\Omega^{\Lambda_{N,K}^{(m)}\setminus\partial\mathring{\Lambda}_K}\prod_{(x,y)\in \caE_{\Lambda_{N,K}^{(m)}}}(1-\Omega_x\cdot\Omega_y)\,. \label{Phi_bulk}
\end{align}
The desired bounds on \eqref{two_bounds} will be a consequence of proving that $\Phi_{N,K}^{(m)}$ and $\Phi_{N,K}^{{\rm bulk},(m)}$ can be written in terms of a polymer representation that has a converge cluster expansion. 

We make two comments about this approach. First, any sequence of ground states could have been used in place of $\omega_N^{{\rm bulk},(m)}$ in \eqref{two_bounds}. The benefit of using this particular choice comes from the set of polymers associated with this state. In the case of the hexagonal lattice, this will be introduced in the next section and commented on further there, and for the square lattice this will be discussed in Section~\ref{sec:lieb}. Second, when applying the cluster expansion in Section~\ref{sec:proof}, it is computationally convenient that the bulk state and bulk-boundary map are defined on the same volume in the first term of \eqref{two_bounds}. This is why we insert and remove $\omega_N^{{\rm bulk},(m)}(A)$ and apply the triangle inequality rather than simply insert $\omega^{(m)}(A)=\lim_M\omega_M^{{\rm bulk},(m)}(A)$.

We conclude this section by proving a result that produces bounds on \eqref{two_bounds} in terms of the operator norm of the observable and provides further clarity as to how the cluster expansion arises.

	\subsection{Obtaining Bounds in Terms of $\|A\|$}
One of the main motivations for proving Theorem~\ref{thm:gs_indistinguishability} is that it implies that the ground states of various AKLT models satisfy LTQO. In the proof of gap stability from \cite{nachtergaele:2022}, it is important that the LTQO bound makes explicit its dependence on (i) the operator norm of the observable and (ii) the support of the observable. Given \eqref{bb_map} and \eqref{bulk_state}, it is clear that \eqref{two_bounds} can be easily bounded in terms of $\|A({\bf \Omega})\|_{L^\infty}$ and that, by \eqref{matrix_elt},
\[
\|A\| \leq \|A({\bf \Omega})\|_{L^\infty}\,.
\]

However, these norms are not equivalent on $\caA_{\Gamma}^{\rm loc}$ the infinite sets $\Gamma$ of interest. As an example, there is a canonical choice of the symbol $A({\bf \Omega})$ so that for any choice of distinct vertices in the hexagonal lattice, $x_i\in \bbH$, $i\in\bN$, the operator $A_n =\bigotimes_{i=1}^n(\partial_{u_{x_i}}\cdot u_{x_i}/4)\in \caA_{\Gamma}^{\rm loc}$ satisfies
\[
\|A_n\| =1, \quad \|A_n({\bf \Omega})\|_{L^\infty} = \left(\frac{5}{4}\right)^n \quad \forall\, n\,.
\]
For details on how to choose the symbol, see \cite[Equation 4.14]{lucia:2023} and the surrounding text. 

The next result, first proved in \cite{lucia:2023}, produces a bound on the two terms from \eqref{two_bounds} that depends on $\|A\|$ and positions us to use cluster expansion techniques to prove how quickly these two quantities converge and how this convergence depends on $\supp(A)$. To simplify notation, we suppress the dependence of the functions $\Phi_{N,K}^{(m)}$, $\Phi_{N,K}^{{\rm bulk},{(m)}}$, $Z_N^{(m)}$ and $Z_N^{{\rm bulk},(m)}$ on ${\bf \Omega}$ in the statement and proof.

\begin{lemma}\label{lem:op_norm}
	Let $\Gamma\in\{\bbH,\bZ^2\}$ and $m\geq 0$. Given $N>K\geq 1$, if $\Phi^{{\rm bulk},(m)}_{N,K}>0$ for any value of $(\Omega_x: x\in \partial\mathring{\Lambda}_K)$, then for any $M\geq N$ and $A\in\caA_{\Lambda_{K-1}^{(m)}}$, 
	\begin{align}
		\left|\omega^{\mathrm{bulk},(m)}_N(A)-\omega^{\mathrm{bulk},(m)}_M(A)\right|& \leq ||A|| H\left(\frac{\Phi^{\mathrm{bulk},(m)}_{N,K}}{Z^{\mathrm{bulk},(m)}_N}, \frac{\Phi^{\mathrm{bulk},(m)}_{M,K}}{Z^{\mathrm{bulk},(m)}_M}\right) \label{bulk_cauchy_bound} \\
		\sup_{\Omega^{\partial\Lambda_N}}\left|\omega_N^{(m)}(A;\Omega^{\partial\Lambda_N})-\omega^{\mathrm{bulk},(m)}_N(A)\right| & \leq ||A||\sup_{\Omega^{\partial\Lambda_N}}H\left(\frac{\Phi^{\mathrm{bulk},(m)}_{N,K}}{Z^{\mathrm{bulk},(m)}_N}, \frac{\Phi_{N,K}^{(m)}}{Z_N^{(m)}}\right) \label{bb_bulk_bound}
	\end{align}
	where $H(f,g)=D_{\infty}(f||g)e^{D_{\infty}(f||g)}$, $f,g\in L^1(d\rho_{\mathring{\Lambda}_K})$, is defined in terms of the classical $\infty$-Renyi divergence
	\begin{equation}\label{Renyi_divergence}
		D_{\infty}(f||g):=\left\Vert\log\left(f\right)-\log\left(g\right)\right\Vert_{L^{\infty}(d\rho_{\mathring{\Lambda}_K})}\,.
	\end{equation}
\end{lemma}
To be explicit, the Renyi-divergences are given by
\begin{align}
	D_{\infty}\left(\frac{\Phi^{\mathrm{bulk},(m)}_{N,K}}{Z^{\mathrm{bulk},(m)}_N}\left\Vert\frac{\Phi^{\mathrm{bulk},(m)}_{M,K}}{Z^{\mathrm{bulk},(m)}_M}\right.\right)&=\sup_{\Omega_x:x\in\partial\mathring{\Lambda}_K}\left|\log\left(\frac{\Phi^{\mathrm{bulk},(m)}_{N,K}}{Z^{\mathrm{bulk},(m)}_N}\right)-\log\left(\frac{\Phi^{\mathrm{bulk},(m)}_{M,K}}{Z^{\mathrm{bulk},(m)}_M}\right)\right| \label{divergence_bulk}\\
	\sup_{\Omega^{\partial\Lambda_N}}D_{\infty}\left(\frac{\Phi^{\mathrm{bulk},(m)}_{N,K}}{Z^{\mathrm{bulk},(m)}_N}\left\Vert\frac{\Phi_{N,K}^{(m)}}{Z_N}\right.\right)& =\sup_{\Omega_x:x\in\partial\Lambda_{N,K}}\left|\log\left(\frac{\Phi^{\mathrm{bulk},(m)}_{N,K}}{Z^{\mathrm{bulk},(m)}_N}\right)-\log\left(\frac{\Phi_{N,K}^{(m)}}{Z_N^{(m)}}\right)\right|\label{divergence_boundary}
\end{align} 

\begin{proof} Consider the embedding of the Hilbert space
	\[
	\caH_{\mathring{\Lambda}_K^{(m)}}\subseteq \caK_{\mathring{\Lambda}_K^{(m)}} :=  \bigotimes_{x\in \mathring{\Lambda}_K^{(m)}\setminus\Lambda_{K-1}^{(m)}} L^2(S^2,d\Omega_x) \otimes \caH_{\Lambda_{K-1}^{(m)}}\,.
	\]
	which follows since $\caH_x \subseteq L^2(S^2,d\Omega_x) $ for any vertex $x$. The advantage of this larger space is that one can take square roots of certain states without leaving the Hilbert space and, moreover, the operator norm is unchanged for any $A\in\caA_{\Lambda_{K-1}^{(m)}}$, i.e.
	\[	\|A\|_{\cH_{\Lambda_{K-1}^{(m)}}} = \|A\|_{\caK_{\mathring{\Lambda}_K^{(m)}}}
	\]
	where the subscript denotes the Hilbert space in which the norm is taken. 
	
	Consider first the left-hand-side of \eqref{bulk_cauchy_bound}. To simplify notation, let $f_N = \Phi^{\mathrm{bulk},(m)}_{N,K}/Z^{\mathrm{bulk},(m)}_N$, and note that $Z^{\mathrm{bulk},(m)}_N\neq 0$ as $\omega^{\mathrm{bulk},(m)}_N$ is a ground state on $\caA_{\mathring{\Lambda}_K^{(m)}}$. From \eqref{Z_to_Phi} one sees that $f_N$ is a bounded, real-valued continuous function of the angles $(\theta_x,\phi_x:x\in\partial\mathring{\Lambda}_K)$. Thus, $\psi_1,\psi_2\in \caK_{\mathring{\Lambda}_K^{(m)}}$ where
	\[
	\psi_1=\mathrm{sign}(\sigma)|\sigma|^{1/2}\xi, \quad \psi_2=|\sigma|^{1/2}\xi, \quad \sigma=f_N-f_M, \quad  \xi=\prod_{(x,y)\in\mathring{\Lambda}_K}(u_xv_y-u_yv_x)\,.
	\]
	As $\rho_{\mathring{\Lambda}_K^{(m)}}=|\xi|^2$ and $||\psi_1||_{\caK_{\mathring{\Lambda}_K^{(m)}}}=||\psi_2||_{\caK_{\mathring{\Lambda}_K^{(m)}}}$, by \eqref{bulk_state} one finds
	\begin{align*}
		\left|\omega^{\mathrm{bulk},(m)}_N(A)-\omega^{\mathrm{bulk},(m)}_M(A)\right|&=\left|\langle\psi_1,(\mathbbm{1}\otimes A)\psi_2\rangle_{\caK_{\mathring{\Lambda}_K^{(m)}}}\right|\\&\leq ||A||_{\cH_{\mathring{\Lambda}_{K-1}^{(m)}}} ||\psi_2||_{\caK_{\mathring{\Lambda}_K^{(m)}}}^2 =||A||_{\cH_{\Lambda_{K-1}^{(m)}}} \left\|f_N-f_M\right\|_{L^1(d\rho_{\mathring{\Lambda}_K^{(m)}})}
	\end{align*}
	where one recalls $d\rho_{\mathring{\Lambda}_K^{(m)}}= d\Omega^{\mathring{\Lambda}_K^{(m)}}\rho_{\mathring{\Lambda}_K^{(m)}}$. Moreover, since  $\Phi^{\mathrm{bulk},(m)}_{N,K}>0$, by \eqref{bulk_state} and \eqref{Z_to_Phi}
	\[
	\|f_N\|_{L^1(d\rho_{\mathring{\Lambda}_K^{(m)}})}=\frac{1}{Z_N^{{\rm bulk},(m)}}\int d\rho_{\mathring{\Lambda}_K^{(m)}}\Phi^{\mathrm{bulk},(m)}_{N,K} = 1\,.
	\]
	The bound in \eqref{bulk_cauchy_bound} then follows from applying H\"older's inequality and $|e^x-1|\leq |x|e^{|x|}$ to obtain
	\[||f_N-f_M||_{L^1(d\rho_{\mathring{\Lambda}_K^{(m)}})}\leq\left\Vert\frac{f_N}{f_M}-1\right\Vert_{L^{\infty}(d\rho_{\mathring{\Lambda}_K})}\leq D_{\infty}(f_N||f_M)e^{D_{\infty}(f_N||f_M)}.\]

	For \eqref{bb_bulk_bound}, recall that $Z_N^{(m)}(\Omega^{\partial\Lambda_N})\neq 0$ for any choice of boundary variables $(\Omega_x:x\in\partial\Lambda_N)$, see the proof of Theorem~\ref{lem:bulk_bdy_map}. Thus, for any such choice, \eqref{bb_map}, \eqref{bulk_state} give
	\[\left|\omega_N^{(m)}(A;\Omega^{\partial\Lambda_N})-\omega^{\mathrm{bulk},(m)}_N(A)\right|=\left|\int d\Omega^{\mathring{\Lambda}_K^{(m)}}\rho_{\mathring{\Lambda}_K^{(m)}}\left[\frac{\Phi^{\mathrm{bulk},(m)}_{N,K}}{Z^{\mathrm{bulk},(m)}_N}-\frac{\Phi_{N,K}^{(m)}}{Z_N^{(m)}}\right]A({\bf\Omega})\right|\,.\]
	The proof then runs analogously to that of \eqref{bulk_cauchy_bound} after recognizing that $\Phi_{N,K}^{(m)}$ is also a bounded, real-valued continuous function of the angles $(\theta_x,\phi_x:x\in\partial\mathring{\Lambda}_K)$.
\end{proof}

The key observation for this proof is that since $\Lambda_{K-1}^{(m)}$ does not contain any of the boundary vertices of $\mathring{\Lambda}_K$, then one can enlarge the Hilbert space on these boundary sites without changing the operator norm, and take square roots of functions without leaving the larger Hilbert space. We note that while Lemma~\ref{lem:op_norm} is stated for the specific case of interest, namely, the maps in \eqref{mdec_bbmap_bulk}, it and the proof can be generalized to other contexts using the definitions in \eqref{Phi_bbmap_general}-\eqref{Phi_bulk_general}. Namely, given three volumes $\Lambda'\subseteq \Lambda_1\subseteq \Lambda_2$,  such that properties (i) and (ii) hold for both pairs $(\Lambda',\Lambda_1)$ and $(\Lambda',\Lambda_2)$, and an observable $A\in \caA_{\mathring{\Lambda}'}$ supported on the interior of $\Lambda'$, the same proof follows from replacing $\Lambda_N^{(m)}$ with $\Lambda_1$, $\Lambda_M^{(m)}$ with $\Lambda_2$, $\mathring{\Lambda}_K^{(m)}$ with $\Lambda'$, and $\Lambda_{K-1}^{(m)}$ with $\mathring{\Lambda}'$.

	\section{The Polymer Representation for the Hexagonal Lattice Model}\label{sec:hex_polymers}

	We now turn to discussing the polymer representations for the (un)decorated hexagonal lattice models, which were first given  in \cite{lucia:2023} and follow from a slight modification from the argument in \cite{kennedy:1988}. We first give the polymers and weight function for the undecorated model, and then describe how to modify these to include decorated versions of the lattice.
	
	\subsection{The Polymer Sets and Weight Functions}
	
Let $\Gamma=\bbH$ be the hexagonal lattice. The polymer representations arise from using \eqref{spherical_coord} to integrate evaluate \eqref{Phi_bbmap}-\eqref{Phi_bulk} in the case that $\Lambda_N$ is the sequence of volumes from \eqref{Lambda_N}. In this case, the minimal graph $\Lambda_{N,K}$ from the definition of $\Phi_{N,K}^{\rm bulk}$ and $\Phi_{N,K}$ is the volume defined in \eqref{LambdaNK}, see \eqref{Lambda_NK_properties}.

The desired polymer sets for the representations of \eqref{Phi_bbmap}-\eqref{Phi_bulk} will consist of a collection of finite loops and self-avoiding walks in the hexagonal lattice. That is, they will be subsets of
\begin{equation}\label{polymersP}
\caP_\Gamma = \left\{\gamma=(\caV_\gamma,\caE_\gamma)\subseteq \Gamma: |\caV_\gamma|<\infty\;\wedge\;\gamma \;\text{connected}\;\wedge\; \deg_\gamma(x)\leq 2 \; \forall x\in \caV_\gamma\right\}\,.
\end{equation}
The global sets of loops and self-avoiding walks are then defined, respectively, by
\[\caL_\Gamma = \{\gamma\in \caP_\Gamma: \deg_\gamma(x) = 2\; \forall\, x\in \caV_\gamma\}, \quad \caW_\Gamma = \caP_\Gamma \setminus \caL_\Gamma\,.\]
As each $\gamma\in \caW_\Gamma$ is connected, it has precisely two \emph{endpoints}, ${\rm ep}(\gamma):=\{x,y\}\in \caV_\gamma$, which are the unique pair of vertices such that $\deg_\gamma(x) =\deg_\gamma(y)=1$.

For any $0 \leq K \leq N-2$ let
\begin{align*}
\caL_{N,K} & = \{\gamma\in \caL_\Gamma: \gamma \subseteq \Lambda_{N,K}\}\\
\caW_{N,K} & =\{\gamma\in \caW_\Gamma: \gamma \subseteq \Lambda_{N,K}\wedge {\rm ep}(\gamma)\subseteq \partial\Lambda_{N,K}\}\\
\caW_{N,K}^{\rm bulk} & =\{\gamma\in \caW_\Gamma: \gamma \subseteq \Lambda_{N,K}\wedge {\rm ep}(\gamma)\subseteq \partial\mathring{\Lambda}_{K}\}
\end{align*}
where we use the convention that $\partial\mathring{\Lambda}_K = \emptyset$ if $K=0$, implying that
\[
\caW_{N,0}^{\rm bulk}=\emptyset\quad \text{and}\quad\partial\Lambda_{N,0} = \partial\Lambda_N  \,. 
\] 
Then, given any $A\in\caA_{\mathring{\Lambda}_K}$, the polymer sets for $\Phi_{N,K}$ and $\Phi_{N,K}^{\rm bulk}$ are, respectively,
\begin{equation}\label{polymer_sets}
	\caP_{N,K} = \caL_{N,K}\cup \caW_{N,K}, \quad \caP_{N,K}^{\rm bulk} = \caL_{N,K}\cup \caW_{N,K}^{\rm bulk}\,.
\end{equation}
The main benefit of using the bulk state in \eqref{two_bounds} is that the endpoints of any walk $\gamma\in \cW_{N,K}^{\rm bulk}$ only lie on the \emph{inner boundary}, $\partial \mathring{\Lambda}_K$, of $\Lambda_{N,K}$ and do not involve vertices on the \emph{outer boundary}, $\partial \Lambda_N$, see Figure \ref{fig:walks_and_loops}.

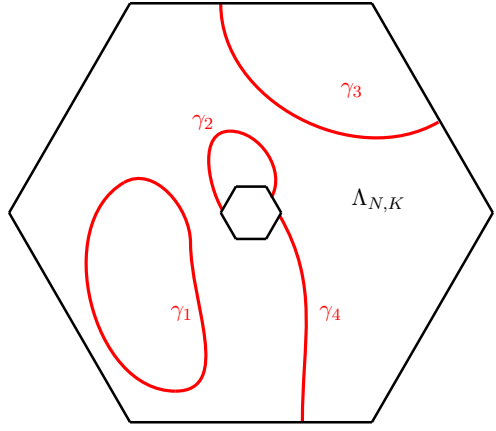
\begin{figure}
	\scalebox{.8}{\begin{tikzpicture}
		\def\n{.5};
		\newcommand{\hexcoord}[2]
		{[shift=(60:#1),shift=(120:#1),shift=(0:#2),shift=(-60:#2),shift=(0:#2),shift=(60:#2)]}

		\draw\hexcoord{0}{0} (60:\n)--(120:\n) node[midway, below, yshift=05*\n pt,xshift=120*\n pt] {$\Lambda_{N,K}$};
		\draw [line width=1.5pt, color=red] (-.95*\n,.1*\n)  to[out=120,in=180] coordinate[pos=.35] (A) (-.7*\n,2.7*\n) ;
		\draw [line width=1.5pt, color=red] (.67*\n,.5*\n)  to[out=60,in=0] coordinate[pos=.35] (A) (-.7*\n,2.7*\n) ;
		\draw [line width=1.5pt, color=red] (.95*\n,-.1*\n)  to[out=-60,in=90] coordinate[pos=.35] (A) (1.7*\n,-6.9*\n) ;

		\draw [line width=1.5pt, color=red] (6.2*\n,3*\n)  to[out=-150,in=-90] coordinate[pos=.35] (A) (-1*\n,6.9*\n) ;
		
		\draw [line width=1.5pt, color=red] (-4.2*\n,1*\n)  to[out=-150,in=-180] coordinate[pos=.35] (A) (-2.5*\n,-5.9*\n) ;
		\draw [line width=1.5pt, color=red] (-2*\n,-1*\n)  to[out=-90,in=0] coordinate[pos=.35] (A) (-2.5*\n,-5.9*\n) ;
		\draw [line width=1.5pt, color=red] (-2*\n,-1*\n)  to[out=90,in=30] coordinate[pos=.35] (A) (-4.2*\n,1*\n) ;
		\draw\hexcoord{0}{0} (60:\n)--(120:\n) node[midway, below, yshift=-105*\n pt,xshift=-65*\n pt,color=red] {$\gamma_1$};
		
		\draw\hexcoord{0}{0} (60:\n)--(120:\n) node[midway, below,
		yshift=75*\n pt,xshift=-45*\n pt,color=red] {$\gamma_2$};
		
		\draw\hexcoord{0}{0} (60:\n)--(120:\n) node[midway, below,
		yshift=105*\n pt,xshift=95*\n pt,color=red] {${\gamma}_3$};
		\draw\hexcoord{0}{0} (60:\n)--(120:\n) node[midway, below, yshift=-105*\n pt,xshift=75*\n pt,color=red] {$\gamma_4$};
		\foreach \k in {1,8}
				\foreach \x in {0,60,120,180,240,300}
				\foreach \y in {240}
				{\draw[shift=(\x+120:\k*\n), shift=(\x+60:\k*\n), very thick, color=black]\hexcoord{0}{0}
					(\x+\y:\k*\n)--(\x+\y+60:\k*\n);
				}
	\end{tikzpicture}}
	\caption{
		An illustration of the polymer types comprising $\caP_{N,K}$ and $\caP_{N,K}^{\rm bulk}$. Namely, $\gamma_1,\gamma_2\in\caP^{\mathrm{bulk}}_{N,K}$, and  $\gamma_1,\gamma_2,\gamma_3,\gamma_4\in\caP_{N,K}$}.\label{fig:walks_and_loops}
\end{figure}
To define the weights associated to each type of polymer, a direct computation shows that for any $x\in\caV_\Gamma$ and neighboring edges $(x,y), (x,z)\in \caE_\Gamma$,
\begin{equation}\label{basic_ints}
 \int d\Omega_x (\Omega_{x}\cdot\Omega_{x})=1,\quad \int d\Omega_x (\Omega_{y}\cdot\Omega_{x})(\Omega_{x}\cdot\Omega_{z}) = \frac{1}{3}\Omega_y\cdot\Omega_z\,.
\end{equation}
Let $|\gamma| := |\caE_\gamma|$ denote the length of any $\gamma\in\caP_\Gamma$. The weight function $W(\gamma)$ for any $\gamma\in\caP_\Gamma$ is then defined as follows. For $\gamma \in \caL_\Gamma$, applying \eqref{basic_ints},
\begin{equation}\label{loop_weight}
 W(\gamma) := \int d\Omega^{\caV_\gamma}\prod_{(x,y)\in \caE_\gamma}(-\Omega_x\cdot\Omega_y)=\left(-\frac{1}{3}\right)^{|\gamma|-1}\partial\gamma(\Omega), \quad \partial\gamma(\Omega) = -1.
\end{equation}
For any $\gamma\in \caW_\Gamma$, $|\caV_\gamma| = |\gamma|+1$. If ${\rm ep}(\gamma) = \{v,w\}$, then again using \eqref{basic_ints} we get
\begin{equation}\label{walk_weight}
W(\gamma) := \int d\Omega^{\caV_\gamma\setminus{\rm ep}(\gamma)}\prod_{(x,y)\in \caE_\gamma}(-\Omega_x\cdot\Omega_y)=\left(-\frac{1}{3}\right)^{|\gamma|-1}\partial\gamma(\Omega), \quad \partial\gamma(\Omega) = -\Omega_v\cdot\Omega_w\,.
\end{equation}

In the case of the $m$-decorated hexagonal lattice, $\Gamma^{(m)}$, $m\in \bN$, the global set of polymers is 
\[
\caP^{(m)}=\{\gamma^{(m)}\subseteq \Gamma^{(m)}: \gamma \in \caP_\Gamma\}
\]
where $\gamma^{(m)}$ is the subgraph obtained from decorating the edges of $\gamma$ with $m$ additional vertices. Given that the integration formulas in \eqref{basic_ints} still hold regardless of the infinite graph under consideration, we can generalize the weight function $W$ to any $\gamma^{(m)}\in\caP_{\Gamma^{(m)}}$, $m\geq 0$, by replacing $\gamma$ with $\gamma^{(m)}$ everywhere in \eqref{loop_weight}-\eqref{walk_weight}. Since  for all $\gamma\in\caP_\Gamma$, 
\[\partial\gamma^{(m)}(\Omega)=\partial\gamma(\Omega),\qquad \text{and}\qquad |\gamma^{(m)}|=(m+1)|\gamma|\]
this implies
\begin{equation}\label{mweight}
W(\gamma^{(m)})=\left(-\frac{1}{3}\right)^{(m+1)|\gamma|-1}\partial\gamma(\Omega):=W_m(\gamma) \qquad\forall \gamma^{(m)}\in \caP_{\Gamma^{(m)}\,.}
\end{equation}
As a result, we will be able to use the undecorated polymers to for the polymer representation of the bulk-boundary map and bulk state for either the decorated and undecorated hexagonal models given that we use the appropriate modified weight function $W_m$.

\subsection{Polymer Representations for the Hexagonal Model}\label{sec:hex_poly_reps} 

We now turn to explicit polymer models of interest. For any $m\in \bN_0$, we consider the bulk-boundary map $\omega_N^{(m)}$ and bulk state $\omega_N^{\mathrm{bulk},(m)}$ as in \eqref{bb_map}-\eqref{bulk_state} defined with respect to IAS $\Lambda_N^{(m)}$ from Section~\ref{sec:main_results}. We recall that in this case, given $0\leq K < N$ and observable $A\in\cA_{\mathring{\Lambda}_K^{(m)}}$, we have
\begin{equation}\label{eq:Z_to_Phi_m}
	Z_N^{(m)}(A;\Omega^{\partial\Lambda_N}) =\int d\rho_{\mathring{\Lambda}^{(m)}_K}A({\bf\Omega})\Phi_{N,K}^{(m)}({\bf\Omega}), \qquad
Z^{\mathrm{bulk},(m)}_N(A) =\int d\rho_{\mathring{\Lambda}^{(m)}_K}A({\bf\Omega})\Phi^{\mathrm{bulk}, (m)}_{N,K}({\bf\Omega}) 
\end{equation}
where $\Phi_{N,K}^{(m)}({\bf \Omega}) $ and $	\Phi_{N,K}^{{\rm bulk},(m)}({\bf \Omega}) $ are as in \eqref{Phi_bbmap}-\eqref{Phi_bulk}. As will be shown in Lemma~\ref{lem:poly_rep}, integrating these expressions yields a sum defined over subsets of polymers from $\caP_{N,K}$ and $\caP_{N,K}^{\rm bulk}$ with summands given in terms of the weight function $W_m$. The polymer subsets of particular interest are so-called hard core sets. Two polymers $\gamma,\gamma'\in \caP_{\Gamma}$ are \emph{disjoint} if $\caV_\gamma\cap\caV_{\gamma'} =\emptyset$ and we write
\begin{equation}\label{disjoint_polymers}
	\gamma | \gamma' \quad\text{if}\quad \caV_\gamma \cap \caV_{\gamma'} = \emptyset, \quad \text{and} \quad \gamma \nmid \gamma' \quad \text{otherwise.}
\end{equation}
A subset of polymers $\{\gamma_1,\ldots,\gamma_n\}\subseteq \caP_\Gamma$ is \emph{hard core} if $\gamma_i|\gamma_j$ for all $i\neq j$. The hard core condition can be encoded in the polymer interaction $\delta(\gamma,\gamma')=1$ if $\gamma | \gamma'$ and $\delta(\gamma,\gamma')=0$ if $\gamma \nmid \gamma'$. Hence
\begin{equation}\label{hc_interaction}
\{\gamma_1, \ldots, \gamma_n\} \; \,\text{hard core } \iff \prod_{1\leq i<j\leq n} \delta(\gamma_i,\gamma_j) = 1.
\end{equation}

\begin{lemma}\label{lem:poly_rep}
	Let $0\leq K<N$, and $W_m$ be the weight function from \eqref{mweight} for some $m\in\bZ_{\geq 0}$. Then
	\begin{align}
		\Phi_{N,K}^{(m)}({\bf\Omega}) & :=2^{-(m+1)|\caE_{\Lambda_{N,K}}|}\sum_{\{\gamma_1,...,\gamma_n\}\subseteq\cP_{N,K}}W_m(\gamma_1)\cdots W_m(\gamma_n)\prod_{1\leq i<j\leq n}\delta(\gamma_i,\gamma_j) \label{bb_poly}\\
	\Phi^{\mathrm{bulk},(m)}_{N,K}({\bf\Omega}) & :=2^{-(m+1)|\caE_{\Lambda_{N,K}}|}\sum_{\{\gamma_1,...,\gamma_n\}\subseteq\cP^{\mathrm{bulk}}_{N,K}}W_m(\gamma_1) W_m(\gamma_n) \cdot\prod_{1\leq i<j\leq n}\delta(\gamma_i,\gamma_j)\label{bulk_poly}
	\end{align}
\end{lemma}

The proof of this result was first given in \cite{lucia:2023} and follows the approach used in \cite{kennedy:1988}. We provide a shorter proof below.  

\begin{proof}
	We give the proof for the case $m=0$. The case $m>0$ follows from appropriately modifying the argument below using the observations from \eqref{mweight} and	\[
	|\caE_{\Lambda_{N,K}^{(m)}}| = (m+1)|\caE_{\Lambda_{N,K}}|\,.
	\] 
	
Now, expanding the product in the integral from \eqref{Phi_bbmap} in the case $m=0$, it follows that
	\begin{align}
 \int d\Omega^{\Lambda_{N,K}\setminus\partial\Lambda_{N,K}}\prod_{(x,y)\in \caE_{\Lambda_{N,K}}}(1-\Omega_x\cdot\Omega_y) 
		& =  \sum_{E\subseteq \caE_{\Lambda_{N,K}}}\int d\Omega^{\Lambda_{N,K}\setminus\partial\Lambda_{N,K}}\prod_{(x,y)\in E}(-\Omega_x\cdot\Omega_y)\,.\label{subgraph_sum}
	\end{align}
As the measure $d\Omega_x$ is invariant under the map $\Omega_x\mapsto -\Omega_x,$ any Lebesgue integrable function $f:S^2 \to \bR$ satisfies
	\begin{equation}\label{invariant_consequence}
		\int d\Omega_x f(-\Omega_x) = \int d\Omega_x f(\Omega_x)\qquad \forall x\in \Gamma\,,
	\end{equation}
and, in particular, this integral vanishes if $f$ is an odd function. Thus, if $G_E=(\caV_E,E)\subseteq\Lambda_{N,K}$ denote the minimal graph with edge set $E\subseteq \caE_{\Lambda_{N,K}}$, and if
	\[\deg_{G_E}(v) \in\{1,3\}\quad \text{for some}\quad v\in \caV_E\cap\mathring{\Lambda}_{N,K}.\] 
Then by \eqref{invariant_consequence} integrating over $d\Omega_v$ yields
	\begin{align}\label{zero_graphs}
		\int d\Omega_v\prod_{(x,y)\in E}(-\Omega_x\cdot\Omega_y)= 0\,.
	\end{align}
	As can be seen in Figure~\ref{fig:lambdank}, the definition of $\Lambda_{N,K}$ is such that $\deg_{\Lambda_{N,K}}(x)=1$ for any $x\in\partial\Lambda_{N,K}$. Given this and \eqref{zero_graphs}, the sum in \eqref{subgraph_sum} can be reduced to all subsets $E$ such that
	\begin{equation}\label{edge_set_conditions}
		\deg_{G_E}(x) =2 \quad \text{if}\quad x\in\caV_E\setminus\partial\Lambda_{N,K} \quad \text{and} \quad \deg_{G_E}(x) = 1 \quad \text{if}\quad x\in V_E\cap\partial\Lambda_{N,K} \,.
	\end{equation}
	Thus, the connected components of $G_E$ comprise a hard core set of polymers $\{\gamma_1, \ldots, \gamma_n\}\subseteq \caP_{N,K}$. Moreover, given any hard core set of polymers $\{\gamma_1, \ldots, \gamma_n\}\subseteq \caP_{N,K}$, the union of their respective edge sets $E=\cup_i\caE_{\gamma_i}$ produces a set which satisfies \eqref{edge_set_conditions}. Therefore,
	\begin{equation}\label{hc_sum}
		\Phi_{N,K}^{(0)}({\bf \Omega})=2^{-|\caE_{\Lambda_{N,K}}|}\sum_{\substack{\text{hard core}\\\{\gamma_1,\ldots, \gamma_n\}\subseteq \caP_{N,K}}}\int  d\Omega^{\Lambda_{N,K}\setminus\partial\Lambda_{N,K}}\prod_{i=1}^n\prod_{(x,y)\in \caE_{\gamma_i}}(-\Omega_x\cdot\Omega_y)\,.
	\end{equation}
The result follows from using \eqref{hc_interaction} to rewrite the hard core condition in the above summation, and then applying \eqref{loop_weight}-\eqref{walk_weight} and $\int d\Omega_x = 1$ to integrate \eqref{hc_sum}.
	
	The proof for  $Z_N^{\rm bulk}(A)$ runs analogously, with the only difference being that one also integrates \eqref{subgraph_sum} over $x\in\partial\Lambda_N$. Since $\deg_{\Lambda_{N,K}}(x) = 1$ for all $x\in \partial\Lambda_N$, \eqref{zero_graphs} implies that \eqref{subgraph_sum} can be reduced to all edge sets $E$ whose minimal graph satisfies the conditions of \eqref{edge_set_conditions} and are subgraphs of $\Lambda_{N,K}\setminus\partial\Lambda_{N}$. The set of all such graphs is the precisely the set of graphs whose connected components are a hard core sets of polymers $\{\gamma_1,\ldots,\, \gamma_n\}\subseteq\caP_{N,K}^{\rm bulk}$.
\end{proof}

 Thus, in order to prove the desired bounds on \eqref{two_bounds} reduces to analyzing the logarithms of the right hand side of \eqref{bb_poly}-\eqref{bulk_poly}. After verifying that a convergence criterion holds, the cluster expansion gives an explicit formula for these logarithms, which will be inserted into \eqref{Renyi_divergence} and manipulated in Section~\ref{sec:proof} to prove Theorem~\ref{thm:gs_indistinguishability}. That $\Phi_{N,K}^{\rm bulk}>0$ will be an immediate consequence of the cluster expansion, so Lemma~\ref{lem:op_norm} will apply to our polymer representations. We now review the cluster expansion convergence criterion from \cite{kotecky:1986}, and then verify this for specific polymer representations of $\Phi^{{\rm bulk},(m)}_{N,K}$ and $\Phi_{N,K}^{(m)}$ in Section~\ref{sec:Cluster_convergence}.

\section{The Koteck\'{y}-Preiss-Ueltschi Condition for Cluster Expansion Convergence}\label{sec:KP}
	
Theorem~\ref{thm:gs_indistinguishability} will be a consequence of estimating the logarithm differences from Lemma~\ref{lem:op_norm} that arise from \eqref{eq:Z_to_Phi_m} using the expansion in Lemma~\ref{lem:poly_rep} as well as their decorated square lattice model counterparts. These logarithms can be written explicitly as cluster expansions so long as a summability criterion on the polymer sets is satisfied. In this section, we review the convergence criterion due to Ueltschi \cite{ueltschi:2004} which generalizes the convergence criterion for polymers with a hard core interaction from Koteck\'{y} and Preiss \cite{kotecky:1986} to polymer sets with a soft core interaction, and we refer to this more general condition as the Koteck\'{y}-Preiss-Ueltschi (KPU) criterion. As will be shown in Section~\ref{sec:lieb}, the ground states for the AKLT model on the decorated square lattice have a soft core polymer representation, so this more general setting is needed. In addition, we also state a consequence of the convergence criterion, Corollary~\ref{lem:kplemma} below, that will be vital in our proof of Theorem~\ref{thm:gs_indistinguishability}. 

Let $\caP$ be an nonempty, finite set of polymers, and $W:\caP\to\bR$ be a weight function on $\caP$. Assume that $\caP$ is endowed with a polymer interaction $\delta:\caP\times\caP \to [0,1]$, meaning that 
\[\delta(\gamma,\gamma')=\delta(\gamma',\gamma), \qquad \delta(\gamma,\gamma)=0\qquad \forall\gamma,\gamma'\in \caP.\] 
The interaction is \emph{hard core} if $\delta$ only takes values in $\{0,1\}$, and \emph{soft core} otherwise. The \emph{partition function} associated to $\caP,W$ and $\delta$ is defined by
\begin{equation}
	\Xi(\caP, W,\delta)=\sum_{\vgg=\{\gamma_1,\ldots,\gamma_n\}\subseteq \caP}W(\vgg)\cdot\prod_{1\leq i<j\leq n}\delta(\gamma_i,\gamma_j), \qquad W(\vgg) = \prod_{i=1}^nW(\gamma_i)\,.
\end{equation} 

To connect this to the polymer representation for the AKLT model on the hexagonal lattice from Section~\ref{sec:hex_poly_reps}, fix $N>K\geq 1$ and recall the weight function $W_m$ from \eqref{mweight} and polymer interaction $\delta$ from \eqref{hc_interaction}. Then, $\Phi_{N,K}^{(m)}$ and $Z_N^{(m)}$ can be written as
\[
\Phi_{N,K}^{(m)} = 2^{-(m+1)|\caE_{\Lambda_{N,K}}|}\Xi(\caP_{N,K},W_m,\delta), \qquad Z_{N}^{(m)} = 2^{-(m+1)|\caE_{\Lambda_{N}}|}\Xi(\caP_{N,0}, W_m,\delta)
\]
and analogously for $\Phi_{N,K}^{\rm bulk,(m)}$ and $Z_{N}^{\rm bulk,(m)}$ using the bulk polymer sets. Here, we note that $W_m(\gamma)$ is actually a continuous function of the boundary variables $(v_x,u_x:x\in\partial\Lambda_{N,K})$. However, in Section~\ref{sec:proof} we will apply the cluster expansion pointwise, and in Section~\ref{sec:convergence_proof} we use the uniform bound $|W_m(\gamma)|\leq 3^{-(m+1)|\gamma|+1}$ to verify that the convergence criterion from Theorem~\ref{thm:kp} below holds for all choices of the boundary variables.

For any set of polymers, $\caP$, with polymer interaction, $\delta$, a nonempty subset $C=\{\gamma_1,\ldots,\gamma_n\} \subset \caP$ is called a \emph{cluster} if it cannot be partitioned into two nonempty sets $C_1, \, C_2$ such that
\begin{equation}\label{cluster_def}
\zeta(\gamma_1,\gamma_2)=0 \quad \forall \gamma_1\in C_1,\, \gamma_2\in C_2 \qquad \text{where} \qquad \zeta(\gamma,\gamma'):=\delta(\gamma,\gamma')-1\,.
\end{equation}
Said differently, $C$ is a cluster if the graph $G=(\{1,\ldots,n\}, E)$ with $E=\{(i,j):\zeta(\gamma_i,\gamma_j)\neq 0\}$ is connected. The set of all clusters is denoted $\caC(\caP)$ and, for any $\gamma\in \caP$ and $C\in\caC(\caP)$, we write $\gamma\nmid C$ if $C\cup\{\gamma\}$ is a cluster. 

For the AKLT models considered here, every polymer can be naturally mapped to a connected subgraph of the lattice. In the hexagonal model case, $\zeta(\gamma,\gamma')=0$ if and only if the subgraphs associated to the polymers are disjoint.  In the square model case, $\zeta(\gamma,\gamma')=0$ if and only if the subgraphs associated to the polymers are disjoint except for, possibly, on $\partial\Lambda_{N,K}$. In either case, a set $C=\{\gamma_1,\ldots,\gamma_n\}$ being a cluster requires that the graph $G=(\cup_i\caV_{\gamma_i},\cup_i\caE_{\gamma_i})$ is connected, and this condition is an equivalence in the hexagonal model case. We note that the notation for intersecting polymers introduced in \eqref{disjoint_polymers} for the hexagonal AKLT model is consistent with the cluster notation above since $\gamma\nmid C$ if and only if there exists $\gamma'\in C$ such that $\gamma\nmid \gamma'$.

When satisfied, a cluster expansion criterion guarantees that the formal sum for $\log \Xi(\caP,W)$ given below, called the \emph{cluster expansion}, is absolutely convergent and, thus, a well-defined equality of functions. Specifically,
\begin{equation} \label{cluster_expansion}
	\log \Xi(\caP,W) = \sum_{C\subseteq \caC(\caP)} W^T(C), \qquad W^T(C) = \sum_{n\geq 1}\sum_{\substack{(\gamma_1,\ldots, \gamma_n)\in\caP^{n}:\\ C=\{\gamma_1,\ldots,\gamma_n\}} }\varphi(\gamma_1,\ldots, \gamma_n) \prod_{i=1}^nW(\gamma_i)\,
\end{equation}
where, $\varphi$ is the Ursell function defined by 
\[
\varphi(\gamma_1) = 1, \qquad \varphi(\gamma_1,\ldots,\gamma_n) = \frac{1}{n!}\sum_{\substack{G=(V_n,E) :\\ G\; \text{connected}}}\prod_{(i,j)\in E}\zeta(\gamma_i,\gamma_j),\quad n \geq 2\]
and $V_n=\{1,\ldots, n\}$. We state the Kotechk\'y-Preiss-Ueltschi cluster expansion convergence criterion, and point the interested reader to \cite[Chapter 5]{friedli:2017} for a derivation of \eqref{cluster_expansion} and a simplified proof of its absolute convergence when the KPU criterion holds for a discrete polymer model.

	\begin{thm}[KPU Criterion for Cluster Expansion Convergence \cite{kotecky:1986,ueltschi:2004}]\label{thm:kp}
		Let $\caP$ be a non\-empty finite set of polymers with weight function $W$ and polymer interaction $\delta:\caP\to[0,1]$. Suppose there exists functions $a,b:\caP\to [0,\infty)$ such that 
		\begin{equation}\label{KP_criterion}
			\sum_{\gamma'\in\caP}e^{a(\gamma')+b(\gamma')}|W(\gamma')|\cdot|\zeta(\gamma,\gamma')|\leq a(\gamma), \qquad \forall \gamma\in \caP
		\end{equation}
	where $\zeta=\delta-1$. Then the triple sum in \eqref{cluster_expansion} is absolutely convergent. Moreover, for any $\gamma\in\caP$,
		\begin{equation}
			\sum_{C\in \caC(\caP), C\nmid\gamma}|W^T(C)|e^{b(C)}\leq a(\gamma) \quad \text{where}\quad b(C) =\sum_{\gamma\in C}b(\gamma)\,.
		\end{equation}

	\end{thm}

The convergence of the cluster expansion would also hold if $b(\gamma)=0$ for all $\gamma \in \caP$. However, the following corollary will be particularly useful for proving Theorem~\ref{thm:gs_indistinguishability} given that a nonzero $b:\caP\to[0,\infty)$ can be established. This is a slight modification of the bound proved in \cite{kotecky:1986}.
	\begin{cor}[\cite{kotecky:1986}]\label{lem:kplemma} Let $\caP$ be any nonempty, finite set of polymers with weight function $W$ and polymer interaction $\delta:\caP\to[0,1]$, and suppose that $a,b:\caP\to[0,\infty)$ are such that Theorem~\ref{thm:kp} holds. Then, for any nonempty $A\subseteq \caP$ and $\gamma\in\caP$,
		\begin{equation}
\sum_{\substack{C\in\caC(\caP):\\C\cap A\neq\emptyset,\, C\nmid\gamma}}|W^T(C)|\leq e^{-b(\gamma, A)}a(\gamma)
		\end{equation}
	where $b(\gamma, A) =\inf\{b(C): C\in\caC(\caP),\,C\nmid\gamma,\,C\cap A\neq\emptyset\} $.
	\end{cor}
	\begin{proof}
		 Since $b$ is nonnegative, the desired bound follows from Theorem~\ref{thm:kp} as
		$$\sum_{\substack{C\in\caC(\caP):\\C\cap A\neq\emptyset,\, C\nmid\gamma}}|W^T(C)|e^{b(\gamma, A)}\leq\sum_{C\in \caC(\caP),\,C\nmid\gamma}|W^T(C)|e^{b(C)}\leq a(\gamma)\,.$$
	\end{proof}

We verify \eqref{KP_criterion} in Section~\ref{sec:Cluster_convergence} for the (un)decorated hexagonal AKLT model. The case of the AKLT model on the decorated square lattice will be the content of Section~\ref{sec:lieb}.

\section{Convergence of the Cluster Expansion for the Hexagonal Model}\label{sec:Cluster_convergence}
We once again turn to considering the AKLT model on the hexagonal lattice. Using the weight function \eqref{mweight} and polymer interaction \eqref{hc_interaction}, the first step in proving Theorem~\ref{thm:gs_indistinguishability} for this model is to show that the cluster expansion \eqref{cluster_expansion} converges for either $\caP_{N,K}$ or $\caP_{N,K}^{\rm bulk}$ with appropriate choices of $N>K\geq 0$. Since $\caP_{N,K}^{\rm bulk}\subseteq \caP_{N,K}$, the convergence for both polymer sets follows from verifying \eqref{KP_criterion} for $\caP_{N,K}$. Key to proving Theorem~\ref{thm:gs_indistinguishability} is that the functions used to check this criterion are independent of the chosen $N$ and $K$. To this end, define $a:\bN_{\geq 3}\to[0,\infty)$ by
\begin{equation}\label{a_function}
	a(l) = 
	\begin{cases}
		0.52, &  l=3\\
		0.56, & l=4\\
		0.66, & l=5\\
		0.70, & l=6\\
		0.15l, & l\geq 7\\
	\end{cases}
\end{equation} 
Above, the domain of $a$ reflects the fact that $|\gamma|\geq 3$ for any $\gamma\in\cP_{N,K}.$

\begin{thm}[Cluster Expansion Convergence for $\caP_{N,K}$]\label{thm:KP_condition}
	Let $m\in \bZ_{\geq 0}$. Fix $N$ and $K$ such that $N-K\geq 53$ and $K\geq 25$ or $K=0$, and consider $\cP_{N,K}$, $W_m$ and $\zeta=\delta-1$ as defined in \eqref{polymer_sets}, \eqref{mweight} and \eqref{hc_interaction}. Then, for any $\gamma\in\cP_{N,K}$ 
	\begin{equation}\label{eqn:cluster_condition}
		\sum_{\substack{\gamma'\in\cP_{N,K}}}|W_m(\gamma')|\cdot|\zeta(\gamma,\gamma')|\cdot e^{a_m(|\gamma'|)+\epsilon(m+1)|\gamma'|}\leq a_m(|\gamma|)\,
	\end{equation}    
where $a_m(l)=(m+1)a(l)$ and $\epsilon=0.0086$.
	
\end{thm}

	 The approach we use to prove Theorem~\ref{thm:KP_condition} is the same as the one from \cite{kennedy:1988} that verified \eqref{eqn:cluster_condition} held for $\cP_{N,0}$ with $\epsilon=0$. As such, we prove Theorem~\ref{thm:KP_condition} in  Section~\ref{sec:convergence_proof} for $K\geq 25$ as this can be adapted with minor (but tedious) changes to also hold for $K=0$ with $\epsilon = 0.0086$. We specifically work with $\cP_{N,K}$ in order to establish the dependence of \eqref{LTQO} on $|\gamma^{K,(m)}|$. Since $\cP_{N,K}\not\subseteq \cP_{N,0}$, the result from \cite{kennedy:1988} cannot be directly invoked to verify Theorem~\ref{thm:gs_indistinguishability}, and certain arguments used in that work need to be modified so that they hold for the more general case.

The strategy for proving Theorem~\ref{thm:KP_condition} is as follows. From \eqref{mweight}, it is clear that for any $\gamma\in \caP_{N,K}$,
\[
|W_m(\gamma)| \leq 3^{-(m+1)|\gamma|+1} \qquad \forall m\geq 0.
\]
Since $\zeta(\gamma,\gamma')=1$ if $\gamma\nmid\gamma'$ and $\zeta(\gamma,\gamma')=0$ otherwise, we restrict the sum in \eqref{eqn:cluster_condition} to $\gamma\nmid\gamma'$  and drop the $\zeta$ term. Then, it follows that for any $\gamma\in\cP_{N,K}$,
\begin{equation} \label{eqn:cluster_approach}
	\sum_{\substack{\gamma'\in\caP_{N,K}\\ \gamma'\nmid\gamma}}|W_m(\gamma')|\cdot e^{a_m(|\gamma'|)+\epsilon(m+1)|\gamma'|}\leq\sum_{l=3}^{\infty}w_m(l)|\cP_{N,K}^{\gamma,l}|
\end{equation}
where we introduce
\begin{equation}\label{little_w}
	w_m(l):=3\left(\frac{e^{a(l)+\epsilon l}}{3^{l}}\right)^{m+1} \quad\text{and}\quad \cP_{N,K}^{\gamma,l} := \{\gamma'\in \cP_{N,K}: \gamma\nmid\gamma', \; |\gamma'|=l\}\,.
\end{equation}
Since $w_m(l)/(m+1)$ is decreasing in $m$, the result for all $m$ follows from verifying 
\begin{equation}\label{m0_convergence}
\sum_{l=3}^{\infty}w_0(l)|\cP_{N,K}^{\gamma,l}| \leq a(|\gamma|).
\end{equation}
To simplify notation moving forward, we suppress the subscript above and write $w(l)=w_0(l)$.

One only needs to produce sufficiently tight upper bound on $|\cP_{N,K}^{\gamma,l}|$ to verify \eqref{eqn:cluster_condition} holds. This will be achieved by considering loops and walks separately, so for any $\gamma\in\cP_{N,K}$ and $l\geq 3$, let 
\begin{align}
	\mathcal{W}_{N,K}^{\gamma,l} & = \{\gamma'\in \caW_{N,K}: \gamma\nmid\gamma', \; |\gamma'|=l\}\label{walks_gamma_l}\\   
	\mathcal{L}_{N,K}^{\gamma,l} & = \{\gamma'\in \caL_{N,K}: \gamma\nmid\gamma', \; |\gamma'|=l\} \label{loops_gamma_l}
\end{align}
denote the set of all walks (resp. loops) of length $l$ from $\caP_{N,K}$ that intersect $\gamma$. For convenience later on, we also denote the set of all walks (resp. loops) of length $l$ by
\begin{equation}
	\mathcal{W}_{N,K}^{l} = \{\gamma\in \caW_{N,K}: |\gamma|=l\}, \quad \mathcal{L}_{N,K}^{l} = \{\gamma\in \caL_{N,K}: |\gamma|=l\}. \label{walk_loop_l}
\end{equation}

Given these definitions, one can write
\begin{equation}\label{split}
	\begin{aligned}
	\sum_{l=3}^{\infty}w(l)|\cP_{N,K}^{\gamma,l}|&
		=\sum_{\substack{3\leq l\leq 20 \\ l \; \text{even}}} w(l)|\caW_{N,K}^{\gamma,l}|+\sum_{\substack{3\leq l\leq 20 \\ l \; \text{odd}}} w(l)|\caW_{N,K}^{\gamma,l}|+
		\sum_{l=6,10} w(l)|\caL_{N,K}^{\gamma,l}|\\
		&+\sum_{12\leq l\leq 28}w(l)|\caL_{N,K}^{\gamma,l}|
		+\sum_{l\geq 30}w(l)|\caL_{N,K}^{\gamma,l}|+\sum_{l\geq 21}w(l)|\caW_{N,K}^{\gamma,l}|
	\end{aligned}
\end{equation}
where the choice of these six specific sums is motivated by grouping together sets of polymer whose cardinalities we will bound using similar arguments. Here, we use that loops must have even length, and there are no loops of length $l=8$ in the hexagonal lattice.

\subsection{Cardinality bounds on polymer sets}\label{sec:Cardinality_bounds} The goal of this section is to produce the desired cardinality bounds for the sets in \eqref{split} to prove \eqref{m0_convergence} for all $\gamma\in\caP_{N,K}$. While motivated by the approach in \cite{kennedy:1988}, which establishes these bounds for $\caP_{N,0}$, the nontrivial topology of $\Lambda_{N,K}$ and the additional polymers this creates requires that we again verify sufficient cardinality bounds on these polymer sets. 

For a fixed $\gamma\in\caP_{N,K}$, the main challenge for establishing \eqref{eqn:cluster_condition} comes from bounding the contributions in \eqref{split} from polymers $\gamma'\nmid\gamma$ of small length, specifically,
\begin{equation}
	\gamma' \in \bigcup_{3\leq l\leq 20} \cW_{N,K}^{\gamma,l}\cup \bigcup_{6\leq l\leq 28} \cL_{N,K}^{\gamma,l}\,.
\end{equation}
The assumptions on $N$ and $K$ in Theorem~\ref{thm:KP_condition} are made to simplify bounding $|\cW_{N,K}^{\gamma,l}|$ and $|\cL_{N,K}^{\gamma,l}|$ for such $l$ and otherwise play no role in the proof of the main result. We first, however, give generic bounds on the cardinalities of these sets that, while not tight enough to bound all six sums in \eqref{split}, will be sufficient to bound the final two sums for all cases of $\gamma$ considered in Section~\ref{sec:convergence_proof}. For brevity, we fix $N$ and $K$ such that $N-K\geq 53$ and $K\geq 25$, and suppress these labels from the sets of polymers, e.g., \eqref{little_w}-\eqref{walk_loop_l}. All bounds produced in this section are uniform in such a choice.

The generic upper bounds produced on $|\caW^{\gamma, l}|$ and $|\caL^{\gamma, l}|$ will depend on the type and length of the polymer $\gamma\in\cP$. To this end, for any $l,l'\in \bN_{\geq 3}$, define
\begin{alignat}{23}
	M_{\ell \ell}(l,l') & =\sup_{\gamma\in\cL^{l}}|\cL^{\gamma,l'}| &&\qquad &&  M_{\ell w}(l,l') =\sup_{\gamma\in\cL^{l}}|\cW^{\gamma,l'}|
	\label{Ml}\\
	M_{w \ell}(l,l') & =\sup_{\gamma\in\cW^{l}}|\cL^{\gamma,l'}|  
	&&\qquad&& M_{ww}(l,l') =\sup_{\gamma\in\cW^{l}}|\cW^{\gamma,l'}| \label{Mw}
\end{alignat}
In words, $M_{\ell w}(l,l')$ is the maximum number of walks of length $l'$ that intersects a fixed loop $\gamma$ of length $l$, and similarly for the other three functions.

The bounds on \eqref{Ml}-\eqref{Mw} in Lemma~\ref{lem:loopbd} are based on the observation that the geometry of the hexagonal lattice and the definition of $\Lambda_{N,K}$ guarantee that any pair of polymers $\gamma,\gamma'\in \cP$ must share an edge $e$ if $\gamma\nmid\gamma'$.  To state these bounds, first partition the set of edges of $\Lambda_{N,K}$ by
\begin{equation*}
	\caE^b = \{  e \in \caE_{\Lambda_{N,K}} : e \text{ is a boundary edge}\}, \quad \caE^i = \caE_{\Lambda_{N,K}}\setminus \caE^b \\
\end{equation*}
where $e=(v,w)\in \caE_{N,K}$ is a \emph{boundary edge} if $\{v,w\}\cap\partial\Lambda_{N,K}\neq \emptyset$, and an \emph{interior edge} otherwise. Then, the following three functions capture the maximum number of polymers of a given type and length that can contain a fixed edge:
\begin{align}
	N_w^i(l) & =\sup_{e\in \caE^i}|\{\gamma\in\mathcal{W}^l : e\in\gamma\}| \\
	N_w^b(l) & =\sup_{e\in \caE^b}|\{\gamma\in\mathcal{W}^l : e\in\gamma\}|\\
	N_\ell(l) & =\sup_{e\in\caE_{N,K}}|\{\gamma\in\cL^l:\:e\in\gamma\}| = \sup_{e\in\caE^i}|\{\gamma\in\cL^l:\:e\in\gamma\}|
\end{align}
The second equality for $N_\ell$ follows since a loop cannot contain a boundary edge (as all boundary vertices satisfy $\deg_{\Lambda_{N,K}}(v) = 1$). Above and the rest of the text, for an edge $e$ and vertex $v$ we will slightly abuse notation and write
\[
e,v\in G\iff e\in \caE_G, \; v\in \caV_G\,.
\]

\begin{lemma}[\cite{kennedy:1988}]\label{lem:loopbd} For any $l\geq 3$ the functions from \eqref{Ml}-\eqref{Mw} satisfy
	\begin{align}
		M_{\ell\ell}(l,l')& \leq lN_\ell(l')  \label{Mll_bound} \\
		M_{w\ell}(l,l')&\leq (l-2)N_\ell(l') \label{Mwl_bound}\\
		M_{\ell w}(l,l') & \leq lN_w^i(l') \label{Mlw_bound}\\
		M_{ww}(l,l') &\leq2N_w^b(l')+(l-2)N_w^i(l')  \label{Mww_bound}
	\end{align}
	Moreover, $N_\ell(l)\leq 2^{l-3}$ for $l\geq 6$,  $N_w^b(l)\leq 2^{l-5}$ for $l\geq 11$ and $N_w^i(l)\leq(2l+95)2^{l-10}$ for $l\geq 21$.
\end{lemma}

 We note that, trivially, $N_\ell(l) \leq N_\ell^\Gamma(l)$ where $N_\ell^\Gamma(l)$ is the maximum number of loops of length $l$ in the hexagonal lattice that include a fixed edge, i.e.,
\[
N_\ell^\Gamma(l) = \sup_{e\in \cE_\Gamma}|\{\gamma\in \cL_\Gamma: |\gamma|=l\, \wedge\, e\in \gamma\}|\,.
\]
The values of $N_\ell^\Gamma(l)$ for $l\leq 28$ were calculated with the code in Appendix~\ref{subsec:N} and are given in Table~\ref{table:loops}. 

\begin{table}
	\begin{tabular}{ |p{2cm}||p{2cm}|}
		\hline
		\multicolumn{2}{|c|}{$N_\ell^\Gamma(l)$ for $l\leq 28$} \\
		\hline
		$l$ & $N_\ell^\Gamma(l)$\\
		\hline
		6&2\\
		8&0\\
		10&10\\
		12&8\\
		14&56\\
		16&96\\
		18&390\\
		20&920\\
		22&3168\\
		24&8592\\
		26&28002\\
		28&81368\\
		\hline
	\end{tabular}
	\caption{The maximum number of loops of length $l\leq 28$ that intersecting a single edge. }
	\label{table:loops}
\end{table}

\begin{proof}
	To begin, fix any $\gamma\in\cW^{l}$, and suppose that $\gamma'\in\cL^{\gamma,l'}$. Since $\gamma'$ is a loop, there must be an interior edge $e\in \gamma \cap \gamma'$. As $\gamma$ has exactly $l-2$ interior edges, 
	\[|\cL^{\gamma,l'}| \leq (l-2)N_\ell(l'),\] 
	which implies \eqref{Mwl_bound}. In the case that $\gamma'\in \cW^{\gamma, l'}$, it must be that $\gamma'$ contains either a boundary edge of $\gamma$ or an interior edge of $\gamma$, and so \eqref{Mww_bound} follows from
	\[
	|\cW^{\gamma, l'}| \leq 2N_w^b(l')+(l-2)N_w^i(l')\,.
	\]
	The analogous arguments produce the bounds on $M_{\ell\ell}(l,l')$ and $M_{\ell w}(l,l')$ after using that every edge of $\gamma\in\cL^{l}$ is an interior edge.
	
	For the bound on $N_\ell(l)$, fix an edge $e=(v,v')$ that must be contained in a loop of size $l$. Starting from $e$, lay $l-3$ additional edges starting from $v'$ to create a walk of length $l-2$ with endpoints ${v,w}$. There are at most two choices for laying each successive edge, and so there are at most $2^{l-3}$ such paths. The result follows since the geometry of the hexagonal lattice is such that there is a unique pair of edges that can connect $v$ and $w$ if the walk can be closed to a loop of length $l$. (If there were two or more such pairs of edges, then there would be a loop of length four in the hexagonal lattice.) 
	
	To bound $N_w^b(l)$ and $N_w^i(l)$, notice that any walk $\gamma\in\cW^l$ containing an edge $e=(v,v')\in\caE_{N,K}$ can be viewed as the concatenation of $e$ with two disjoint walks $\gamma_{v},\gamma_{v'}\in\cP_\Gamma$ that begin at the vertices $v$ and $v'$, respectively, and end at the boundary, $\partial\Lambda_{N,K}$. Here, we use the convention that one of these walks is the empty graph if $e$ is a boundary edge. Thus, $l=|\gamma_{v}|+|\gamma_{v'}|+1$. Now, let $n(v,e,l)$ be the number of walks $\gamma\in \cP_\Gamma$ of length $l$ that begin at $v$, do not contain $e$, and end at the boundary $\partial\Lambda_{N,K}$, and set
	\[N(l)=\sup_{\substack{v\in\caV_{N,K}\\ e=(v,v')\in\caE_{N,K}}}n(v,e,l).\]
	An upper bound on the value of $N(l)$ for $l\leq 10$ is given in Table~\ref{table:N}. These values were computed using the code found in Appendix~\ref{subsec:Pb}.

	\begin{table}
		\centering
		\begin{tabular}{ |p{2cm}||p{2cm}|}
			\hline
			\multicolumn{2}{|c|}{$\caN(l)$ for $l\leq10$} \\
			\hline
			$l$ & $\caN(l)$\\
			\hline
			1&1\\
			2&2\\
			3&2\\
			4&4\\
			5&6\\
			6&8\\
			7&16\\
			8&24\\
			9&40\\
			10&64\\
			\hline
		\end{tabular}
		\caption{Maximum number of walks $\caN(l)$ of length $l\leq 10$ from a fixed site to the boundary. This serves as an upper bound on $N(l)$.}
		\label{table:N}
	\end{table} 
	
	Considering the case that $e\in \caE^b$ and $e\in\caE^i$ separately, the above observations imply
	\begin{equation}\label{Nw_initial_bounds}
		N_w^b(l) \leq N(l-1), \quad N_w^i(l) \leq \sum_{l'=1}^{l-2}N(l')N(l-l'-1)
	\end{equation}
	where, in the case that $e=(v,v')\in \caE^i$, one uses that $\min\{|\gamma_{v}|, |\gamma_{v'}|\}\geq 1$. As $N(10)= 64$, the crude bound $N(l)\leq 2^{l}$ for $l\geq 11$ can be improved since any such walk is the concatenation of a walk of length $l-10$ starting at $v$ and avoiding $e$, and a walk of length ten ending at a boundary edge. Thus, for $l\geq 10$
	\begin{equation}\label{N_bound}
		N(l)\leq 2^{l-10}N(10) \leq 2^{l-4}\,.
	\end{equation}
	Inserting this into \eqref{Nw_initial_bounds} produces the desired upper-bound on $N_w^b$. For the bound on $N_1^{i}(l)$, note that $l-l'-1\geq 11$ when $l'\leq 10$ and $l\geq 22$, and so \eqref{N_bound} implies
	\begin{align}
		\sum_{l'=1}^{l-2}N(l')N(l-l'-1)&  =  2\sum_{l'=1}^{10} N(l')N(l-l'-1) + \sum_{l'=11}^{l-12} N(l')N(l-l'-1)\nonumber \\
		& \leq  2\sum_{l'=1}^{10} N(l')2^{l-l'-5} +(l-22)2^{l-9} \leq (2l+95)2^{l-10} \label{95}
	\end{align}
	where the last bound is obtained by inserting the values from Table~\ref{table:N}. In the case that $l=21$, the same bound follows from applying the above argument to the expansion
	\[
	\sum_{l'=1}^{19}N(l')N(20-l') = 2\sum_{l'=1}^{9} N(l')N(l-l'-1) + N(10)^2\,.
	\]
	
\end{proof}
The above results are sufficient for proving Theorem~\ref{thm:KP_condition} when the length of the fixed polymer, $\gamma$, is three. In the case that $|\gamma|>3$, the bounds from Lemma~\ref{lem:loopbd} can once again be used to bound the last four sums from \eqref{split}. The remainder of this section is focused on bounded the first two sums in \eqref{split}, for which the desired result will be partitioned by $|\gamma|$.

 As shown in \cite{kennedy:1988} and restated in Lemma~\ref{lem:small_walks}, in the case that $4\leq |\gamma|\leq 6$, the first two sums in \eqref{split} can be bounded using a difference of supremums. In this case, for any $4 \leq l' \leq 20$, let
\begin{align}
S_{ww}(l,l')  & = \sup_{\gamma\in \caW^{l}}\sum_{j=3}^{l'}|\caW^{\gamma,j}|-\sup_{\gamma\in \caW^{l}}\sum_{j=3}^{l'-1}|\caW^{\gamma,j}|, \qquad 4\leq l \leq 6 \label{S_ww_bdy}\\
S_{\ell w}(6,l')  & = \sup_{\gamma\in \caL^{6}}\sum_{j=3}^{l'}|\caW^{\gamma,j}|-\sup_{\gamma\in \caL^{6}}\sum_{j=3}^{l'-1}|\caW^{\gamma,j}|\label{S_le_bdy}\end{align}
and for $l'=3$ set
\begin{equation}\label{Sl3}
S_{ww}(l,3) = M_{ww}(l,3) \quad \text{for}\quad 4 \leq l \leq 6, \qquad S_{\ell w}(6,3) =M_{\ell w}(6,3)\,.
\end{equation}

\begin{table}
	\centering
	\begin{tabular}{ |p{2cm}|p{2.3cm}|p{2.3cm}|p{2.3cm}|p{2.3cm}|p{2.3cm}|}
		\hline
		& $(a,l)=(w,3)$& $(a,l)=(w,4)$ & $(a,l)=(w,5)$ & $(a,l)=(w,6)$ & $(a,l)=(\ell,6)$\\
		\hline
		$M_{aw}(l,3)$ &1&2&2&3&2\\
		$S_{aw}(l,4)$ &2&2&2&3&2\\
		$S_{aw}(l,5)$ &2&1&1&1&1\\
		$S_{aw}(l,6)$ &2&2&2&3&2\\
		$S_{aw}(l,7)$ &6&2&2&2&2\\
		$S_{aw}(l,8)$ &8&9&8&12&10\\
		$S_{aw}(l,9)$ &14&7&7&7&7\\
		$S_{aw}(l,10)$ &18&22&18&27&24\\
		$S_{aw}(l,11)$ &38&20&20&20&20\\
		$S_{aw}(l,12)$ &52&70&56&78&70\\
		$S_{aw}(l,13)$ &106&62&61&62&62\\
		$S_{aw}(l,14)$ &150&224&164&225&221\\
		$S_{aw}(l,15)$ &296&193&186&193&193\\
		$S_{aw}(l,16)$ &428&655&494&644&641\\
		$S_{aw}(l,17)$ &868&606&568&606&606\\
		$S_{aw}(l,18)$ &1284&2084&1516&1940&2066\\
		$S_{aw}(l,19)$ &2530&1930&1760&1930&1930\\
		$S_{aw}(l,20)$ &3818&6504&4692&5793&6578\\
		$M_{a\ell}(l,6)$ &1&2&\textbf{3}&3&7\\
		$M_{a\ell}(l,10)$ &3&7&\textbf{11}&10&30\\
		\hline
	\end{tabular}
	\caption{The values of the quantities in the first column for each the choice of $(a,l)$. The numbers in bold are those which are different than those of Table II of \cite{kennedy:1988} due to walks on the inner boundary. All but the last two rows are calculated via the code in Appendices~\ref{subsec:saws_col1-4}-\ref{subsec:saws_col5}, and the last two rows can be calculated easily by hand.}
	\label{table:saw}
\end{table}

\begin{lemma}[\cite{kennedy:1988}]\label{lem:small_walks} For $S_{ww}(l,l')$ and $S_{lw}(6,l')$ as above, one has
	\begin{align}
	\sup_{\gamma\in \caW^l}\sum_{l'=3}^{20}|\caW^{\gamma,l'}|w(l') & \leq \sum_{l'=3}^{20} S_{ww}(l,l')w(l'), \qquad 4\leq l\leq 6 \label{small_ww}\\
	\sup_{\gamma\in \caL^6}\sum_{l'=3}^{20}|\caW^{\gamma,l'}|w(l') & \leq \sum_{l'=3}^{20} S_{\ell w}(6,l')w(l')\,.\label{small_lw}
	\end{align}
\end{lemma}

 For $3\leq l'\leq 20$, there are only a finite number of distinct sets $\caW^{\gamma,l'}$ up to translations and rotations in $\Lambda_{N,K}$, and these distinct sets are independent of the choice of $K>20$ and $N\geq K+53$. As such, the values of $S_{ww}(l,l')$ and $S_{\ell w}(6,l')$ were calculated for any $l,l'$ as above using the code in Appendices~\ref{subsec:saws_col1-4}-\ref{subsec:saws_col5}, and the values are listed in Table~\ref{table:saw}, where for notational convenience, we also define $S_{ww}(3,l) := M_{ww}(3,l)\,.$
 
\begin{proof}
	The proofs for \eqref{small_ww} and \eqref{small_lw} are analogous, so we give the argument for \eqref{small_ww}. It follows from \eqref{S_ww_bdy} and \eqref{Sl3} that
	\begin{equation}
		\sum_{l'=3}^{20}S_{ww}(l,l')w(l') =  \sum_{l'=3}^{19}\left(\left(w(l')-w(l'+1)\right)\sup_{\gamma\in \caW^{l}}\sum_{j=3}^{l'}|\caW^{\gamma, j}|\right)+w(20)\sup_{\gamma\in \caW^{l}} \sum_{j=3}^{20}|\caW^{\gamma,j}| \nonumber
	\end{equation}
Since $w(l)$ is a positive monotone decreasing function, for any fixed $\gamma\in \caW^{l}$ the above implies
\begin{align}
	\sum_{l'=3}^{20}S_{ww}(l,l')w(l') & \geq \sum_{l'=3}^{19}\left(\left(w(l')-w(l'+1)\right)\sum_{j=3}^{l'}|\caW^{\gamma, j}|\right)+w(20)\sum_{j=3}^{20}|\caW^{\gamma,j}| \nonumber \\
	&	 = \sum_{l'=3}^{20} |\caW^{\gamma, j}|w(l')\ \nonumber
\end{align}
from which \eqref{small_ww} follows.
\end{proof}

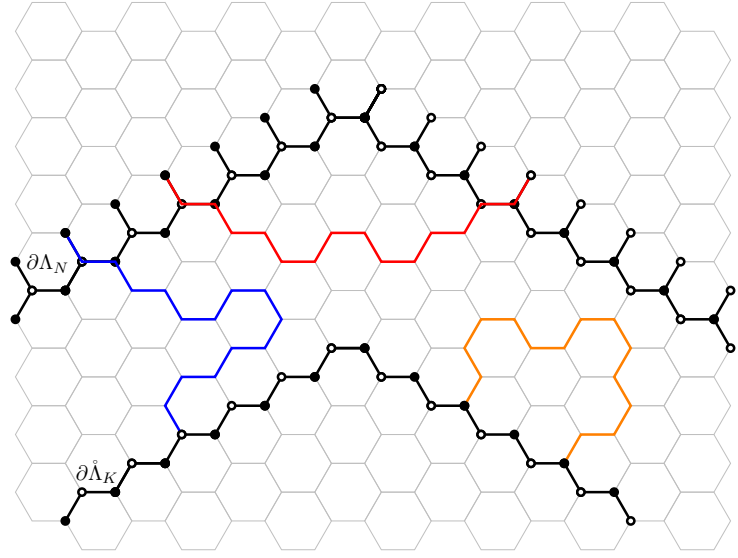
\begin{figure}
	\begin{center}
		\scalebox{.8}{\begin{tikzpicture}
				\def\n{.55};
				
				\tikzset{
					wvertex/.style={circle, draw=black, fill=white, inner sep=1.2pt},
					bvertex/.style={circle, draw=black, fill=black, inner sep=1.2pt}
				}
				
				\newcommand{\hexcoord}[2]
				{[shift=(60:#1),shift=(120:#1),shift=(0:#2),shift=(-60:#2),shift=(0:#2),shift=(60:#2)]}
				\foreach \x in {1,...,9}
				\foreach \y in {2,...,8}{
					\draw[color=black!25]\hexcoord{\x*\n}{\y*\n}
					(0:\n)--(60:\n)--(120:\n)--(180:\n)--(-120:\n)--(-60:\n)--cycle;
					\draw[color=black!25,shift=(-60:\n), shift=(0:\n)]\hexcoord{\x*\n}{\y*\n}
					(0:\n)--(60:\n)--(120:\n)--(180:\n)--(-120:\n)--(-60:\n)--cycle;
				}
				
				\foreach \x in {2,2.5,3,3.5,4,4.5,5}
				{\draw[very thick, color=black]\hexcoord{2*\n+\x*\n}{\x*\n}
					(60:\n) node[bvertex]{} -- (120:\n) node[wvertex]{} -- (180:\n) node[bvertex]{};
					\draw[very thick, color=black]\hexcoord{3*\n+\x*\n}{\x*\n}
					(180:\n) node[bvertex]{} -- (240:\n) node[wvertex]{};
					\draw[very thick, color=black]\hexcoord{9.5*\n-\x*\n}{3.5*\n+\x*\n}
					(300:\n) node[bvertex]{} -- (0:\n) node[wvertex]{};
					\draw[very thick, color=black]\hexcoord{9*\n-\x*\n}{3*\n+\x*\n}
					(120:\n) node[wvertex]{} -- (60:\n) node[bvertex]{} -- (0:\n) node[wvertex]{};
					\draw[very thick, color=black]\hexcoord{9*\n-1*\n}{3*\n+2*\n}
					(0:\n) node[wvertex]{} -- (300:\n) node[bvertex]{};}
				\draw[very thick, color=black]\hexcoord{4.5*\n}{2.5*\n}	(60:\n) node[bvertex]{} -- (120:\n) node[wvertex]{} -- (180:\n) node[bvertex]{} node[above,xshift=-15*\n pt, yshift=7.5*\n pt]{\scalebox{.9}{$\partial\Lambda_N$}};
				\draw[very thick, color=black]\hexcoord{3.5*\n}{8.5*\n}(360:\n) node[wvertex]{} -- (60:\n) node[bvertex]{} -- (120:\n) node[wvertex]{};
				
				\foreach \x in {2.5,3,3.5,4,4.5,5}
				{\draw[very thick, color=black, shift=(240:2*\n), shift=(300:2*\n)]\hexcoord{\x*\n}{\x*\n}
					(60:\n) node[bvertex]{} -- (120:\n) node[wvertex]{} -- (180:\n) node[bvertex]{};
					\draw[very thick, color=black, shift=(180:3*\n), shift=(300:\n), shift=(240:\n)]\hexcoord{6.5*\n-\x*\n}{3.5*\n+\x*\n}				(120:\n) node[wvertex]{} -- (60:\n) node[bvertex]{} -- (0:\n) node[wvertex]{};}
				\draw[very thick, color=black,  shift=(240:2*\n), shift=(300:2*\n)]\hexcoord{3*\n}{3*\n}	(60:\n) node[bvertex]{} -- (120:\n) node[wvertex]{} -- (180:\n) node[bvertex]{} node[above,xshift=-3mm]{\scalebox{.9}{$\partial\mathring{\Lambda}_K$}};
				
				\draw[very thick, color=blue] \hexcoord{5*\n}{2*\n} (60:\n) node[bvertex]{}
				-- ++(300:\n) -- ++(0:\n) -- ++(300:\n) -- ++(0:\n) -- ++(300:\n) -- ++(0:\n)
				-- ++(60:\n) -- ++(0:\n) -- ++(300:\n) -- ++(240:\n) -- ++(180:\n) -- ++(240:\n) --++(180:\n) --++(240:\n) --++(300:\n)
				node[wvertex]{};
				
				
				\draw[very thick, color=orange] \hexcoord{2*\n}{6*\n} (60:\n) node[bvertex]{}
				-- ++(60:\n) -- ++(120:\n) -- ++(60:\n) -- ++(0:\n) -- ++(300:\n) -- ++(0:\n)
				-- ++(60:\n) -- ++(0:\n) -- ++(300:\n) -- ++(240:\n) -- ++(300:\n) -- ++(240:\n) --++(180:\n) --++(240:\n)
				node[bvertex]{};
				
				\draw[very thick, color=red] \hexcoord{6*\n}{3*\n} (60:\n) node[bvertex]{}
				-- ++(300:\n) -- ++(0:\n) -- ++(300:\n) -- ++(0:\n) -- ++(300:\n) -- ++(0:\n)
				-- ++(60:\n) -- ++(0:\n) -- ++(300:\n) -- ++(0:\n) -- ++(60:\n) -- ++(0:\n) --++(60:\n) --++(0:\n) --++(60:\n)
				node[wvertex]{};
				
		\end{tikzpicture}}
	\end{center}
	\caption{Three walks $\gamma\in\mathcal{W}_{N,K}$ with the boundary $\partial\Lambda_{N,K}$ colored to show the bipartition of $\Lambda$. The walk has even length if and only if the endpoints have the same coloring. This illustrates that the parity of $|\gamma|$ only depends on (1) if $\gamma$ crosses a corner of $\Lambda_{N,K}$ and (2) whether or not the $\gamma$ connects the inner and outer boundaries, $\partial\Lambda_N$ and $\partial\mathring{\Lambda}_K$.}
	\label{fig:case1and2_polymers}\end{figure}

We now turn bounding the cardinality of the polymer sets $\caW^{\gamma,l'}$ for the first sum in \eqref{split} (i.e., when $l'$ is even) in the case that the fixed polymer $\gamma\in\caP_{N,K}$ has length $|\gamma|>6$. The case $l'$ odd will be treated numerically and will be discussed further in the proof of Theorem~\ref{thm:KP_condition}. Although similar bounds were given in \cite{kennedy:1988}, the bounds needed here are somewhat different as we work with the ring-shaped domain $\Lambda_{N,K}$ whose non-trivial topology has to be properly taken into account. 

To state this result, note that as we are restricting to pairs $(N,K)$ with $N\geq K+53\geq 78$, the graph distance between any two vertices $x\in \partial\Lambda_N$ and $y\in \partial\mathring{\Lambda}_K$ necessarily satisfy $d(x,y) > 106.$ This can be easily seen from considering the geometry of $\Lambda_{N,K},$ see e.g., Figure~\ref{fig:case1and2_polymers}. Since 
\[\partial\Lambda_{N,K}=\partial\Lambda_N\cup \partial\mathring{\Lambda}_K\] 
and the endpoints of any walk $\gamma'\in\caW^{l'}$ must belong to $\partial\Lambda_{N,K}$, the restriction $l' \leq 20$ implies that such a walk must begin and end on the same boundary, i.e.
\begin{equation} \label{eq:too_small}
	\mathrm{ep}(\gamma') \subseteq \partial \mathring{\Lambda}_K, \quad \text{or} \quad \mathrm{ep}(\gamma') \subseteq \partial \Lambda_N\,.
\end{equation}
Moreover, since $\bbH$ is bipartite, one can see that if $l'$ is even, then $\gamma'$ cannot cross one of the corners of $\Lambda_{N,K}$, see Figure~\ref{fig:case1and2_polymers}. Thus, if we consider $\Lambda_{N,K}$ with a positive counter-clockwise winding angle convention, we can define a consistent labeling the endpoints of any such $\gamma'$ as the left endpoint, $v_L(\gamma')$, and right endpoint, $v_R(\gamma')$. Given this labeling, the following result bounds $M_{ww}(l,l')$ for $l> 6$ and even $4\leq l'\leq 20$.

\begin{table}
	\centering
	\begin{tabular}{ |p{2cm}||p{2cm}|}
		\hline
		\multicolumn{2}{|c|}{$R(l)$ for even $l$ with $4\leq l\leq20$} \\
		\hline
		$l$ & $R(l)$\\
		\hline
		4&1\\
		6&1\\
		8&4\\
		10&9\\
		12&26\\
		14&75\\
		16&215\\
		18&649\\
		20&1943\\
		\hline
	\end{tabular}
	\caption{Maximum number of walks of length $l$ with a fixed right hand endpoint,  $R(l)$.}
	\label{table:R}
\end{table}

\begin{lemma}\label{m3}
	Fix $l>6$ and $\#\in\{\ell,w\}$, and assume $4\leq l'\leq 20$ is even. For any $\Lambda_{N,K}$ with $K\geq 25$ and $N\geq K+53$, the following bound holds:
	\begin{equation}
		M_{\#w}(l,l') \leq \begin{cases}
			l/2+ 1, & l'=4 \\
			l/2 + 2, & l'=6 \\
			(l/2+11l'/4-12)R(l'),  & 8\leq l'\leq 20
		\end{cases}
	\end{equation}
	where
	\begin{equation}\label{defR}
		R(l') = \sup_{v\in \partial \Lambda_{N,K}} |\{\gamma^\prime \in \cW^{l'} : v_R(\gamma^\prime) = v\}|.
	\end{equation}
\end{lemma}

\begin{proof}	
	The cases $l'=4$ and $6$ can be verified by a tedious but simple counting argument that starts from noting that there is a unique walk $\gamma'\in\caW^{l'}$ up to translations and reflections. The number of possible $\gamma'$ that intersect $\gamma$ can then be counted from considering the edges $\gamma$ can share with the translations and reflections of this $\gamma'$.

	Now, fix $8\leq l' \leq 20$ and $\gamma\in\caP^l$ with $l>6$. For any \emph{global} walk $\lambda \in \caW_\Gamma$ contained in $\Lambda_{N,K}$, define
	\be\label{Xgamma}
	X_{l'}(\lambda):= \left\{v\in \partial \Lambda_{N,K} : v= v_R(\gamma') \, \wedge \, \lambda  \nmid \gamma^\prime \; \text{for some}\; \gamma'\in \cW^{l'} \right\}\,.
	\ee
	Considering this and $R(l')$ from \eq{defR}, one can trivially bound
	$$
	M_{ww}(l,l')\leq R(l')\max_{\gamma\in\caW^l}|X_{l'}(\gamma)|.
	$$
	
	To bound $|X_{l'}(\gamma)|$, we decompose $\Lambda_{N,K}$ into three subregions depending on the unique $\delta\in \bN$ satisfying $4\delta\leq l'\leq 4\delta +2$. The first region, called the {\em $\delta$-interior} of $\Lambda_{N,K}$, is
	\be
	B(\delta)=\Lambda_{N-\delta,K+\delta},
	\ee
	and the two other subregions are the connected components of $\Lambda_{N,K}\setminus B(\delta)$ (considered as sets of edges), which we call the inner and outer corridors: \[C_i(\delta)=\mathring{\Lambda}_{K+\delta}\setminus\mathring{\Lambda}_K, \qquad C_o(\delta)=\Lambda_N\setminus \Lambda_{N-\delta}.\] 
	As the endpoints of any $\gamma'\in\caW^{l'}$ must lie on $\partial\Lambda_{N,K}$, $\delta$ is such that $\gamma'$ is contained in either $C_i(\delta)$ or $C_o(\delta)$, and thus the intersection $\gamma\cap\gamma'$ must occur in this corridor.

	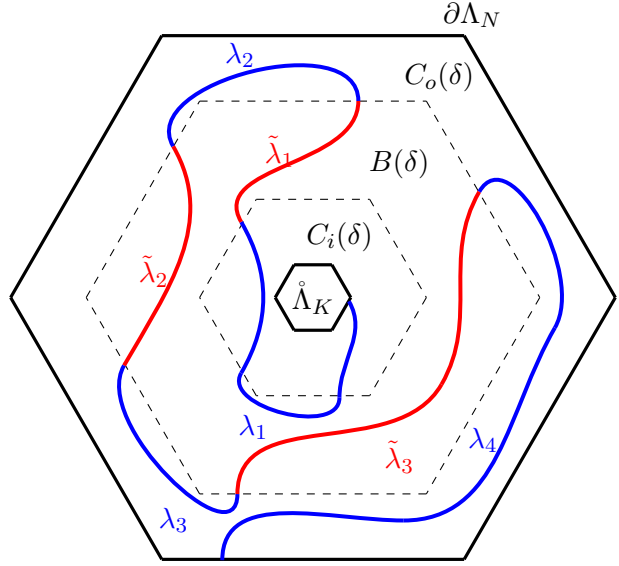
\begin{figure}
		\begin{tikzpicture}
			\def\n{.5};
			\newcommand{\hexcoord}[2]
			{[shift=(60:#1),shift=(120:#1),shift=(0:#2),shift=(-60:#2),shift=(0:#2),shift=(60:#2)]}
			
			\foreach \k in {1,8}
			\foreach \x in {0,60,120,180,240,300}
			\foreach \y in {240}
			{\draw[shift=(\x+120:\k*\n), shift=(\x+60:\k*\n), very thick, color=black]\hexcoord{0}{0}
				(\x+\y:\k*\n)--(\x+\y+60:\k*\n);
			}
			\foreach \k in {3,6}
			\foreach \x in {0,60,120,180,240,300}
			\foreach \y in {240}
			{\draw[shift=(\x+120:\k*\n), shift=(\x+60:\k*\n), dashed, color=black]\hexcoord{0}{0}
				(\x+\y:\k*\n)--(\x+\y+60:\k*\n);
			}
			
			\draw\hexcoord{0}{0} (60:\n)--(120:\n) node[midway, below, yshift=40*\n pt,xshift=20*\n pt] {$C_i(\delta)$};
			\draw\hexcoord{0}{0} (60:\n)--(120:\n) node[midway, below, yshift=95*\n pt,xshift=65*\n pt] {$B(\delta)$};
			\draw\hexcoord{0}{0} (60:\n)--(120:\n) node[midway, below, yshift=160*\n pt,xshift=95*\n pt] {$C_o(\delta)$};
			\draw\hexcoord{0}{0} (60:\n)--(120:\n) node[midway, below, yshift=-5*\n pt,xshift=0*\n pt] {$\mathring{\Lambda}_K$};
			\draw\hexcoord{0}{0} (60:\n)--(120:\n) node[midway, below, yshift=205*\n pt,xshift=120*\n pt] {$\partial\Lambda_N$};
			\draw [line width=1.5pt, color=blue] (.95*\n,-.1*\n)  to[out=-60,in=90] coordinate[pos=.35] (A) (.7*\n,-2.65*\n) ;
			\draw [line width=1.5pt, color=blue] (.7*\n,-2.65*\n)  to[out=-90,in=-120] coordinate[pos=.35] (A) (-1.9*\n,-2*\n) ;
			\draw [line width=1.5pt, color=blue] (-1.9*\n,-2*\n)  to[out=60,in=-60] coordinate[pos=.35] (A) (-1.9*\n,2*\n) ;
			\draw [line width=1.5pt, color=red] (-1.9*\n,2*\n)  to[out=120,in=-90] coordinate[pos=.35] (A) (1.2*\n,5.2*\n) ;
			
			\draw [line width=1.5pt, color=red] (1.2*\n,5.2*\n)  to[out=120,in=-90] coordinate[pos=.35] (A) (1.2*\n,5.2*\n) ;
			\draw [line width=1.5pt, color=blue] (1.2*\n,5.2*\n)  to[out=90,in=120] coordinate[pos=.35] (A) (-3.7*\n,4*\n) ;
			\draw [line width=1.5pt, color=red] (-3.7*\n,4*\n)  to[out=-60,in=60] coordinate[pos=.35] (A) (-5*\n,-1.8*\n) ;
			\draw [line width=1.5pt, color=blue] (-5*\n,-1.8*\n)  to[out=-120,in=-90] coordinate[pos=.35] (A) (-2*\n,-5.2*\n) ;
			\draw [line width=1.5pt, color=red] (-2*\n,-5.2*\n)  to[out=90,in=240] coordinate[pos=.35] (A) (3.4*\n,-2.3*\n) ;
			\draw [line width=1.5pt, color=red] (3.4*\n,-2.3*\n)  to[out=60,in=240] coordinate[pos=.35] (A) (4.4*\n,2.8*\n) ;
			\draw [line width=1.5pt, color=blue] (4.4*\n,2.8*\n)  to[out=60,in=60] coordinate[pos=.35] (A) (6.4*\n,-.8*\n) ;
			\draw [line width=1.5pt, color=blue] (6.4*\n,-.8*\n)  to[out=240,in=0] coordinate[pos=.35] (A) (2.4*\n,-5.9*\n) ;
			\draw [line width=1.5pt, color=blue] (2.4*\n,-5.9*\n)  to[out=180,in=90] coordinate[pos=.35] (A) (-2.4*\n,-6.95*\n) ;
			\draw\hexcoord{0}{0} (60:\n)--(120:\n) node[midway, below, yshift=-105*\n pt,xshift=-45*\n pt,color=blue] {$\lambda_1$};
			\draw\hexcoord{0}{0} (60:\n)--(120:\n) node[midway, below,
			yshift=105*\n pt,xshift=-25*\n pt,color=red] {$\tilde{\lambda}_1$};
			\draw\hexcoord{0}{0} (60:\n)--(120:\n) node[midway, below, yshift=15*\n pt,xshift=-120*\n pt,color=red] {$\tilde{\lambda}_2$};
			\draw\hexcoord{0}{0} (60:\n)--(120:\n) node[midway, below,
			yshift=175*\n pt,xshift=-55*\n pt,color=blue] {${\lambda}_2$};
			\draw\hexcoord{0}{0} (60:\n)--(120:\n) node[midway, below, yshift=-175*\n pt,xshift=-105*\n pt,color=blue] {${\lambda}_3$};
			\draw\hexcoord{0}{0} (60:\n)--(120:\n) node[midway, below, yshift=-125*\n pt,xshift=65*\n pt,color=red] {$\tilde{\lambda}_3$};
			\draw\hexcoord{0}{0} (60:\n)--(120:\n) node[midway, below, yshift=-115*\n pt,xshift=128*\n pt,color=blue] {$\lambda_4$};
		\end{tikzpicture}
		\caption{
			An illustration of the regions $C_i(\delta)$, $C_o(\delta)$, and $B(\delta)$ and the decomposition of a path $\gamma$ into pieces $\gamma=\lambda_1\vee\tilde{\lambda}_1\vee\dots\tilde{\lambda}_3\vee\lambda_4$. Note that $\lambda_i$ has an excursion into $B(\delta)$ of with length less than $\ell_0$.}\label{fig:decomposition}
	\end{figure}

	To proceed, we define a cutoff length 
	\[\ell_0 =\frac{11l'}{2}-24\] 
	and decompose $\gamma$ in disjoint segments $\gamma = \lambda_1\vee \tilde \lambda_1\vee\ \lambda_2\vee\tilde\lambda_2 \vee\cdots\vee\tilde\lambda_{p-1}\vee\ \lambda_p$ as follows: enumerate by $\tilde{\lambda}_1, \ldots \tilde{\lambda}_{p-1}$ the connected components of $\gamma\cap B(\delta)$ whose length is at least $\ell_0$, and all connected components of $\gamma$ in between these by $\lambda_1, \ldots \lambda_{p}$, see Figure~\ref{fig:decomposition}. In the case that $\gamma$ is a loop, it may be that the number of $\tilde{\lambda}_j$ and $\lambda_k$ terms are the same, in which case we use the convention that $\lambda_p=\emptyset$. Regardless, the bounds proved below still hold. 
	
	The constraint $N-K\geq53\geq \ell_0/2+2\delta$ guarantees that each $\lambda_j$ intersects only one of the corridors, $C_i(\delta)$ or $C_o(\delta)$. Moreover, if $\gamma'\nmid \gamma$ it has to intersect $\lambda_j$ for some $j$. Thus, if we show
	\be\label{Xbound}
	|X_{l'}(\lambda_j)| \leq |\lambda_j|/2+\ell_0/2
	\ee
	for all $j$, then one has
	\begin{equation}\label{Xl_gamma}
		|X_{l'} (\gamma)|  \leq \sum_{j=1}^p |X_{l'}(\lambda_j)| \leq \frac{1}{2}\sum_{j=1}^p \vert\lambda_j\vert +\frac{p\ell_0}{2}
	\end{equation}
	from which the result follows since
	$$
	|\gamma|=\sum_{j=1}^p  |\lambda_j|  + \sum_{j=1}^{p-1}   |\tilde\lambda_j| \geq \sum_{j=1}^p  |\lambda_j| +(p-1) \ell_0\,.
	$$
	We estimate $|X_{l'}(\lambda_j)|$ by considering separately the cases that $\lambda_j$ intersects the inner and outer corridor. In each case, we partition the argument by the length of $\lambda_j$.
\begin{figure}
			\scalebox{.8}{\begin{tikzpicture}
						\def\n{.3};
						\newcommand{\hexcoord}[2]
						{[shift=(60:#1),shift=(120:#1),shift=(0:#2),shift=(-60:#2),shift=(0:#2),shift=(60:#2)]}
						\foreach \j in {0,1,2,3,4}
						\foreach \k in {-5,5}
						\foreach \x in {0,60,120,180,240,300}
						\foreach \y in {240,300}
						\draw[shift=(\x+120:\n), shift=(\x+60:\n),thick, color=gray]\hexcoord{\j*\n}{\k*\n}
						(\x+\y:\n)--(\x+\y+60:\n);
						\foreach \h in{0,0.5,1,1.5,2,2.5,3,3.5,4,4.5,5}
						\foreach \j in {3-\h}
						\foreach \k in {\h,-\h}
						\foreach \x in {0,60,120,180,240,300}
						\foreach \y in {240,300}
						\draw[shift=(\x+120:\n), shift=(\x+60:\n), thick,color=gray]\hexcoord{\j*\n}{\k*\n}
						(\x+\y:\n)--(\x+\y+60:\n);
						\foreach \j in {0,1,2,3}
									\foreach \k in {-4,4}
									\foreach \x in {0,60,120,180,240,300}
									\foreach \y in {240,300}
									\draw[shift=(\x+120:\n), shift=(\x+60:\n),thick, color=gray]\hexcoord{\j*\n}{\k*\n}
									(\x+\y:\n)--(\x+\y+60:\n);
									\foreach \h in{0,0.5,1,1.5,2,2.5,3,3.5,4,4.5,5,5}
									\foreach \j in {2.4-\h}
									\foreach \k in {\h,-\h}
									\foreach \x in {0,60,120,180,240,300}
									\foreach \y in {240,300}
									\draw[shift=(\x+120:\n), shift=(\x+60:\n),thick, color=gray]\hexcoord{\j*\n}{\k*\n}
									(\x+\y:\n)--(\x+\y+60:\n);

									\foreach \k in {-3,-2.5,-2,-1.5,-1,-0.5,0}
																														\foreach \j in {8*\n+\k}
																														\draw[very thick, color=gray, shift=(240:\n), shift=(300:\n)]\hexcoord{\j*\n}{\k*\n} (0:\n)--(300:\n)--(240:\n);
														\foreach \k in {-3,-2.5,-2,-1.5,-1,-0.5,0}
																																			\foreach \j in {8*\n+\k}
																																			\draw[very thick, color=gray, shift=(240:\n), shift=(300:\n)]\hexcoord{\j*\n}{-\k*\n} (180:\n)--(240:\n)--(300:\n);
												\foreach \k in {-2,-1,0,1,2}
																																						\foreach \j in {2.5*\n+\k}
																																						\draw[very thick, color=gray, shift=(240:\n), shift=(300:\n)]\hexcoord{\j*\n}{3.5*\n} (120:\n)--(180:\n)--(240:\n);
																																							\foreach \k in {-1,0,1,2}
																																																																			\foreach \j in {2.5*\n+\k}
																																																																			\draw[very thick, color=gray, shift=(240:\n), shift=(300:\n)]\hexcoord{\j*\n}{-3.5*\n} (300:\n)--(0:\n)--(60:\n);
									\draw[very thick, color=gray, shift=(240:\n), shift=(300:\n)]\hexcoord{(2.5*\n-2*\n)*\n}{-3.5*\n} (0:\n)--(60:\n);																													
									\foreach \k in {0}
									\foreach \j in {8,9,10,11}
									\draw[very thick, color=blue, shift=(240:\n), shift=(300:\n)]\hexcoord{\j*\n}{\k*\n} (120:\n)--(180:\n)--(240:\n) (300:\n)--(0:\n)--(60:\n);
									\foreach \k in {0}
																		\foreach \j in {12}
																		\draw[very thick, color=blue, shift=(240:\n), shift=(300:\n)]\hexcoord{\j*\n}{\k*\n} (240:\n)--(180:\n) (300:\n)--(0:\n);

									\foreach \k in {-3.5,-4,-4.5,-5}
																			\foreach \j in {\n-\k}
																			\draw[very thick, color=blue, shift=(240:\n), shift=(300:\n)]\hexcoord{\j*\n}{\k*\n} (0:\n)--(60:\n)--(120:\n) (180:\n)--(240:\n)--(300:\n);
																								
					\foreach \k in {-5.5}
																								\foreach \j in {\n-\k}
																								\draw[very thick, color=blue, shift=(240:\n), shift=(300:\n)]\hexcoord{\j*\n}{\k*\n} (240:\n)--(300:\n)(0:\n)--(60:\n);

									\foreach \k in {3.5,4,4.5,5}
													\foreach \j in {\n+\k}
													\draw[very thick, color=blue, shift=(240:\n), shift=(300:\n)]\hexcoord{\j*\n}{\k*\n} (60:\n)--(120:\n)--(180:\n) (240:\n)--(300:\n)--(0:\n);
													
						\foreach \k in {5.5}
																		\foreach \j in {\n+\k}
																		\draw[very thick, color=blue, shift=(240:\n), shift=(300:\n)]\hexcoord{\j*\n}{\k*\n}  (240:\n)--(300:\n)(120:\n)--(180:\n);								
			
								
								\foreach \k in {-5.5,-5,-4.5,-4,-3.5,-3,-2.5,-2,-1.5,-1,-0.5,0}
																								\foreach \j in {12*\n+\k}
																								\draw[ thick, color=gray, shift=(240:\n), shift=(300:\n)]\hexcoord{\j*\n}{\k*\n} (60:\n)--(120:\n)--(180:\n);
								\foreach \k in {-5.5,-5,-4.5,-4,-3.5,-3,-2.5,-2,-1.5,-1,-0.5,0}
																													\foreach \j in {12*\n+\k}
																													\draw[ thick, color=gray, shift=(240:\n), shift=(300:\n)]\hexcoord{\j*\n}{-\k*\n} (0:\n)--(60:\n)--(120:\n);
						\foreach \k in {-2,-1,0,1,2,3,4}
																																\foreach \j in {2.5*\n+\k}
																																\draw[thick, color=gray, shift=(240:\n), shift=(300:\n)]\hexcoord{\j*\n}{5.5*\n} (300:\n)--(0:\n)--(60:\n);
																																	\foreach \k in {-1,0,1,2,3,4}
																																																													\foreach \j in {2*\n+\k}
																																																													\draw[ thick, color=gray, shift=(240:\n), shift=(300:\n)]\hexcoord{\j*\n}{-6*\n} (300:\n)--(0:\n)--(60:\n);

																																																																																											\draw[shift=(120:\n), shift=(60:\n), fill=gray]\hexcoord{8.5*\n}{5*\n} (300:8*\n) node[red]{$C_{2,3}$};
																																																																														\draw[shift=(120:\n), shift=(60:\n), fill=gray]\hexcoord{9*\n}{-7.5*\n} (300:8*\n) node[red]{$C_{6,1}$};																																									\draw[shift=(120:\n), shift=(60:\n), fill=gray]\hexcoord{15.25*\n}{-1.25*\n} (300:8*\n) node[red]{$C_{1,2}$};
																																														
																			\draw[shift=(120:\n), shift=(60:\n), fill=gray]\hexcoord{12.5*\n}{-4.25*\n} (300:8*\n) node[black]{$T_{1}$};											
																			\draw[shift=(120:\n), shift=(60:\n), fill=gray]\hexcoord{12.5*\n}{2.25*\n} (300:8*\n) node[black]{$T_{2}$};																																												\draw[shift=(120:\n), shift=(60:\n), fill=gray]\hexcoord{3.5*\n}{-7.5*\n} (300:8*\n) node[black]{$T_{6}$};																													\draw[shift=(120:\n), shift=(60:\n), fill=gray]\hexcoord{3.5*\n}{5*\n} (300:8*\n) node[black]{$T_{3}$};																											
			\foreach \k in {-5,-4.5,-4,-3.5,-3,-2.5,-2,-1.5,-1}
																													\foreach \j in {13*\n+\k}
																													\draw[very thick, color=blue, shift=(240:\n), shift=(300:\n)]\hexcoord{\j*\n}{\k*\n} (0:\n)--(300:\n)--(240:\n);
		\foreach \k in {-5,-4.5,-4,-3.5,-3,-2.5,-2,-1.5,-1}
																																		\foreach \j in {13*\n+\k}
																																		\draw[very thick, color=blue, shift=(240:\n), shift=(300:\n)]\hexcoord{\j*\n}{-\k*\n} (180:\n)--(240:\n)--(300:\n);
							\foreach \k in {-0.5}																										
																																								\foreach \j in {13*\n+\k}
																																								\draw[very thick, color=blue, shift=(240:\n), shift=(300:\n)]\hexcoord{\j*\n}{-\k*\n} (240:\n)--(300:\n);																	\foreach \k in {-5.5}																										
																																																																																\foreach \j in {13*\n+\k}
																																																																																\draw[very thick, color=blue, shift=(240:\n), shift=(300:\n)]\hexcoord{\j*\n}{-\k*\n} (180:\n)--(240:\n);								
																																																																																			\foreach \k in {-.5}
																																																																																																													\foreach \j in {13*\n+\k}
																																																																																																													\draw[very thick, color=blue, shift=(240:\n), shift=(300:\n)]\hexcoord{\j*\n}{\k*\n} (300:\n)--(240:\n);		
																																																																																																																\foreach \k in {-5.5}
																																																																																																																																										\foreach \j in {13*\n+\k}
																																																																																																																																										\draw[very thick, color=blue, shift=(240:\n), shift=(300:\n)]\hexcoord{\j*\n}{\k*\n} (0:\n)--(300:\n);								
		\foreach \k in {-1.5,-.5,0.5,1.5,2.5}
																																					\foreach \j in {2.5*\n+\k}
																																					\draw[very thick, color=blue, shift=(240:\n), shift=(300:\n)]\hexcoord{\j*\n}{6*\n} (120:\n)--(180:\n)--(240:\n);
			\foreach \k in {-2.5}
																																						\foreach \j in {2.5*\n+\k}
																																						\draw[very thick, color=blue, shift=(240:\n), shift=(300:\n)]\hexcoord{\j*\n}{6*\n} (120:\n)--(180:\n);			\foreach \k in {3.5}
																																																																												\foreach \j in {2.5*\n+\k}
																																																																												\draw[very thick, color=blue, shift=(240:\n), shift=(300:\n)]\hexcoord{\j*\n}{6*\n} (180:\n)--(240:\n);																																			
				\foreach \k in {-1,0,1,2,3}
																																							\foreach \j in {2.5*\n+\k}
																																							\draw[very thick, color=blue, shift=(240:\n), shift=(300:\n)]\hexcoord{\j*\n}{-5.5*\n} (120:\n)--(180:\n)--(240:\n);				\foreach \k in {-2}
																																																																														\foreach \j in {2.5*\n+\k}
																																																																														\draw[very thick, color=blue, shift=(240:\n), shift=(300:\n)]\hexcoord{\j*\n}{-5.5*\n} (120:\n)--(180:\n);																																	
	%
	%
	%
		\foreach \k in {-3,-2.5,-2,-1.5,-1,-0.5}
																												\foreach \j in {8*\n+\k}
																												\draw[very thick, color=blue, shift=(240:\n), shift=(300:\n)]\hexcoord{\j*\n}{\k*\n} (0:\n)--(300:\n)--(240:\n);
												\foreach \k in {-3,-2.5,-2,-1.5,-1,-0.5}
																																	\foreach \j in {8*\n+\k}
																																	\draw[very thick, color=blue, shift=(240:\n), shift=(300:\n)]\hexcoord{\j*\n}{-\k*\n} (180:\n)--(240:\n)--(300:\n);
																												
										\foreach \k in {-2,-1,0,1}
																																				\foreach \j in {2.5*\n+\k}
																																				\draw[very thick, color=blue, shift=(240:\n), shift=(300:\n)]\hexcoord{\j*\n}{3.5*\n} (120:\n)--(180:\n)--(240:\n);
																																			
																																					\foreach \k in {-1,0,1}
																																																																	\foreach \j in {2.5*\n+\k}
																																																																	\draw[very thick, color=blue, shift=(240:\n), shift=(300:\n)]\hexcoord{\j*\n}{-3.5*\n} (300:\n)--(0:\n)--(60:\n);
							\draw[very thick, color=blue, shift=(240:\n), shift=(300:\n)]\hexcoord{(2.5*\n-2*\n)*\n}{-3.5*\n} (0:\n)--(60:\n);			
								\foreach \k in {0}
																												\foreach \j in {8,9,10,11,12}
																												\draw[very thick, color=red, shift=(240:\n), shift=(300:\n)]\hexcoord{\j*\n}{\k*\n} (240:\n)--(300:\n) (120:\n)--(60:\n);
																												
																												\foreach \k in {-3.5,-4,-4.5,-5,-5.5}
																																																\foreach \j in {\n-\k}
																																										\draw[very thick, color=red, shift=(240:\n), shift=(300:\n)]\hexcoord{\j*\n}{\k*\n} (120:\n)--(180:\n) (0:\n)--(300:\n);
																										\foreach \k in {3.5,4,4.5,5,5.5}
																																				\foreach \j in {\n+\k}
																														\draw[very thick, color=red, shift=(240:\n), shift=(300:\n)]\hexcoord{\j*\n}{\k*\n} (240:\n)--(180:\n) (0:\n)--(60:\n);				
			\end{tikzpicture}}
			\caption{The labeled corners $C_{i,i+1}\subset C_i(\delta)$ in red and the trapezoids $T_i\subset C_i(\delta)$ in gray with trapezoidal boundaries in blue. \label{fig:corners}}
		\end{figure}
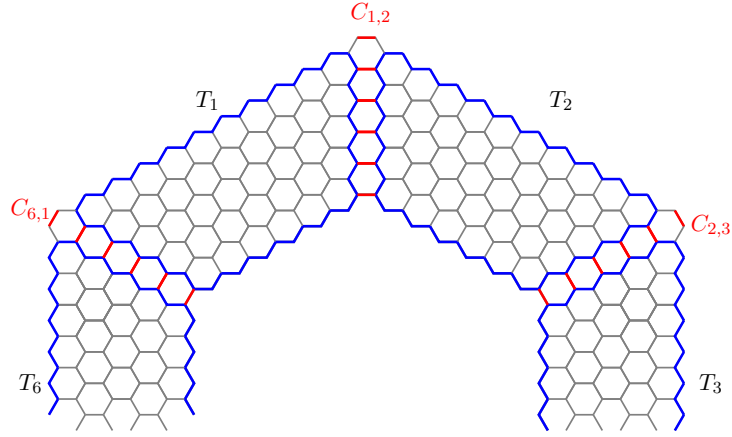
		\par
	\textbf{Case $\lambda_j\cap C_i(\delta)\neq \emptyset$:} First note that $X_{l'}(\lambda_j)\subseteq \partial\mathring{\Lambda}_K$ in this corridor. Moreover, $C_i(\delta)$ is the union of six trapezoids $T_1, \ldots, T_6$ (viewed as a collection of edges) along with edges that connect two neighboring trapezoids, see Figure~\ref{fig:shadedregion}. We call the set of all edges connecting  two neighboring trapezoids, $T_i,$ $T_j$, a \emph{corner}, denoted $C_{i,j}$. A walk $\gamma'\in \caW^{l'}$ whose endpoints belong to two neighboring trapezoids necessarily has odd length. As we only consider even $8\leq l'\leq 20$, the endpoints of such a walk must belong to the same trapezoid. Given the winding convention of our volume, this implies there are six sites of  $ \partial\mathring{\Lambda}_K$ next to the corners that are not admissible right vertices. If $|\lambda_j|\geq 12K-2l'+10$, then since $|\partial\mathring{\Lambda}_K|=6K$, this implies
	\[
	|X_{l'}(\lambda_j)| \leq |\mathring{\Lambda}_K|-6 \leq \frac{|\lambda_j|}{2}+l'-11 < \frac{|\lambda_j|}{2}+\frac{\ell_0}{2}.
	\]

	Now, assume  $|\lambda_j|\leq 12K-2l'+9$. Let 
	\[\mathcal{T}_{\lambda_j}=T_{i_1}'\cup C_{i_1,i_{2}}\cup T_{i_2}'\ldots \cup C_{i_{m-1}, i_m} \cup T_{i_m}',\qquad 1\leq m \leq 7\] 
	be the smallest contiguous region comprised of trapezoids and corners that contains $\lambda_j$, where $T_{i_m}'\subseteq T_{j}$ has the same height as $T_{j}$ and its legs are translations of the legs of $T_{j}$, see Figure~\ref{fig:shadedregion}. Recall that $\gamma^{(M)}$ is the shortest loop containing $\mathring{\Lambda}_M$ for any $M\geq 1$. Since $\gamma^{(K+1)}$ is the smallest such loop contained in $C_i(\delta)$ and $\caT_{\lambda_j}$ is the smallest such volume containing $\lambda_j$, it must be that 
	\[|\lambda_j|\geq |\alpha_{K+1}|\quad \text{where}\quad\alpha_{K+1}=\gamma^{(K+1)}\cap \mathcal{T}_{\lambda_j}\in\caW_\Gamma\] 
	and the intersection is understood with respect to their edge sets. Let $\beta_K\in\caW_\Gamma$ be the arc of $\gamma^{(K)}$ enclosed by $\mathcal{T}_{\lambda_j}$ plus the the first $l'-2$ edges extending out from $\mathcal{T}_{\lambda_j}$, see Figure~\ref{fig:shadedregion}. As $| \gamma^{(M)}\cap \mathcal{T}_{\lambda_j}|$ is a increasing function of $M$, it follows that
	\[|\beta_K| \leq |\alpha_{K+1}|+2(l'-2)\leq|\lambda_j|+2(l'-2).\]
	Any walk $\gamma'\in\caW^{l'}$ that intersects $\lambda_j$ must intersect $\mathcal{T}_{\lambda_j}$, and so the endpoints of $\gamma'$ must belong to $\caV_{\beta_K}\cap\partial\mathring{\Lambda}_K$. The left most such vertex cannot be the right endpoint of any such $\gamma'$, and any pair of admissible right vertices from this set are separated by at least two edges of $\beta_K$. Thus,
	\[
	|X_{l'}(\lambda_j)| \leq \frac{|\beta_K|}{2} \leq \frac{\lambda_j}{2}+l'-2 < \frac{\lambda_j}{2}+ \frac{\ell_0}{2}.
	\]
	
	\textbf{Case $\lambda_j\cap C_o(\delta)\neq \emptyset$:}  Notice that $X_{l'}(\lambda_j)\subseteq \partial\Lambda_N$ in this corridor. As in the interior corridor case, $C_o(\delta)$ is the disjoint union of six trapezoids and the corners between them, and there are six vertices of $\partial\Lambda_N$ next to the corners that are not admissible right vertices. Thus, if $|\lambda_j|\geq 12N - 9l'/2 + 8$,
	\begin{equation}
		|X_{l'}(\lambda_j)| \leq |\partial\Lambda_N|-6  \leq \frac{|\lambda_j|}{2} + \frac{9l'}{4} - 10 < \frac{|\lambda_j|}{2}+\frac{\ell_0}{2}.
	\end{equation}
	where we use that $|\partial\Lambda_N|-6 = 6N-6.$
\begin{figure}
		\scalebox{.85}{\begin{tikzpicture}
			\def\n{.3};
			\newcommand{\hexcoord}[2]
			{[shift=(60:#1),shift=(120:#1),shift=(0:#2),shift=(-60:#2),shift=(0:#2),shift=(60:#2)]}
				\foreach \x in {1.5}
						\foreach \y in {-0.5*\n,-1*\n,-1.5*\n,0,0.5*\n,1*\n,1.5*\n,2*\n,2.5*\n,3*\n,3.5*\n}{
							\filldraw[fill=gray!25]\hexcoord{(\x+\y+\n)*\n}{(\y-2*\n)*\n}
							(0:\n)--(60:\n)--(120:\n)--(180:\n)--(-120:\n)--(-60:\n)--cycle;
							\filldraw[fill=gray!25,shift=(-60:\n), shift=(0:\n)]\hexcoord{(\x+\y+\n)*\n}{(\y-2*\n)*\n}
							(0:\n)--(60:\n)--(120:\n)--(180:\n)--(-120:\n)--(-60:\n)--cycle;
							\filldraw[fill=gray!25,shift=(-60:\n), shift=(0:\n)]\hexcoord{(\x+\y+2*\n)*\n}{(\y-3*\n)*\n}
														(0:\n)--(60:\n)--(120:\n)--(180:\n)--(-120:\n)--(-60:\n)--cycle;
							\filldraw[fill=gray!25,shift=(-60:\n), shift=(0:\n)]\hexcoord{(\x+\y+3*\n)*\n}{(\y-4*\n)*\n}
																					(0:\n)--(60:\n)--(120:\n)--(180:\n)--(-120:\n)--(-60:\n)--cycle;
							\filldraw[fill=gray!25,shift=(-60:\n), shift=(0:\n)]\hexcoord{(\x+\y+2.5*\n)*\n}{(\y-3.5*\n)*\n}
																					(0:\n)--(60:\n)--(120:\n)--(180:\n)--(-120:\n)--(-60:\n)--cycle;
						}
						\foreach \x in {1.5}
												\foreach \y in {-\n,-0.5*\n,0,0.5*\n,1*\n,1.5*\n,2*\n,2.5*\n,3*\n,3.5*\n,4*\n,4.5*\n}{
													\filldraw[fill=gray!25]\hexcoord{(\x-\y+\n)*\n}{(\y+\n)*\n}
													(0:\n)--(60:\n)--(120:\n)--(180:\n)--(-120:\n)--(-60:\n)--cycle;
\filldraw[fill=gray!25]\hexcoord{(\x-\y+2*\n)*\n}{(\y+\n)*\n}
													(0:\n)--(60:\n)--(120:\n)--(180:\n)--(-120:\n)--(-60:\n)--cycle;
\filldraw[fill=gray!25]\hexcoord{(\x-\y+3*\n)*\n}{(\y+\n)*\n}
													(0:\n)--(60:\n)--(120:\n)--(180:\n)--(-120:\n)--(-60:\n)--cycle;
\filldraw[fill=gray!25]\hexcoord{(\x-\y+4*\n)*\n}{(\y+\n)*\n}
													(0:\n)--(60:\n)--(120:\n)--(180:\n)--(-120:\n)--(-60:\n)--cycle;
\filldraw[fill=gray!25]\hexcoord{(\x-\y+5*\n)*\n}{(\y+\n)*\n}
													(0:\n)--(60:\n)--(120:\n)--(180:\n)--(-120:\n)--(-60:\n)--cycle;
												}
			\foreach \j in {0,1,2,3,4}
			\foreach \k in {-5,5}
			\foreach \x in {0,60,120,180,240,300}
			\foreach \y in {240,300}
			\draw[shift=(\x+120:\n), shift=(\x+60:\n),thick, color=gray]\hexcoord{\j*\n}{\k*\n}
			(\x+\y:\n)--(\x+\y+60:\n);
			\foreach \h in{0,0.5,1,1.5,2,2.5,3,3.5,4,4.5,5}
			\foreach \j in {3-\h}
			\foreach \k in {\h,-\h}
			\foreach \x in {0,60,120,180,240,300}
			\foreach \y in {240,300}
			\draw[shift=(\x+120:\n), shift=(\x+60:\n), thick,color=gray]\hexcoord{\j*\n}{\k*\n}
			(\x+\y:\n)--(\x+\y+60:\n);
			\foreach \j in {0,1,2,3}
						\foreach \k in {-4,4}
						\foreach \x in {0,60,120,180,240,300}
						\foreach \y in {240,300}
						\draw[shift=(\x+120:\n), shift=(\x+60:\n),thick, color=gray]\hexcoord{\j*\n}{\k*\n}
						(\x+\y:\n)--(\x+\y+60:\n);
						\foreach \h in{0,0.5,1,1.5,2,2.5,3,3.5,4,4.5,5,5}
						\foreach \j in {2.4-\h}
						\foreach \k in {\h,-\h}
						\foreach \x in {0,60,120,180,240,300}
						\foreach \y in {240,300}
						\draw[shift=(\x+120:\n), shift=(\x+60:\n),thick, color=gray]\hexcoord{\j*\n}{\k*\n}
						(\x+\y:\n)--(\x+\y+60:\n);
												
%
						\foreach \k in {-3,-2.5,-2,-1.5,-1,-0.5}
																											\foreach \j in {8*\n+\k}
																											\draw[very thick, color=Green, shift=(240:\n), shift=(300:\n)]\hexcoord{\j*\n}{\k*\n} (0:\n)--(300:\n)--(240:\n);
											\foreach \k in {-3,-2.5,-2,-1.5,-1,-0.5}
																																\foreach \j in {8*\n+\k}
																																\draw[very thick, color=Green, shift=(240:\n), shift=(300:\n)]\hexcoord{\j*\n}{-\k*\n} (180:\n)--(240:\n)--(300:\n);
		\foreach \k in {0}
																																		\foreach \j in {8*\n+\k}
																																		\draw[very thick, color=Green, shift=(240:\n), shift=(300:\n)]\hexcoord{\j*\n}{-\k*\n} (240:\n)--(300:\n);																														
									\foreach \k in {-2,-1,0,1}
																																			\foreach \j in {2.5*\n+\k}
																																			\draw[very thick, color=Green, shift=(240:\n), shift=(300:\n)]\hexcoord{\j*\n}{3.5*\n} (120:\n)--(180:\n)--(240:\n);
																																			\foreach \k in {2}
																																																																						\foreach \j in {2.5*\n+\k}
																																																																						\draw[very thick, color=Green, shift=(240:\n), shift=(300:\n)]\hexcoord{\j*\n}{3.5*\n} (180:\n)--(240:\n);
																																				\foreach \k in {-1,0,1}
																																																																\foreach \j in {2.5*\n+\k}
																																																																\draw[very thick, color=Green, shift=(240:\n), shift=(300:\n)]\hexcoord{\j*\n}{-3.5*\n} (300:\n)--(0:\n)--(60:\n);
						\draw[very thick, color=Green, shift=(240:\n), shift=(300:\n)]\hexcoord{(2.5*\n-2*\n)*\n}{-3.5*\n} (0:\n)--(60:\n);
						\foreach \k in {2}
																																																																						\foreach \j in {2.5*\n+\k}
																																																																						\draw[very thick, color=Green, shift=(240:\n), shift=(300:\n)]\hexcoord{\j*\n}{-3.5*\n} (300:\n)--(0:\n);
												\draw[very thick, color=Green, shift=(240:\n), shift=(300:\n)]\hexcoord{(2.5*\n-2*\n)*\n}{-3.5*\n} (0:\n)--(60:\n);																																								

%

%

%

					
					\foreach \k in {-5.5,-5,-4.5,-4,-3.5,-3,-2.5,-2,-1.5,-1,-0.5,0}
																					\foreach \j in {12*\n+\k}
																					\draw[ thick, color=gray, shift=(240:\n), shift=(300:\n)]\hexcoord{\j*\n}{\k*\n} (60:\n)--(120:\n)--(180:\n);
					\foreach \k in {-5.5,-5,-4.5,-4,-3.5,-3,-2.5,-2,-1.5,-1,-0.5,0}
																										\foreach \j in {12*\n+\k}
																										\draw[ thick, color=gray, shift=(240:\n), shift=(300:\n)]\hexcoord{\j*\n}{-\k*\n} (0:\n)--(60:\n)--(120:\n);
			\foreach \k in {-2,-1,0,1,2,3,4}
																													\foreach \j in {2.5*\n+\k}
																													\draw[thick, color=gray, shift=(240:\n), shift=(300:\n)]\hexcoord{\j*\n}{5.5*\n} (300:\n)--(0:\n)--(60:\n);
																														\foreach \k in {-1,0,1,2,3,4}
																																																										\foreach \j in {2*\n+\k}
																																																										\draw[ thick, color=gray, shift=(240:\n), shift=(300:\n)]\hexcoord{\j*\n}{-6*\n} (300:\n)--(0:\n)--(60:\n);
																																																										
																																																\draw[shift=(120:\n), shift=(60:\n), fill=gray]\hexcoord{7*\n}{-3*\n} (300:8*\n) node[Green]{$\beta_{K}$};														
\draw[shift=(120:\n), shift=(60:\n), fill=gray]\hexcoord{4.5*\n}{-4.5*\n} (300:8*\n) node[Green]{\small $\left.
        \begin{array}{cc}
                \\ \\ \\ \\ \\
    	\end{array}
    \right\}$};
    \draw[shift=(120:\n), shift=(60:\n), fill=gray]\hexcoord{4.5*\n}{-3.5*\n} (300:8*\n) node[Green]{\small $l'-2$};

							\foreach \k in {-3,-2.5,-2,-1.5,-1,-.5,0}
																							\foreach \j in {8*\n+\k}
																							\draw[very thick, color=blue, shift=(240:\n), shift=(300:\n)]\hexcoord{\j*\n}{\k*\n} (60:\n)--(120:\n)--(180:\n);																																
					\foreach \k in {-3.5,-3,-2.5,-2,-1.5,-1,-0.5,0}
																										\foreach \j in {8*\n+\k}
																										\draw[very thick, color=blue, shift=(240:\n), shift=(300:\n)]\hexcoord{\j*\n}{-\k*\n} (0:\n)--(60:\n)--(120:\n);
					\foreach \k in {-3.5}
																															\foreach \j in {8*\n+\k}
																															\draw[very thick, color=blue, shift=(240:\n), shift=(300:\n)]\hexcoord{\j*\n}{-\k*\n} (300:\n)--(0:\n);	
								\foreach \k in {-3.5,-3,-2.5,-2,-1.5,-1}
																								\foreach \j in {9*\n+\k}
																								\draw[very thick, color=red, shift=(240:\n), shift=(300:\n)]\hexcoord{\j*\n}{\k*\n} (60:\n)--(120:\n)--(180:\n);	
								\foreach \k in {-4}
																															\foreach \j in {9*\n+\k}
																															\draw[very thick, color=red, shift=(240:\n), shift=(300:\n)]\hexcoord{\j*\n}{\k*\n} (60:\n)--(120:\n);																																																								\foreach \k in {5.5,5,4.5}
																																																\foreach \j in {8*\n-\k}
																																																\draw[very thick, color=red, shift=(240:\n), shift=(300:\n)]\hexcoord{\j*\n}{\k*\n} (180:\n)--(240:\n)--(300:\n);																								
						\foreach \k in {-4,-3.5,-3,-2.5,-2,-1.5}
																											\foreach \j in {9*\n+\k}
																											\draw[very thick, color=red, shift=(240:\n), shift=(300:\n)]\hexcoord{\j*\n}{-\k*\n} (0:\n)--(60:\n)--(120:\n);
																			  \draw[shift=(120:\n), shift=(60:\n), fill=gray]\hexcoord{7*\n}{0*\n} (300:8*\n) node[blue]{\small $\alpha_{K+1}$};																																		
																																															\foreach \k in {1.5}
																																																																	\foreach \j in {8.5,9.5,10.5,11.5,12.5,13.5}
																																																																	\draw[very thick, color=red, shift=(240:\n), shift=(300:\n)]\hexcoord{\j*\n}{\k*\n} (120:\n)--(180:\n)--(240:\n);	
																																															\foreach \k in {-1}
																																																																																																																\foreach \j in {9,10,11,12}
																																																																																																																\draw[very thick, color=red, shift=(240:\n), shift=(300:\n)]\hexcoord{\j*\n}{\k*\n} (300:\n)--(0:\n)--(60:\n);	
																																																																																										\foreach \k in {-0.5,0,0.5,1}
																																																																																																																		\foreach \j in {13*\n+\k}
																																																																																																																		\draw[very thick, color=red, shift=(240:\n), shift=(300:\n)]\hexcoord{\j*\n}{\k*\n} (60:\n)--(120:\n)--(180:\n);		
						\foreach \k in {-4.5,-5,-5.5}
																						\foreach \j in {\n-\k}
																						\draw[very thick, color=red, shift=(240:\n), shift=(300:\n)]\hexcoord{\j*\n}{\k*\n} (0:\n)--(60:\n)--(120:\n);																																																																																																															\foreach \k in {1}
																																																																																																																																																																																																																																			\foreach \j in {13*\n+\k}
																																																																																																																																																																																																																																			\draw[very thick, color=red, shift=(240:\n), shift=(300:\n)]\hexcoord{\j*\n}{\k*\n} (60:\n)--(0:\n);	
					\foreach \k in {4.5}
																																																																																																																																																																																																																																								\foreach \j in {3.5,4.5}
																																																																																																																																																																																																																																								\draw[very thick, color=red, shift=(240:\n), shift=(300:\n)]\hexcoord{\j*\n}{\k*\n} (240:\n)--(180:\n)--(120:\n);																																				\foreach \k in {6}\foreach \j in {2}
		\draw[very thick, color=red, shift=(240:\n), shift=(300:\n)]\hexcoord{\j*\n}{\k*\n} (240:\n)--(180:\n);																																																																																																																																																																																									\draw[shift=(120:\n), shift=(60:\n), fill=gray]\hexcoord{15*\n}{-1*\n} (300:8*\n) node[red]{\small $\lambda_{j}$};												
															
		\end{tikzpicture}}
		\caption{The region $T_{\lambda_j}$ in shaded gray and walks $\alpha_{K+1}$ and $\beta_K$ for a walk $\lambda_j$ that intersects $C_i(\delta)$.}\label{fig:shadedregion}
	\end{figure}
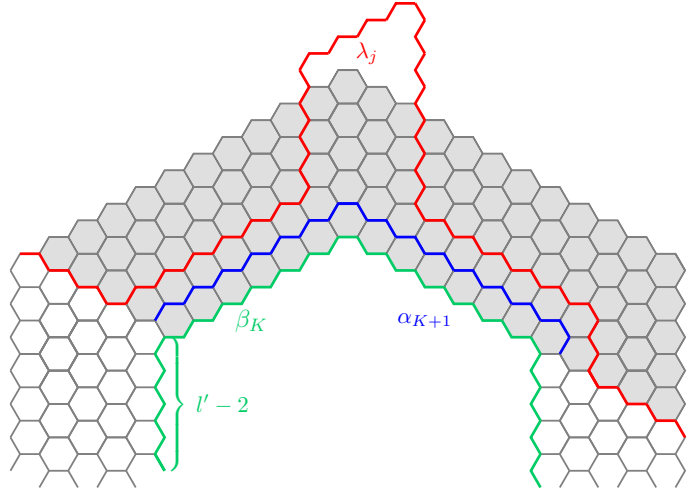
	
	Suppose that $|\lambda_j|\leq 12N - 9l'/2 + 7$. Each boundary vertex $v\in\partial\Lambda_N$ shares an edge with a unique vertex of $\gamma^{(N)}$, and any walk that has $v$ as an endpoint must contain this edge. Hence, the set of vertices from $\gamma^{(N)}$ that shares an edge with a boundary site can be used to bound $|X_{l'}(\lambda_j)|.$
	
	We define $\caT_{\lambda_j}$ as in the interior case, and note that $\gamma^{(N-\delta+1)}$ is the smallest fundamental loop contained in $C_o(\delta)$. Let
	\[
	\alpha_{M}=\gamma^{(M)}\cap \caT_{\lambda_j}, \qquad M \in\{N, N-\delta+1\}\,.
	\]
	and take $\beta_{N}$ to be the walk that extends $\alpha_N$ by $l'-2$ edges along $\gamma^{(N)}$ on both sides. Similar to the interior case, $|\alpha_{N-\delta+1}|\leq |\lambda_j|$. Moreover, any walk $\gamma'\in \caW^{l'}$ that intersects $\lambda_j$ must intersect $\caT_{\lambda_j}$, and hence has endpoints indexed by the vertices of $\beta_N$ that share an edge with a vertex $v\in\partial\Lambda_N$. In contrast to the interior case, though, \[|\beta_{N}|=|\alpha_N|+2(l'-2)>|\alpha_{N-\delta+1}|+2(l'-2),\] 
	and one needs to take more care to obtained the desired bound. 
	
	To this end, let $\kappa\in\bZ_{\geq 0}$ denote the number of corners contained in $\caT_{\lambda_j}$. Given the definitions of $\Lambda_{N,K}$ and $\caT_{\lambda_j}$, a simple geometric argument gives
	\[
	|\alpha_N| = |\alpha_{N-\delta+1}| + 2(\kappa+1)(\delta-1).
	\]
	Considering the set of boundary vertices indexed by $\beta_N$, the left most one is not an admissible right vertex, as are one of the vertices next to each of the $\kappa$ corners enclosed by $\alpha_N$. Moreover, the admissible right vertices are separated by at least two edges of $\beta_N$. Therefore,
	\[
	|X_{l'}(\lambda_j)| \leq \frac{\beta_N-\kappa}{2} \leq \frac{|\lambda_j|}{2} + l'-\frac{3}{2} + (\kappa+1)\left(\delta-\frac{3}{2}\right) \leq  \frac{|\lambda_j|}{2}+\frac{\ell_0}{2}\,.
	\]
	where we use that $\kappa\leq 6$ and $\delta\leq \frac{l'}{4}$.  This completes the proof.
	
\end{proof}

	\subsection{Proof of Theorem~\ref{thm:KP_condition}. Convergence of the cluster expansion}\label{sec:convergence_proof}
	\par
	Convergence for the case $K=0$ was proved in \cite[Appendix]{kennedy:1988}. The proof for $K>20$ follows the analog of \cite{kennedy:1988} using $\epsilon=.0086$ the updated cardinality bounds from Section~\ref{sec:Cardinality_bounds}. For ease of comparison, the proof is structured in the same was as in \cite{kennedy:1988}, and we verify the criterion \eqref{m0_convergence} case by case for $|\gamma|=3, 4,5,6$ and $|\gamma|>6$, distinguishing when $\gamma$ is a walk or a loop. We recall that $\cW^l$ and $\cL^l$ denote, respectively, the sets of walks and loops of length $l$ in $\caP_{N,K}$ and that $w(l) = 3^{-l+1}e^{a(l)+\epsilon l}.$
	\vskip12pt
	
	\textbf{Case $\gamma\in\caW^3$:} There are precisely six walks of length three in $\Lambda_{N,K}$, one at each of the exterior corners of $\Lambda_{N,K}$, and the argument is the same for each of them.  
	For any such $\gamma$, the number of loops $\gamma'\in \caL^{l}$ of length $l<10$ that intersect $\gamma$, i.e. $M_{w\ell}(3,l)$, and walks $\gamma'\in \caW^{l}$ of length $l\leq20$ that intersect $\gamma$, i.e. $M_{ww}(3,l)$, can be explicitly counted using a computer algorithm (see Appendices \ref{subsec:saws_col1-4} and \ref{subsec:saws_col5}). The resulting values are given in the first column of Table~\ref{table:saw}, where again we note that there are no loops of length $l<6$ or $l=8$ in the hexagonal lattice. 
	Considering \eqref{split} and substituting the values from this table, applying Lemma~\ref{lem:loopbd}, and utilizing Table~\ref{table:loops} yields
	\begin{align*}
			\sup_{\gamma\in\caW^3}\frac{1}{a(3)}\sum_{l=3}^{\infty}w(l)|\cP_{N,K}^{\gamma,l}| \leq &  \frac{1}{a(3)}\left(\sum_{3\leq l\leq 20} w(l)M_{ww}(3,l)+\sum_{l=6,10} w(l)M_{w\ell}(3,l) \right. \\
		&\phantom{\frac{1}{a(3)}(}+ \sum_{l\geq 21} w(l)\left(2^{l-4}+(2l+95)2^{l-10}\right) \\
		&\phantom{\frac{1}{a(3)}(}\left.+\sum_{6\leq l \leq 14} w(2l)N_\ell^\Gamma(2l) + \sum_{l\geq 15}w(2l)2^{2l-3}\right) \\
		\leq & \:0.9249
	\end{align*}
This establishes \eqref{m0_convergence} for $\gamma\in \caW^3$. \\
\\
\textbf{Case $\gamma\in \cW^4,\cW^5,\cW^6, \caL^6$:} The desired bounds in all of these cases follow from the same argument. In the case $\gamma\in\caW^4$, applying Lemma~\ref{lem:loopbd}-\ref{lem:small_walks} and using the values in Table~\ref{table:saw}, one finds
\begin{align*}
	\sup_{\gamma\in\caW^4}\frac{1}{a(4)}\sum_{l=3}^{\infty}w(l)|\cP_{N,K}^{\gamma,l}| \leq &  \frac{1}{a(4)}\left(\sum_{3\leq l\leq 20} w(l)S_{ww}(4,l)+\sum_{l=6,10} w(l)M_{w\ell}(4,l) \right. \\
	&\phantom{\frac{1}{a(4)}(}+ \sum_{l\geq 21} w(l)\left(2^{l-4}+2(2l+95)2^{l-10}\right) \\
	&\phantom{\frac{1}{a(4)}(}\left.+\sum_{6\leq l \leq 14} 2w(2l)N_\ell^\Gamma(2l) + \sum_{l\geq 15}w(2l)2^{2l-2}\right) \\
	\leq & \:0.9210
\end{align*}
The bounds for $\gamma\in\cW^5,\cW^6, \caL^6$ use the same argument up to appropriate modifications of the sums. This results in the following values:
\begin{align}
\sup_{\gamma\in\caW^5}\frac{1}{a(5)}\sum_{l=3}^{\infty}w(l)|\cP_{N,K}^{\gamma,l}| \leq &  {\:0.9416}\\ \sup_{\gamma\in\caW^6}\frac{1}{a(6)}\sum_{l=3}^{\infty}w(l)|\cP_{N,K}^{\gamma,l}|  \leq &  {\:0.9127}\\
\sup_{\gamma\in\caL^6}\frac{1}{a(6)}\sum_{l=3}^{\infty}w(l)|\cP_{N,K}^{\gamma,l}|  \leq &  {\:0.8778}
\end{align}

\begin{table}[t]
	\centering
	\begin{tabular}{ |p{3cm}||p{3cm}|p{3cm}|}
		\hline
		\multicolumn{2}{|c|}{$Q(l)$ for $l\leq 19$} \\
		\hline
		$l$ & $Q(l)$\\
		\hline
		3&1\\
		5&2\\
		7&7\\
		9&20\\
		11&64\\
		13&202\\
		15&647\\
		17&2094\\
		19&6803\\
		\hline 
	\end{tabular}
	\medskip
	
	\caption{An upper bound on $Q(l)$, the number of walks with odd length $l\leq 19$ whose endpoints are on opposite sides of a fixed inner or outer corner. }
	\label{table:Q}
\end{table}

\textbf{Case $\gamma\in \cW^l$ with $l>6$:} We again apply Lemma~\ref{lem:loopbd} to bound 
\begin{align}
	|M_{w\ell}(l,l')| & \leq l2^{l'-3} \qquad \forall \, l'\geq 6 \label{wl_all}\\
	|M_{ww}(l,l')| & \leq 2^{l'-4}+(l-2)(2l'+95)2^{l'-10} \leq l(2l'+95)2^{l'-10}\qquad \forall\,  l'\geq 21\,. \label{ww_big}
\end{align}
For even $4\leq l'\leq 20$, we bound $|M_{ww}(l,l')|$ using Lemma~\ref{m3} and Table \ref{table:R}. 

It is left to bound $|M_{ww}(l,l')|$ for odd $3\leq l'\leq 19$. To this end, as discussed in Lemma~\ref{m3}, recall that the endpoints $\mathrm{ep}(\gamma')=\{v,w\}$ of any walk $\gamma'$ with length $l'\leq 19$ must both belong to the same boundary component, $v,w\in\partial\Lambda_N$ or $v,w\in\partial\mathring{\Lambda}_K$. Moreover, if $l'$ is odd, these endpoints must belong to different sides of a corner, see Figure~\ref{fig:case1and2_polymers}. Hence, enumerating the six corners, this partitions $\caW^{l'}$ into twelve sets:
\begin{equation}\label{wl_decomp}
\caW^{l'} = \bigcup_{\#\in\{i,o\}}\bigcup_{j=1}^6 \caW_{j}^{l',\#}
\end{equation}
where $\gamma'\in \caW_j^{l',i}$ if $\gamma'$ crosses the $j$-th corner and its endpoint belong to the inner boundary, $\partial\mathring{\Lambda}_K$, and analogously for $\gamma'\in \caW_j^{l',o}$ given the endpoints are on the outer boundary, $\partial\Lambda_N$. Given a fixed choice of boundary component, i.e. $\#\in\{i,o\}$, the sets $\caW_j^{l',\#}$ are the same for all $j$ up to rotations. The number of walks of length $l'$ around a fixed outer or inner corner can be enumerated using the code in Appendix \ref{subsec:Q}. Letting 
\[Q(l')=\max\left\{\left|\caW_1^{l',o}\right|, \, \left|\caW_1^{l',i}\right|\right\}\] 
denote the larger of these two numbers for each $3\leq l' \leq 19$ odd, one obtains the values in Table~\ref{table:Q}. 

Now, let 
\[
\kappa^{w}(l,l') = \sup_{\gamma\in \caW^l}\left|\left\{(j,\#): \gamma'\nmid \gamma \; \text{for some} \; \gamma'\in \caW_j^{l',\#}\right\}\right|,\,
\]
be the maximal number of parts from \eqref{wl_decomp} that contains a walk $\gamma'$ that intersects a fixed polymer $\gamma\in \caW^l$. Then, one  has that for any odd $3 \leq l' \leq 19$
\[
|M_{ww}(l,l')| \leq \kappa^{w}(l,l') Q(l').
\]

The desired bound on \eqref{m0_convergence} will follow from showing that $\kappa^{w}(l,l')/l \leq 1/7$ for all $l>6$. This is trivial if $\kappa^{w}(l,l')=1$. If $\kappa^{w}(l,l') >1$, then it must be that
\begin{equation}\label{l_bound}
l \geq \left(\kappa^{w}(l,l')-1\right)d_{\rm min}(l')
\end{equation}
where $d_{\rm min}(l')$ denote the minimal graph distance between two walks of length $l'$ that belong to different parts of the partition from \eqref{wl_decomp}, i.e.
\begin{align}
d_{\rm min}(l')& =\min\left\{d(\gamma_1,\gamma_2):\gamma_i\in \caW_{j_i}^{l',\#_i},\; (j_1,\#_1)\neq (j_2,\#_2)\right\}\label{dmin}\\
& = \min\left\{2K-2(l'-4), \, 2(N-K)-l'+1\right\}\nonumber \\
&\geq 2K_\Gamma-30\,.\nonumber
\end{align}
Above, the second equality follows from considering all the possible pairs $(j_1,\#_1)$ and $(j_2,\#_2)$ and walks in $\caW^{l'}$, and the inequality uses that $N -K\geq N_\Gamma$ (by the assumption in Theorem \ref{thm:gs_indistinguishability}) and $K\geq K_\Gamma$. Since $l\geq d_{\rm min}(l')$ when $\kappa^{w}(l,l')>1$, it follows from \eqref{l_bound} and $K_\Gamma = 25$ that
\begin{equation}\label{kappa_l_bound}
\frac{\kappa^{w}(l,l')}{l} \leq \frac{2}{d_{\rm min}(l')}  \leq \frac{1}{7}
\end{equation}
as desired.

Thus using the bounds from \eqref{wl_all}-\eqref{ww_big}, \eqref{kappa_l_bound}, Lemma~\ref{m3} and the values in Tables~\ref{table:R}-\ref{table:Q}, we obtain
\begin{equation}
	\begin{aligned}
		\sup_{\gamma\in\caW^l}\frac{1}{a(\gamma)}\sum_{l'=3}^{\infty}w(l')|\cP_{N,K}^{\gamma,l'}| &\leq \frac{1}{0.15l}\left(\sum_{\substack{3\leq l'\leq 20: \\
				l' \,\mathrm{ even}}}M_{ww}(l,l')w(l')+ \sum_{\substack{3\leq l'\leq 20: \\
				l' \,\mathrm{ odd}}}\kappa^{w}(l,l')Q(l')w(l')
		\right.\\&\left.+l \sum_{l'= 3}^{\infty} 2^{2l'-3}w(2l')+ l\sum_{l'= 21}^{\infty}(2l'+95)2^{l'-10}w(l')\right)\leq0.9997
	\end{aligned}
\end{equation}

\begin{table}[t]
	\centering
	\begin{tabular}{|l||l|l|l|l|l|l|l|}
		\hline
		$\gamma'\backslash \gamma$ & $\mathcal{W}^3$ & $\mathcal{W}^4$ & $\mathcal{W}^5$ & $\mathcal{W}^6$ & $\mathcal{L}^{6}$ & $\mathcal{W}^{>6}$ & $\mathcal{L}^{>6}$\\
		\hline\hline
		$\mathcal{W}^3$ & 0.3688 & 0.3425 & 0.2906 & 0.2740 & 0.2740 & 0.1826 & 0.1279\\ \hline
		$\mathcal{W}^4$ & 0.2581 & 0.2397 & 0.3050 & 0.2876 & 0.1917 & 0.2876 & 0.2684\\ \hline
		$\mathcal{W}^5$ & 0.0959 & 0.0891 & 0.0756 & 0.0712 & 0.0712 & 0.0475 & 0.0332\\ \hline
		$\mathcal{W}^6$ & 0.0336 & 0.0312 & 0.0397 & 0.0374 & 0.0249 & 0.0457 & 0.0407\\ \hline
		$\mathcal{W}^7$ & 0.0480 & 0.0520 & 0.0442 & 0.0416 & 0.0416 & 0.0278 & 0.0194\\ \hline
		$\mathcal{W}^8$ & 0.0250 & 0.0261 & 0.0296 & 0.0302 & 0.0232 & 0.0836 & 0.0737\\ \hline
		$\mathcal{W}^{9}$ & 0.0171 & 0.0204 & 0.0183 & 0.0181 & 0.0181 & 0.0121 & 0.0085\\ \hline
		$\mathcal{W}^{10}$ & 0.0086 & 0.0097 & 0.0102 & 0.0106 & 0.0085 & 0.0404 & 0.0253\\ \hline
		$\mathcal{W}^{21>l>10}$ & 0.0189 & 0.0239 & 0.0220 & 0.0243 & 0.0219 & 0.0581 & 0.0274\\ \hline
		$\mathcal{W}^{>20}$ & 0.0300 & 0.0472 & 0.0565 & 0.0687 & 0.0927 & 0.0721 & 0.0721\\ \hline
		$\mathcal{L}^{6}$ & 0.0168& 0.0312 & 0.0397 & 0.0374 & 0.0873 & 0.1163 & 0.1163\\ \hline
		$\mathcal{L}^{10}$ & 0.0014 & 0.0031 & 0.0041 & 0.0035 & 0.0106 & 0.0165 & 0.0165\\ \hline
		$\mathcal{L}^{>10}$ & 0.0027 & 0.0049 & 0.0063 & 0.0079 & 0.0119 & 0.0092 & 0.0092\\ \hline\hline
		Total & 0.9249 & 0.9210 & 0.9416 & 0.9127 & 0.8778 & 0.9997 & 0.8389\\ \hline
	\end{tabular}
	\medskip
	
	\caption{Bounds used to verify \ref{eqn:cluster_condition} for the $m$-decorated AKLT model. The sum of the entries in each column must be less than $1$ for the Koteck\'{y}-Preiss condition to be satisfied. Each column is indexed by the length and type of the fixed polymer $\gamma$, and each row corresponds to the bound on the sum for the $\gamma'\nmid \gamma$ of a particular type and length. }
	\label{table:totals}
\end{table}

\textbf{Case $\gamma\in \caL^l$ with $l>6$:} This proceeds analogously to the previous case.  Here again, it will be important that $K$ is sufficiently large, specifically we will use $K\geq K_\Gamma=25$, as required in the statement of Theorem \ref{thm:gs_indistinguishability}. 

Now, let 
\[
\kappa^{\ell}(l,l') = \sup_{\gamma\in \caL^l}\left|\left\{(j,\#): \gamma'\nmid \gamma \; \text{for some} \; \gamma'\in \caW_j^{l',\#}\right\}\right|.
\]
Recall that since there are no loops of length eight, that one has that the minimal loop length in this case is $l\geq 10$. Hence, the result follows from showing $\frac{\kappa^{\ell}(l,l')}{l}\leq\frac{1}{10}$ for all $l\geq 10$. As before, this is trivial if $\kappa^{\ell}(l,l')=1$, and if $\kappa^{\ell}(l,l')>1$ then
$$
l \geq \left(\kappa^{\ell}(l,l')-1\right)d_{\rm min}(l')
$$
with $d_{\rm min}$ as in \eqref{dmin}. Since $l\geq d_{\rm min}(l')\geq 2K_\Gamma -30$, and $K_\Gamma = 25$ one finds that
\begin{equation}\label{kappa_l_bound_2}
	\frac{\kappa^{\ell}(l,l')}{l} \leq \frac{2}{d_{\rm min}(l')}  \leq \frac{1}{10}
\end{equation}
as desired.
By Lemma~\ref{lem:loopbd} and the bounds in Lemma \ref{m3} we have
\begin{equation}
	\begin{aligned}
		\sup_{\gamma\in \caL^{l}}\frac{1}{a(\gamma)}\sum_{l'=3}^{\infty}w(l')|\cP_{N,K}^{\gamma,l'}| &\leq\frac{1}{0.15l}\left( \sum_{\substack{3\leq l'\leq 20: \\
				l' \,\mathrm{ even}}}M_{w \ell}(l,l')w(l')+  \sum_{\substack{3\leq l'\leq 20: \\
				l' \,\mathrm{ odd}}}\kappa^{\ell}(l,l')Q(l')w(l')\right.
		\\&\left.+l \sum_{l'= 3}^{\infty} 2^{2l'-3}w(2l')+ l\sum_{l'= 21}^{\infty}(2l'+95)2^{l'-10}w(l')\right)\\&\leq0.8389
	\end{aligned}
\end{equation}
where in the final inequality we have used \eqref{kappa_l_bound_2} whenever $K\geq 25$.
Hence, for any $\gamma\in \caP_{N,K}$, we have that
\begin{equation*} 
	\sum_{\substack{\gamma'\in\cP_{N,K}\\\gamma\nmid\gamma'}}w(|\gamma'|) <a(|\gamma|)
\end{equation*}
In summary, this proves Theorem \ref{thm:KP_condition} with $\epsilon=0.0086$.

\hfill\qedsymbol\\

As a general note, one can compare our Tables \ref{table:loops}, \ref{table:N}, \ref{table:saw}, \ref{table:R}, \ref{table:totals}, \ref{table:Q} with the Tables I, III, II, VI, IV, V respectively in \cite{kennedy:1988}. Our Table \ref{table:saw} corresponds to their Table II where the only discrepancy is the last two entries in column $3$, which is due to the fact that our polymer set $\caP_{N,K}$ includes walks that begin and end on the inner boundary. Tables \ref{table:loops}, \ref{table:N}, \ref{table:R}, \ref{table:Q} agree exactly with Tables I, II, VI, V respectively of \cite{kennedy:1988}. These values can be checked using the code provided in Appendix \ref{subsec:Pb},\ref{subsec:N},\ref{subsec:R},\ref{subsec:Q}. Table \ref{table:totals} differs from Table IV of \cite{kennedy:1988} in the following four ways: (1) we set $\epsilon=0.0086$ which changes all values a small amount; (2) the row of values for $\caW^{>20}$ is calculated with a slightly improved bound \eqref{95} for all columns; (3) column $6$ and $7$ use different bounds given in Lemma \ref{m3} (the worst-case scenario argument from \cite{kennedy:1988} is inaccurate) for rows corresponding to $\caW^8$, $\caW^{10}$, $\caW^{21>l>10}$, $\caW^{>20}$; (4) and lastly our column $5$ contains values that do not correspond to those Table IV of \cite{kennedy:1988} when $\epsilon_h=0$. We could not recover the values of column $5$ of Table IV of \cite{kennedy:1988} with the formulae and numbers provided; however the result is unaffected since all column sums are less than $1$.
	
\section{Proof of Theorem~\ref{thm:gs_indistinguishability} for the Hexagonal Case}\label{sec:proof}
We now give the proof of the main result.

\begin{proof}[Proof of Theorem~\ref{thm:gs_indistinguishability}]
	Fix any $N,K$ such that $N\geq K+53\geq 78,$ and $m\geq 0$. Since $\caP_{N,K}^{\rm bulk}\subseteq \caP_{N,K}$, it follows by Theorem~\ref{thm:KP_condition} that Theorem~\ref{thm:kp} holds for both sets of polymers. Recalling \eqref{mweight} this implies
	\begin{equation}\label{bulk_log}
		\log(\Phi_{N,K}^{{\rm bulk},(m)})=\log(Z(\caP_{N,K}^{{\rm bulk}},W_m,\delta))
	\end{equation}
	has a convergent cluster expansion given by \eqref{cluster_expansion}. Since the weight function $W_m$ and the Ursell function are both real-valued, one sees $\Phi_{N,K}^{{\rm bulk},(m)}>0$. Hence, Lemma~\ref{lem:op_norm} holds. Combining this with \eqref{two_bounds}, one finds 
	\begin{align}
		\left|\braket{\Psi(f)}{A\Psi(f)} - \omega^{(m)}(A)\right|
		\leq \|A\|&\left[\sup_{\Omega^{\partial\Lambda_N}}H\left(\frac{\Phi^{\mathrm{bulk},(m)}_{N,K}}{Z^{\mathrm{bulk},(m)}_N}, \frac{\Phi_{N,K}^{(m)}}{Z_N^{(m)}}\right)\right. \label{H_bound}\\
		&+\left.\limsup_{M}H\left(\frac{\Phi^{\mathrm{bulk},(m)}_{N,K}}{Z^{\mathrm{bulk},(m)}_N}, \frac{\Phi^{\mathrm{bulk},(m)}_{M,K}}{Z^{\mathrm{bulk}}_M,(m)}\right)\right] \nonumber
	\end{align}
	where $\omega^{(m)}$ is the unique infinite volume ground state for the AKLT model on the $m$-decorated hexagonal lattice, $\Psi(f)\in \caG_{\Lambda_N^{(m)}}$ is any normalized ground state, and $A\in \caA_{\Lambda_{K-1}^{(m)}}$. Here, recall that
	\[H(f,g) = D_\infty(f||g)e^{D_\infty(f||g)}\quad \text{where}\quad D_\infty(f||g)= \|\log(f)-\log(g)\|_{L^\infty(d\rho_{\mathring{\Lambda}_K})}\] 
	is the classical Renyi divergence. 
	
	We proceed by using the cluster expansion for $\Phi_{N,K}^{{\rm bulk},(m)}$ and $\Phi_{N,K}^{(m)}$ with \eqref{divergence_bulk}-\eqref{divergence_boundary} to appropriately bound $D_{\infty}$ for the two cases needed for \eqref{H_bound}. The strategy is to first simplify each term on the right hand side in \eqref{H_bound} by canceling the common terms in the cluster expansions. We note that \eqref{bulk_log} also holds for $\Phi_{N,K}^{(m)}$ with $\caP_{N,K}$. Hence, it is sufficient to consider the set of undecorated set of polymers with the decorated with function $W_m$.
	
	Since $Z_N^{\#,(m)} = \Phi_{N,0}^{\#,(m)}$ for $\#\in\{\cdot,{\rm bulk}\}$, and also Theorem~\ref{thm:KP_condition} holds for $K=0$, by Lemma~\ref{lem:poly_rep},
	\begin{align}
		\log\left(\frac{\Phi_{N,K}^{\#,(m)}}{Z_N^{\#,(m)}}\right) & =\log\left(2^{(m+1)|\caE_{\mathring{\Lambda}_{K}}|}\right)+\sum_{C\in \caC(\caP_{N,K}^\#)}W_m^T(C)-\sum_{C'\in \caC(\caP_{N,0}^\#)}W_m^T(C')\,.\label{log_quotient}
	\end{align}
	Let $M>N$ and set
	\begin{equation}\label{cluster_diffs}
	\Delta_{M,N,K}^{\rm bulk} = \caD_{M,N,K}^{\rm bulk}\triangle \caD_{M,N,0}^{\rm bulk}, \quad \Delta_{N,K}=\caD_{N,K}\triangle \caD_{N,0}
	\end{equation}
where
\[
\caD_{M,N,K}^{\rm bulk}=\caC(\caP_{M,K}^{\rm bulk})\setminus \caC(\caP_{N,K}^{\rm bulk}), \quad \caD_{N,K} = \caC(\caP_{N,K})\setminus \caC(\caP_{N,K}^{\rm bulk})\,.
\]
Since $\caP_{N,K}^{\rm bulk}\subseteq\caP_{M,K}^{\rm bulk}$ for all $K\geq 0$, inserting \eqref{log_quotient} into \eqref{divergence_bulk} gives
	\begin{align}
		D_\infty\left(\frac{\Phi_{M,K}^{{\rm bulk},(m)}}{Z_M^{{\rm bulk},(m)}} \left\| \frac{\Phi_{N,K}^{{\rm bulk},(m)}}{Z_N^{{\rm bulk},(m)}}\right)\right. & = \sup_{\Omega^{\partial\mathring{\Lambda}_K}}\left|\sum_{C\in\caD_{M,N,K}^{\rm bulk}}W_m^T(C) - \sum_{C\in\caD_{M,N,0}^{\rm bulk}}W_m^T(C)\right| \nonumber\\
		& \leq \sup_{\Omega^{\partial\mathring{\Lambda}_K}}\sum_{C\in\Delta_{M,N,K}^{\rm bulk}}\left|W_m^T(C)\right|\,. \label{set_diff_bound_bulk}
	\end{align}
	Similarly, since  $\caP_{N,K}^{\rm bulk}\subseteq\caP_{N,K}$, by \eqref{divergence_boundary},
	\begin{equation}\label{set_diff_bound_bulk_bdy}
		\sup_{\Omega^{\partial\Lambda_N}}D_\infty\left(\frac{\Phi_{N,K}^{(m)}}{Z_N^{(m)}} \left\| \frac{\Phi_{N,K}^{{\rm bulk},(m)}}{Z_N^{{\rm bulk},(m)}}\right)\right. 
		\leq \sup_{\Omega^{\partial\Lambda_{N,K}}}\sum_{C\in\Delta_{N,K}}\left|W_m^T(C)\right| \,.
	\end{equation} 
	The result will be a consequence of first determining the above symmetric differences and then appropriately applying Corollary~\ref{lem:kplemma}.
	
	Recall the definitions of the polymer sets from \eqref{polymer_sets} and consider the symmetric difference from \eqref{cluster_diffs}. Any cluster $C\in \caD_{M,N,K}^{\rm bulk}$ contains either a loop from $\caL_{M,K}$ or walk from $\caW_{M,K}^{\rm bulk}$ that intersects $\Lambda_M\setminus\Lambda_N$, as otherwise the cluster belongs to $\caC(\caP_{N,K}^{\rm bulk})$. Similarly, any cluster $C'\in \caD_{M,N,0}^{\rm bulk}$ contains a loop from $\caL_{M,0}$ that intersects $\Lambda_M\setminus\Lambda_N$. Moreover, if every polymer in $C'$ is supported on $\Lambda_{M,K}$, then $C'$ also belongs to $\caD_{M,N,K}^{\rm bulk}$. From this, one can then easily deduce that
	\begin{equation}\label{sym_diff_bulk}
		\Delta_{M,N,K}^{\rm bulk} = \left\{C\in \caC(\caP_{M,K}^{\rm bulk})\cup \caC(\caP_{M,0}^{\rm bulk}): \exists \, \gamma,\gamma'\in C \; \text{s.t.} \; \gamma\nmid \mathring{\Lambda}_K, \; \gamma'\nmid(\Lambda_{M}\setminus\Lambda_N)\right\}\,.
	\end{equation}
	Note that the polymers $\gamma,\gamma'$ need not be distinct, and as any $C\in \Delta_{M,N,K}^{\rm bulk}$ is a cluster, meaning $\cup_{\gamma\in C}\gamma$ is a connected graph, the set conditions imply that there is some $\gamma\in C$ that intersects $\partial\mathring{\Lambda}_N$.
	
	To calculate the other symmetric difference in \eqref{cluster_diffs}, note that any cluster $C\in \caD_{N,K}$ contains a walk $\gamma$ that belongs to $\Lambda_{N,K}$ and has at least one endpoint on $\partial\Lambda_N$. Similarly, any cluster $C'\in \caD_{N,0}$ contains a walk that belongs to $\Lambda_N=\Lambda_{N,K}$ and has both endpoints on $\partial\Lambda_N$. If all polymers from $C'$ are supported on $\Lambda_{N,K}$, then $C'\in \caD_{N,K}$. Thus, the symmetric difference is the set of clusters that intersect both $\partial\Lambda_N$ and $\mathring{\Lambda}_K$:
	\begin{equation}\label{sym_diff_bulk_bdy}
		 \Delta_{N,K} = \{ C\in \caC(\caP_{N,K})\cup\caC(\caP_{N,0}): \exists\, \gamma,\gamma' \in C\; \text{s.t.}\; \gamma\nmid \partial\Lambda_N,\; \gamma'\nmid \mathring{\Lambda}_K\}\,.
	\end{equation}
	
	We now turn to considering Lemma~\ref{lem:kplemma}, which will be applied four times, once for each set of polymers used in the definitions of \eqref{sym_diff_bulk}-\eqref{sym_diff_bulk_bdy}. To this end, for $K'\in\{K,0\}$, define
	\begin{equation}\label{A_sets}
		A_{M,N,K'}^{\rm bulk} = \{\gamma\in \caP_{M,K'}^{\rm bulk}: \gamma\cap \partial\Lambda_N \neq \emptyset\}, \quad 
		A_{N,K'}  = \{\gamma\in \caP_{N,K'}: \gamma \cap  \partial\Lambda_N \neq \emptyset\},
	\end{equation}
	and recall that $\gamma^{(K+1)}$ is the shortest loop that encloses $\mathring{\Lambda}_{K+1}$ (this is also the shortest loop that contains every vertex $x\in \partial\mathring{\Lambda}_{K+1}$). Note that $\gamma^{(K+1)}$ is common to all polymer sets being considered:
	\[
	\gamma^{(K+1)}\in \caP_{M,K}^{\rm bulk}\cap \caP_{M,0}^{\rm bulk}\cap \caP_{N,K}\cap \caP_{N,0}\,.
	\]
	Let 
	\[b_m(C) = \sum_{\gamma\in C}\epsilon_h(m+1)|\gamma|\] 
	for any cluster $C\subseteq\caP_{N,K}$ where $\epsilon=0.0086$ is as in Theorem~\ref{thm:KP_condition}, and $A$ any of the four sets from \eqref{A_sets}. If $C$ contains an element of $A$ and intersects $\gamma^{(K+1)}$, then as $G_C=\cup_{\gamma\in C}\gamma$ is connected 
	\[\sum_{\gamma\in C}|\gamma|> d(\partial\Lambda_N, \gamma^{(K+1)})= 2(N-K)-1.\]
	As a consequence,
	\begin{align}
	b_m(\gamma^{(K+1)}, A) & :=\inf\{ b_m(C): C\in \caC(\caP), \, C\nmid \gamma^{(K+1)}, \, C\cap A \neq \emptyset\} \label{b_hex} \\
	&\geq 2\epsilon(m+1)(N-K)\,.\nonumber
	\end{align}
	Here, we use $C\nmid \gamma$ means there exists $\gamma'\in C$ such that $\gamma'\nmid \gamma$.
	
	Therefore, applying Lemma~\ref{lem:kplemma} twice shows, 
	\begin{align*}
		\sum_{C\in \caD_{N,K}\triangle \caD_{N,0}}|W^T(C)| 
		& \leq 
		\sum_{\substack{C\in \caC(\caP_{N,K}): \\ C\cap A_{N,K}\neq \emptyset, \\ C\nmid \gamma^{(K+1)}}}|W^T(C)|+
		\sum_{\substack{C\in \caC(\caP_{N,0}): \\ C\cap A_{N,0}\neq \emptyset, \\ C\nmid \gamma^{(K+1)}}}|W^T(C)| \nonumber\\
		& \leq 2a_m(\gamma^{(K+1)})e^{-2\epsilon(m+1)(N-K)} \label{F_part_one}
	\end{align*}
	where the final expression is independent of the boundary variables $(\Omega_x:x\in\partial\Lambda_{N,K})$. Since \[a_m(\gamma^{(K+1)})=.15(m+1)|\gamma^{(K+1)}|=.9(m+1)(2K+1),\] 
	it immediately follows from \eqref{set_diff_bound_bulk_bdy} that
	\begin{equation}\label{Cauchy_bound}
	\sup_{\Omega^{\partial\Lambda_N}}D_\infty\left(\frac{\Phi_{N,K}^{(m)}}{Z_N^{(m)}} \left\| \frac{\Phi_{N,K}^{{\rm bulk},(m)}}{Z_N^{{\rm bulk},(m)}}\right)\right.  \leq 1.8(m+1)(2K+1)e^{-2\epsilon(m+1)(N-K)} =: F_m(N,K)
	\end{equation}
	Analogously, again applying Corollary~\ref{lem:kplemma} twice and arguing as above shows
	\begin{align*}
		D_\infty\left(\frac{\Phi_{M,K}^{{\rm bulk},(m)}}{Z_M^{{\rm bulk},(m)}} \left\| \frac{\Phi_{N,K}^{{\rm bulk},(m)}}{Z_N^{{\rm bulk},(m)}}\right)\right. & \leq F_m(N,K)	 \label{F_part_two}
	\end{align*}
	Therefore, by \eqref{H_bound} and the definition of $H(f,g)$, one obtains the desired result
	\[
	\left|\braket{\Psi(f)}{A\Psi(f)} - \omega^{(m)}(A)\right| \leq 2F_m(N,K)e^{F_m(N,K)}\,.
	\]
	The result follows from noting that $\eta_\Gamma=2\epsilon$.
\end{proof}

\section{Decorated Square Lattice Models}\label{sec:lieb}
We turn to proving Theorem~\ref{thm:gs_indistinguishability} for the case of the decorated square lattice.  Thus, $\Gamma^{(m)}$ now denotes the $m$-decorated square lattice for $m\geq 0$. The (undecorated) square lattice will also be denoted simply by $\Gamma$. Recall that we consider the sequence of finite volumes $\Lambda_N\subseteq \Gamma$ defined by
\[
\Lambda_N = \bigcup_{\|x\|_\infty \leq n-1} b_x^1(1), \qquad b_x^1(1) = \{y\in \bZ^2: \|y-x\|_1\leq 1\}
\]
and that its interior is $\mathring{\Lambda}_N = \{x\in\bZ^2: \|x\|_\infty\leq n-1\}.$ For $1\leq K<N$ we also define
\[
\Lambda_{N,K} = \bigcup_{K \leq \|x\|_\infty\leq N-1}b_x^1(1)\,.
\] 
The vertex and edge sets of this volume satisfies 
\[\caV_{\Lambda_{N,K}}\cap \caV_{\mathring{\Lambda}_K}=\partial\Lambda_K, \qquad \caE_{\Lambda_{N,K}} = \caE_{\Lambda_N}\setminus \caE_{\mathring{\Lambda}_K}.\] 
For any $m\geq 0$, we denote by $\Lambda_N^{(m)},\mathring{\Lambda}_N^{(m)}, \Lambda_{N,K}^{(m)}\subseteq \Gamma^{(m)}$ the finite graphs obtained by decorated each edge of its undecorated counterpart with $m$ additional vertices, see Figure~\ref{fig:lambda2} and Figure~\ref{fig:lambda_nk_lieb}. 

Analogous to the hexagonal model, Theorem~\ref{thm:gs_indistinguishability} for this case will be a consequence of appropriately bounding the right hand side of \eqref{two_bounds}, where
\begin{equation}\label{square_states}
	\omega_N^{(m)}(A;\Omega^{\partial\Lambda_N})=Z_N^{(m)}(A;\Omega^{\partial\Lambda_N})/Z_N^{(m)}(\Omega^{\partial\Lambda_N}),\qquad \omega_N^{\rm bulk,(m)}(A)=Z_N^{\rm bulk,(m)}(A)/Z_N^{\rm bulk,(m)}
\end{equation}
are the bulk-boundary map and bulk state, respectively, associated to $\Lambda_N^{(m)}$, and the observable $A\in \caA_{\mathring{\Lambda}_K^{(m)}}$, $K<N$; see \eqref{bb_map} and \eqref{bulk_state} as well as \eqref{Z_to_Phi} and \eqref{Phi_bbmap}-\eqref{Phi_bulk}. 

As before, this will be a consequence of rewriting the integrals defining $\Phi_{N,K}^{(m)}(A;\Omega^{\partial\Lambda_N})$ and $\Phi_{N,K}^{\rm bulk,(m)}(A)$ in terms of a polymer representation with a convergent cluster expansion. To motivate this polymer expansion, we first consider the integral defining the bulk-boundary map.

The integration formulas \eqref{basic_ints} are also valid in the case of the decorated square lattice. Thus, for any fixed $K<N$ and any $m\geq 0$, arguing as in \eqref{invariant_consequence}-\eqref{zero_graphs} and using $|\caE_{\Lambda_{N,K}}^{(m)}|=(m+1)|\caE_{\Lambda_{N,K}}|$ gives
\begin{equation}
	\label{square_Zn}
\Phi_{N,K}^{(m)}(A;\Omega^{\partial\Lambda_N}) =  2^{-(m+1)|\caE_{\Lambda_{N,K}}|}\sum_{E\subseteq \caE_{\Lambda_{N,K}}^{(m)}}\int d\Omega^{\Lambda_{N,K}^{(m)}\setminus\partial\Lambda_{N,K}}\prod_{(x,y)\in E}(-\Omega_x\cdot\Omega_y)
\end{equation}
where the integral is only nonzero for edge sets $E$  whose the minimal graph $G_E=(\caV_E,E)$ only has vertices of even degree in the interior of $\Lambda_{N,K}^{(m)}$, that is
\begin{equation}\label{degree_requirement}
	\deg_{G_E}(v) \in \{2,4\} \qquad \forall v\in \caV_{E}\setminus \partial\Lambda_{N,K}\,.
\end{equation}
The main difference between this and the hexagonal case is that the graphs $G_E$ can now have vertices of degree four. As such, for any finite subgraph $G\subseteq \bZ^2$, we denote by
\[\caV^{4}_G=\{v\in\caV_G:\deg_G(v)=4\}\] 
the number of degree 4 vertices in $G$, and also write $\caV_E^4 = \caV^4_{G_E}$ where $G_E$ is the minimal graph associated to a finite set of edges, $E$.

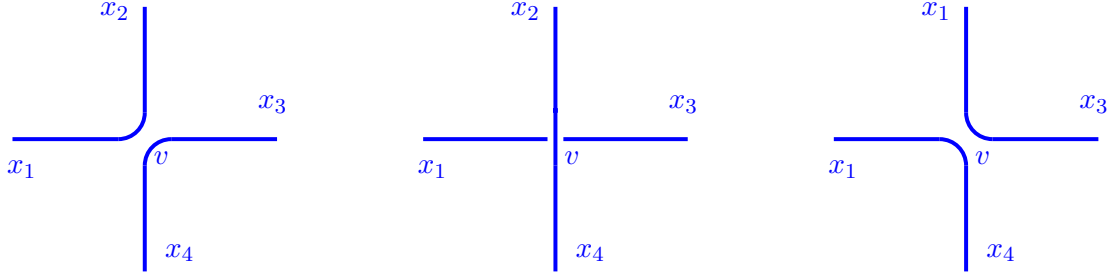
\begin{figure}
	\begin{subfigure}{0.3\textwidth}
		\begin{tikzpicture}
			\def\n{.35};
			\draw [line width=1.5pt, color=blue] (5*\n,0*\n)  to[out=90,in=-90] coordinate[pos=.35] (A) (5*\n,4*\n) ;
			\draw [line width=1.5pt, color=blue] (0*\n,5*\n)  to[out=0,in=180] coordinate[pos=.35] (A) (4*\n,5*\n) ;
			
			\draw [line width=1.5pt, color=blue] (4*\n,5*\n)  to[out=0,in=-90] coordinate[pos=.35] (A) (5*\n,6*\n) ;
			
			\draw [line width=1.5pt, color=blue] (5*\n,6*\n)  to[out=90,in=-90] coordinate[pos=.35] (A) (5*\n,10*\n) ;
			\draw [line width=1.5pt, color=blue] (6*\n,5*\n)  to[out=0,in=180] coordinate[pos=.35] (A) (10*\n,5*\n) ;
			
			\draw [line width=1.5pt, color=blue] (5*\n,4*\n)  to[out=90,in=180] coordinate[pos=.35] (A) (6*\n,5*\n) ;
			
			\draw (0,0) node[midway, below, yshift=200*\n pt,xshift=280*\n pt,color=blue] {$x_3$};
			\draw (0,0) node[midway, below, yshift=130*\n pt,xshift=10*\n pt,color=blue] {$x_1$};
			\draw (0,0) node[midway, below, yshift=40*\n pt,xshift=180*\n pt,color=blue] {$x_4$};
			\draw (0,0) node[midway, below, yshift=300*\n pt,xshift=110*\n pt,color=blue] {$x_2$};
			\draw (0,0) node[midway, below,yshift=140*\n pt,xshift=160*\n, color=blue] {$v$};
		\end{tikzpicture}
	\end{subfigure}\hspace{10pt}
	\begin{subfigure}{0.3\textwidth}
		\begin{tikzpicture}
			\def\n{.35};
			\draw [line width=1.5pt, color=blue] (5*\n,0*\n)  to[out=90,in=-90] coordinate[pos=.35] (A) (5*\n,4*\n) ;
			\draw [line width=1.5pt, color=blue] (0*\n,5*\n)  to[out=0,in=180] coordinate[pos=.35] (A) (4.7*\n,5*\n) ;
			
			\draw [line width=1.5pt, color=blue] (5*\n,6*\n)  to[out=90,in=-90] coordinate[pos=.35] (A) (5*\n,10*\n) ;
			\draw [line width=1.5pt, color=blue] (5.3*\n,5*\n)  to[out=0,in=180] coordinate[pos=.35] (A) (10*\n,5*\n) ;
			\draw [line width=1.5pt, color=blue] (5*\n,4*\n)  to[out=90,in=90] coordinate[pos=.35] (A) (5*\n,6*\n) ;
			\draw (0,0) node[midway, below, yshift=200*\n pt,xshift=280*\n pt,color=blue] {$x_3$};
			\draw (0,0) node[midway, below, yshift=130*\n pt,xshift=10*\n pt,color=blue] {$x_1$};
			\draw (0,0) node[midway, below, yshift=40*\n pt,xshift=180*\n pt,color=blue] {$x_4$};
			\draw (0,0) node[midway, below, yshift=300*\n pt,xshift=110*\n pt,color=blue] {$x_2$};
			\draw (0,0) node[midway, below,yshift=140*\n pt,xshift=160*\n, color=blue] {$v$};
		\end{tikzpicture}
	\end{subfigure}\hspace{10pt}
	\begin{subfigure}{0.3\textwidth}
		\begin{tikzpicture}
			\def\n{.35};
			\draw [line width=1.5pt, color=blue] (5*\n,0*\n)  to[out=90,in=-90] coordinate[pos=.35] (A) (5*\n,4*\n) ;
			\draw [line width=1.5pt, color=blue] (0*\n,5*\n)  to[out=0,in=180] coordinate[pos=.35] (A) (4*\n,5*\n) ;
			\draw [line width=1.5pt, color=blue] (4*\n,5*\n)  to[out=0,in=90] coordinate[pos=.35] (A) (5*\n,4*\n) ;
			\draw [line width=1.5pt, color=blue] (5*\n,6*\n)  to[out=90,in=-90] coordinate[pos=.35] (A) (5*\n,10*\n) ;
			\draw [line width=1.5pt, color=blue] (6*\n,5*\n)  to[out=0,in=180] coordinate[pos=.35] (A) (10*\n,5*\n) ;
			\draw [line width=1.5pt, color=blue] (5*\n,6*\n)  to[out=-90,in=180] coordinate[pos=.35] (A) (6*\n,5*\n) ;
			\draw (0,0) node[midway, below, yshift=200*\n pt,xshift=280*\n pt,color=blue] {$x_3$};
			\draw (0,0) node[midway, below, yshift=130*\n pt,xshift=10*\n pt,color=blue] {$x_1$};
			\draw (0,0) node[midway, below, yshift=40*\n pt,xshift=180*\n pt,color=blue] {$x_4$};
			\draw (0,0) node[midway, below, yshift=300*\n pt,xshift=110*\n pt,color=blue] {$x_1$};
			\draw (0,0) node[midway, below,yshift=140*\n pt,xshift=160*\n, color=blue] {$v$};
		\end{tikzpicture}
	\end{subfigure}\hspace{.5pt}
	\caption{Illustration for each of the three summands in \eqref{basic_ints2}. The vertices $x_i$, $x_j$ connected by a pair of edges in each image corresponds to a term $(\Omega_{x_i}\cdot\Omega_{x_j})$ in the associated summand.}\label{fig:deformation}
\end{figure} 

For a fixed $E$ satisfying \eqref{degree_requirement}, the desired form of \eqref{square_Zn} will result from simplifying
\begin{equation}\label{two_ints}
\int d\Omega^{\Lambda_{N,K}^{(m)}\setminus\partial\Lambda_{N,K}}\prod_{(x,y)\in E}(-\Omega_x\cdot\Omega_y)=\int d\Omega^{\caV_E\setminus\partial\Lambda_{N,K}}\prod_{(x,y)\in E}(-\Omega_x\cdot\Omega_y)
\end{equation}
and evaluating the integral on the right-hand-side above. If $G_E$ contains a vertex, $v$, of degree 4, one needs the following formula in addition to \eqref{basic_ints} to complete this task:
\begin{equation}\label{basic_ints2}
	\int d\Omega_v \prod_{i=1}^{4}(\Omega_{x_i}\cdot\Omega_v)= \frac{1}{15}\sum_{j=2}^4(\Omega_{x_1}\cdot\Omega_{x_{\tau_{j}(2)}})(\Omega_{x_{\tau_j(3)}}\cdot\Omega_{x_{\tau_j(4)}})\,.
\end{equation}
Above, $(x_i,v)\in E$, $i=1,2,3,4$, are the four edges containing $v$, and $\tau_j$ is the permutation on $\{1,2,3,4\}$ that swaps $2$ and $j$. Thus, $\tau_2=\id$. Moreover, we note that the formula in \eqref{basic_ints2} does not depend on $E$ being a subset of edges in $\bZ^2$, just that $v$ is a degree four vertex.

Inserting this back into \eqref{two_ints}, one finds
\begin{equation}\label{E_integral}
	\int d\Omega^{\caV_E\setminus\partial\Lambda_{N,K}}\prod_{(x,y)\in E}(-\Omega_x\cdot\Omega_y) = \frac{1}{15}\sum_{j=2}^4\int d\Omega^{\caV_E\setminus(\partial\Lambda_{N,K}\cup\{v\})}\prod_{(x,y)\in E_j}(-\Omega_x\cdot\Omega_y)
\end{equation}
where $E_j$ is the set of edges between the vertices $\caV_{E}\setminus\{v\}$ defined by 
\[E_j=E\cup\left\{(x_1,x_{\tau_j(2)}), (x_{\tau_j(3)},x_{\tau_j(4)})\right\}\setminus\left\{(v,x_i):i=1,2,3,4\right\}.\]
We note that the two edges added to $E$ above are note edges in $\bZ^2$, but each can be encoded as a pair of edges in $\bZ^2$ as illustrated in Figure~\ref{fig:deformation}. The sum in \eqref{basic_ints2} and visualization in Figure~\ref{fig:deformation} motivate a choice of polymers where the integral on the left hand side of \eqref{E_integral} evaluates to a sum of $3^{|\caV^4(E)|}$ terms, each of which depends on a (distinct) collection of polymers whose graph union is $G_E$, see Figure~\ref{fig:edgehard}. In this case, a polymer will be able to intersect itself and other polymers at vertices but not edges. As a consequence, each term in the sum will be a product of the associated polymer weights and a soft core polymer interaction. This differs from the hexagonal model case, where the integral corresponding to such an $E$ corresponded to a single collection of polymers determined by the connected components of $G_E$ and the polymer interaction is hard core.

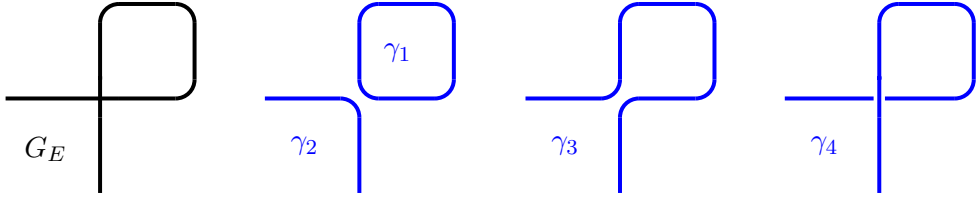
\begin{figure}
	\begin{subfigure}{0.2\textwidth}
		\begin{tikzpicture}
			\def\n{.25};
			\draw [line width=1.5pt, color=black] (5*\n,0*\n)  to[out=90,in=-90] coordinate[pos=.35] (A) (5*\n,4*\n) ;
			\draw [line width=1.5pt, color=black] (0*\n,5*\n)  to[out=0,in=180] coordinate[pos=.35] (A) (6*\n,5*\n) ;
			
			\draw [line width=1.5pt, color=black] (5*\n,6*\n)  to[out=90,in=-90] coordinate[pos=.35] (A) (5*\n,9*\n) ;
			\draw [line width=1.5pt, color=black] (6*\n,5*\n)  to[out=0,in=180] coordinate[pos=.35] (A) (9*\n,5*\n) ;
			\draw [line width=1.5pt, color=black] (9*\n,5*\n)  to[out=0,in=-90] coordinate[pos=.35] (A) (10*\n,6*\n) ;
			\draw [line width=1.5pt, color=black] (5*\n,9*\n)  to[out=90,in=180] coordinate[pos=.35] (A) (6*\n,10*\n) ;
			
			\draw [line width=1.5pt, color=black] (5*\n,4*\n)  to[out=90,in=90] coordinate[pos=.35] (A) (5*\n,6*\n) ;
			
			\draw [line width=1.5pt, color=black] (10*\n,6*\n)  to[out=90,in=-90] coordinate[pos=.35] (A) (10*\n,9*\n) ;
			\draw [line width=1.5pt, color=black] (6*\n,10*\n)  to[out=0,in=180] coordinate[pos=.35] (A) (9*\n,10*\n) ;
			\draw [line width=1.5pt, color=black] (9*\n,10*\n)  to[out=0,in=90] coordinate[pos=.35] (A) (10*\n,9*\n) ;
			\draw (0,0) node[midway, below, yshift=100*\n pt,xshift=60*\n pt,color=black] {$G_E$};
		\end{tikzpicture}
	\end{subfigure}
	\begin{subfigure}{0.2\textwidth}
		\begin{tikzpicture}
			\def\n{.25};
			\draw [line width=1.5pt, color=blue] (5*\n,0*\n)  to[out=90,in=-90] coordinate[pos=.35] (A) (5*\n,4*\n) ;
			\draw [line width=1.5pt, color=blue] (0*\n,5*\n)  to[out=0,in=180] coordinate[pos=.35] (A) (4*\n,5*\n) ;
			\draw [line width=1.5pt, color=blue] (4*\n,5*\n)  to[out=0,in=90] coordinate[pos=.35] (A) (5*\n,4*\n) ;
			\draw [line width=1.5pt, color=blue] (5*\n,6*\n)  to[out=90,in=-90] coordinate[pos=.35] (A) (5*\n,9*\n) ;
			\draw [line width=1.5pt, color=blue] (6*\n,5*\n)  to[out=0,in=180] coordinate[pos=.35] (A) (9*\n,5*\n) ;
			\draw [line width=1.5pt, color=blue] (9*\n,5*\n)  to[out=0,in=-90] coordinate[pos=.35] (A) (10*\n,6*\n) ;
			\draw [line width=1.5pt, color=blue] (5*\n,9*\n)  to[out=90,in=180] coordinate[pos=.35] (A) (6*\n,10*\n) ;
			\draw [line width=1.5pt, color=blue] (5*\n,6*\n)  to[out=-90,in=180] coordinate[pos=.35] (A) (6*\n,5*\n) ;
			\draw [line width=1.5pt, color=blue] (10*\n,6*\n)  to[out=90,in=-90] coordinate[pos=.35] (A) (10*\n,9*\n) ;
			\draw [line width=1.5pt, color=blue] (6*\n,10*\n)  to[out=0,in=180] coordinate[pos=.35] (A) (9*\n,10*\n) ;
			\draw [line width=1.5pt, color=blue] (9*\n,10*\n)  to[out=0,in=90] coordinate[pos=.35] (A) (10*\n,9*\n) ;
			\draw (0,0) node[midway, below, yshift=240*\n pt,xshift=200*\n pt,color=blue] {$\gamma_1$};
			\draw (0,0) node[midway, below, yshift=100*\n pt,xshift=60*\n pt,color=blue] {$\gamma_2$};
		\end{tikzpicture}
	\end{subfigure}\hspace{.5pt}
	\begin{subfigure}{0.2\textwidth}
		\begin{tikzpicture}
			\def\n{.25};
			\draw [line width=1.5pt, color=blue] (5*\n,0*\n)  to[out=90,in=-90] coordinate[pos=.35] (A) (5*\n,4*\n) ;
			\draw [line width=1.5pt, color=blue] (0*\n,5*\n)  to[out=0,in=180] coordinate[pos=.35] (A) (4*\n,5*\n) ;
			
			\draw [line width=1.5pt, color=blue] (4*\n,5*\n)  to[out=0,in=-90] coordinate[pos=.35] (A) (5*\n,6*\n) ;
			
			\draw [line width=1.5pt, color=blue] (5*\n,6*\n)  to[out=90,in=-90] coordinate[pos=.35] (A) (5*\n,9*\n) ;
			\draw [line width=1.5pt, color=blue] (6*\n,5*\n)  to[out=0,in=180] coordinate[pos=.35] (A) (9*\n,5*\n) ;
			\draw [line width=1.5pt, color=blue] (9*\n,5*\n)  to[out=0,in=-90] coordinate[pos=.35] (A) (10*\n,6*\n) ;
			\draw [line width=1.5pt, color=blue] (5*\n,9*\n)  to[out=90,in=180] coordinate[pos=.35] (A) (6*\n,10*\n) ;
			
			\draw [line width=1.5pt, color=blue] (5*\n,4*\n)  to[out=90,in=180] coordinate[pos=.35] (A) (6*\n,5*\n) ;
			
			\draw [line width=1.5pt, color=blue] (10*\n,6*\n)  to[out=90,in=-90] coordinate[pos=.35] (A) (10*\n,9*\n) ;
			\draw [line width=1.5pt, color=blue] (6*\n,10*\n)  to[out=0,in=180] coordinate[pos=.35] (A) (9*\n,10*\n) ;
			\draw [line width=1.5pt, color=blue] (9*\n,10*\n)  to[out=0,in=90] coordinate[pos=.35] (A) (10*\n,9*\n) ;
			\draw (0,0) node[midway, below, yshift=100*\n pt,xshift=60*\n pt,color=blue] {$\gamma_3$};
		\end{tikzpicture}
	\end{subfigure}
	\begin{subfigure}{0.2\textwidth}
		\begin{tikzpicture}
			\def\n{.25};
			\draw [line width=1.5pt, color=blue] (5*\n,0*\n)  to[out=90,in=-90] coordinate[pos=.35] (A) (5*\n,4*\n) ;
			\draw [line width=1.5pt, color=blue] (0*\n,5*\n)  to[out=0,in=180] coordinate[pos=.35] (A) (4.7*\n,5*\n) ;
			
			\draw [line width=1.5pt, color=blue] (5*\n,6*\n)  to[out=90,in=-90] coordinate[pos=.35] (A) (5*\n,9*\n) ;
			\draw [line width=1.5pt, color=blue] (5.3*\n,5*\n)  to[out=0,in=180] coordinate[pos=.35] (A) (9*\n,5*\n) ;
			\draw [line width=1.5pt, color=blue] (9*\n,5*\n)  to[out=0,in=-90] coordinate[pos=.35] (A) (10*\n,6*\n) ;
			\draw [line width=1.5pt, color=blue] (5*\n,9*\n)  to[out=90,in=180] coordinate[pos=.35] (A) (6*\n,10*\n) ;
			
			\draw [line width=1.5pt, color=blue] (5*\n,4*\n)  to[out=90,in=90] coordinate[pos=.35] (A) (5*\n,6*\n) ;
			
			\draw [line width=1.5pt, color=blue] (10*\n,6*\n)  to[out=90,in=-90] coordinate[pos=.35] (A) (10*\n,9*\n) ;
			\draw [line width=1.5pt, color=blue] (6*\n,10*\n)  to[out=0,in=180] coordinate[pos=.35] (A) (9*\n,10*\n) ;
			\draw [line width=1.5pt, color=blue] (9*\n,10*\n)  to[out=0,in=90] coordinate[pos=.35] (A) (10*\n,9*\n) ;
			\draw (0,0) node[midway, below, yshift=100*\n pt,xshift=60*\n pt,color=blue] {$\gamma_4$};
		\end{tikzpicture}
	\end{subfigure}

	\caption{The three collections of trails in $\caS(E)$. Note that all three sets visit the degree four vertex $v$ twice, but every edge is visited exactly once. If one were to first integrate \eqref{E_integral} over $d\Omega_v$, the set $\{\gamma_1,\gamma_2\}$ is the collection of polymers associated from the term $(\Omega_{x_1}\cdot \Omega_{x_4})(\Omega_{x_2}\cdot \Omega_{x_3})$ in \eqref{basic_ints2}, $\{\gamma_3\}$ is the set associated to $(\Omega_{x_1}\cdot \Omega_{x_2})(\Omega_{x_3}\cdot \Omega_{x_4})$, and $\{\gamma_4\}$ is the set associated to $(\Omega_{x_1}\cdot \Omega_{x_3})(\Omega_{x_2}\cdot \Omega_{x_4})$. Note that $\gamma_3$ and $\gamma_4$ correspond to distinct trails in $\bZ^2$.
	}\label{fig:edgehard}
\end{figure}

The desired set of polymers for representing the bulk-boundary map and bulk state will be a subset of the finite trails of the square lattice. A finite trail $\gamma = (\caV_\gamma,\caE_\gamma)$ in $\Gamma^{(m)}$ is a tuple of finite vertex and edge sequences, 
\[\caV_\gamma = (v_1, v_2, \ldots, v_{n+1}), \quad \caE_\gamma = (e_1, e_2, \ldots, e_n)\quad \text{where}\quad e_i = (v_i, v_{i+1}) \quad\forall\, i=1,2,\ldots, n,\] 
and each edge, $e_i,$ is unique. We will call a trail a loop if $v_1=v_{n+1}$, and a walk otherwise. We note that in the standard graph theory definitions, all trails are walks. However, we use the terminology loop and walk to distinguish between the two types of trails. Two walks will be considered equal if they are the same up to reversing both the vertex and edge sequences. Two loops will be consider the same if they are equal up to performing the same cyclic permutation $\sigma\in S_n$ to the vertex and edge sequences, and then (possibly) reversing both sequences. To apply the permutation to the vertex set, we use that $v_{n+1}=v_1$. The set of all unique walks and loops in the square lattice will be denoted $\caW_{\Gamma^{(m)}}$ and $\caL_{\Gamma^{(m)}}$, respectively, and the set of all unique trails by 
\[\caT_{\Gamma^{(m)}}= \caL_{\Gamma^{(m)}}\cup\caW_{\Gamma^{(m)}}.\]
Note that the only difficulty in distinguishing two trails visually comes from if they have the same graph and if that graph has a degree four vertex. However, the three paths in Figure~\ref{fig:deformation} can be used to distinguish them in this case. For example, give $\gamma\in\caE_\gamma$, the first image in Figure~\ref{fig:deformation} corresponds to $\caE_\gamma$ containing the four edges
\[
e_i =(x_1,v), \quad e_{i+1}=(v,x_2), \quad e_j = (x_3, v), \quad e_{j+1}=(v,x_4)
\]
for some $i$ and $j$. 

The first and last vertices of a walk $\gamma\in\caW_{\Gamma^{(m)}}$ will be called the endpoints, denoted $\mathrm{ep}(\gamma)$, and the length of any trail $\gamma \in \caT_{\Gamma^{(m)}}$ will be the number of edges, denoted $|\gamma|$. As before, one always has $\mathrm{ep}(\gamma)\subseteq \Gamma$ regardless of the value $m$. For a walk $\gamma$, we again let $\partial_\gamma(\Omega) = -\Omega_{v_1}\cdot\Omega_{v_{n+1}}$ where $\mathrm{ep}(\gamma)=\{v_1,v_{n+1}\}$, and define $\partial_\gamma(\Omega)=-1$ in the case that $\gamma$ is a loop. The weight function for this model is then given by
\begin{equation}\label{square_weight}
	W(\gamma)  = \left(-\frac{1}{3}\right)^{|\gamma|-1}\left(\frac{3}{5}\right)^{|\caV_\gamma^4|}\partial_\gamma(\Omega), \qquad \forall \gamma \in \caT^{(m)}, \; \forall m\geq 0. 
\end{equation}
As will be evident in the proof of Theorem~\ref{lem:square_polymer} and is suggested by \eqref{basic_ints2}, the factor $(3/5)^{|\caV_4(\gamma)|}$ arises from  first integrating over all degree four vertices, and then over all degree two vertices, which will yield an overall weight $(-1/3)^{|\gamma|-2|\caV^4_\gamma|-1}\cdot(1/15)^{|\caV^4_\gamma|}$ where we slightly abuse notation and use $\caV_\gamma^4$ to denote the number of degree four vertices in the graph $G_\gamma =\{S_v(\gamma), S_e(\gamma)\}$ where $S_v(\gamma)$ and $S_e(\gamma)$ is the set of all vertices and edges, respectively, in the sequences $\caV_\gamma$ and $\caE_\gamma$. We will often abuse notation and simply write $\gamma$ instead of $G_\gamma$. 

Given the above, we define three sets of polymers by:
\begin{align}
	\caL_{N,K}^{(m)} & = \left\{\gamma\in \caL_{\Gamma}: \gamma \subseteq \Lambda_{N,K}^{(m)} \,\wedge\, |S_v(\gamma)\cap\partial\Lambda_K|\leq1\right\} \\
	\caW_{N,K}^{(m)} & = \left\{\gamma\in \caW_{\Gamma}: \gamma \subseteq \Lambda_{N,K}^{(m)}\wedge S_v(\gamma)\cap \partial\Lambda_{N,K}=\mathrm{ep}(\gamma)\right\}\\
	\caW_{N,K}^{{\rm bulk},(m)} & = \left\{\gamma\in \caW_{N,K}^{(m)}: \mathrm{ep}(\gamma)\subseteq \partial\Lambda_{K}\right\}
\end{align}
In words, $\caL_{N,K}$ is the set of all loops in $\Lambda_{N,K}$ that can include at most one boundary vertex, $\caW_{N,K}$ is the set of all walks in $\Lambda_{N,K}$ that only intersect the boundary $\partial\Lambda_{N,K}$ at its endpoints, and $\caW_{N,K}^{{\rm bulk},(m)} \subseteq \caW_{N,K}^{(m)}$ is the set of all such walks that only intersect the interior boundary, $\partial\mathring{\Lambda}_{K}.$ We note that all boundary vertices $x\in\partial\Lambda_{N,K}$ have degree one \emph{except} for the four corners of the interior boundary, which have degree two, see Figure~\ref{fig:lambda_nk_lieb}. The restriction $S_v(\gamma)\cap \partial\Lambda_{N,K}=\mathrm{ep}(\gamma)$ eliminates the possibility of having a walk $\gamma\subseteq\Lambda_{N,K}^{(m)}$ that has one of these degree two boundary sites as a non-endpoint, as this is equivalent to the concatenation of two walks that meet at the corner. Similarly, a loop $\gamma\subseteq\Lambda_{N,K}^{(m)}$ that intersects two or more of these interior boundary corners could be constructed as a concatenation of walks. This leads to the restriction that loops can only visit a single corner on the interior boundary.

Given the definitions above, the set of polymers for the bulk-boundary map and bulk state are, respectively,
\begin{equation}\label{square_polymers}
	\caP_{N,K}^{(m)} = \caL_{N,K}^{(m)} \cup \caW_{N,K}^{(m)}, \qquad \caP_{N,K}^{{\rm bulk},(m)} = \caL_{N,K}^{(m)}\cup \caW_{N,K}^{{\rm bulk},(m)}.
\end{equation}
For $K=0$, we use the convention that $\partial \Lambda_K= \emptyset$. Thus, $\partial\Lambda_{N,K} = \partial\Lambda_N$ and $\caW_{N,0}^{{\rm bulk},(m)}=\emptyset$.
To define the polymer interaction $\delta:\caP_{N,K}\times\caP_{N,K}\to [0,1]$, for any trail, $\gamma\in\caP_{N,K}$, let $\mathring{S}_v(\gamma) = S_v(\gamma)\setminus\partial\Lambda_{N,K}$ denote the interior sites of $\gamma$. Then:
\begin{equation}\label{square_interaction}
	\delta(\gamma,\gamma')=\begin{cases}
		0, & S_e(\gamma)\cap S_e(\gamma')\neq\emptyset \\
		\left(\frac{3}{5}\right)^{|\mathring{S}_v(\gamma)\cap \mathring{S}_v(\gamma')|}   & \mathring{S}_v(\gamma)\cap \mathring{S}_v(\gamma')\neq\emptyset\:\wedge S_e(\gamma)\cap S_e(\gamma')=\emptyset\\
		1 &\mathring{S}_v(\gamma)\cap \mathring{S}_v(\gamma')=\emptyset
\end{cases}\end{equation}
We note that if two trails $\gamma',\gamma$ share a vertex $v$ and have no edges in common, then both $\gamma$ and $\gamma'$ visit $v$ exactly once. Said differently, $v$ is not a degree four vertex of $\gamma$ or $\gamma'$. We now state the desired polymer representation for $Z_{N}^{(m)}(A,\Omega^{\partial\Lambda_N})$ and $Z_N^{{\rm bulk},(m)}(A)$ for the $m$-decorated square lattice.

As in the hexagonal case, there is a bijection between the sets of polymers 
\begin{equation}\label{poly_bijection}
	\caP_{N,K}^{(0)}\ni\gamma\mapsto \gamma^{(m)}\in\caP_{N,K}^{(m)}
\end{equation}
given by decorating the edges of $\gamma$ with $m$ additional vertices, and appropriately redefining the trail. The same relationship also holds for the bulk set of polymers. As a consequence, $|\gamma^{(m)}|=(m+1)|\gamma|$ and since these two walks also have the same set of degree four vertices, it follows that
\begin{equation}\label{mweight_square}
W_m(\gamma):=\frac{W(\gamma)}{3^{m|\gamma|}}=W(\gamma^{(m)}) \qquad \forall \gamma\in \caP_{N,K}
\end{equation}

\begin{lemma}\label{lem:square_polymer} Fix $0\leq K<N$ and a decoration parameter $m\geq 0$. Then, for any $A\in\caA_{\mathring{\Lambda}_K^{(m)}}$
	\begin{align}
		\Phi_{N,K}^{(m)}({\bf\Omega}) & :=2^{-(m+1)|\caE_{\Lambda_{N,K}}|}\sum_{\{\gamma_1,...,\gamma_n\}\subseteq\cP_{N,K}}\prod_{i=1}^nW_m(\gamma_i)\prod_{1\leq i<j\leq n}\delta(\gamma_i,\gamma_j) \\
		\Phi^{\mathrm{bulk},(m)}_{N,K}({\bf\Omega}) & :=2^{-(m+1)|\caE_{\Lambda_{N,K}}|}\sum_{\{\gamma_1,...,\gamma_n\}\subseteq\cP^{\mathrm{bulk}}_{N,K}}\prod_{i=1}^nW_m(\gamma_i) \prod_{1\leq i<j\leq n}\delta(\gamma_i,\gamma_j)
	\end{align}
		where we use the weight function from \eqref{square_weight} and polymer interaction from \eqref{square_interaction}\,.
\end{lemma}

\begin{proof} We focus on establishing the result for $Z_N^{(m)}(A;\Omega^{\partial\Lambda_N})$ since the proof for $Z_N^{{\rm bulk},(m)}(A)$ follows the analogous argument as given for the hexagonal model in Lemma~\ref{lem:poly_rep}. Given \eqref{square_Zn} and \eqref{two_ints}, the proof follows from evaluating the integral \eqref{E_integral} for an arbitrary $E\subseteq \caE_{\Lambda_{N,K}^{(m)}}$ satisfying \eqref{degree_requirement}.
	
	First consider the case $m>0$, and fix $E\subseteq \caE_{\Lambda_{N,K}^{(m)}}$ such that \eqref{degree_requirement} holds. Let $\caS(E)$ be the set of all collections $\{\gamma_1,\ldots,\gamma_n\}\subseteq \caP_{N,K}^{(m)}$ defined by
	\[
	\{\gamma_1,\ldots,\gamma_n\}\in\caS(E) \iff G_E = \bigcup_{i=1}^nG_{\gamma_i}\quad\text{and}\quad \prod_{i<j}\delta(\gamma_i,\gamma_j)\neq 0\,.
	\]
	The definition of $\caP_{N,K}^{(m)}$ is chosen so that $\caS(E)$ is the set of $3^{|\caV_4(E)|}$ collections of trails described by replacing all degree four vertices in $G_E$ with one of the three graphs in Figure~\ref{fig:deformation}, see also Figure~\ref{fig:edgehard}. For any $\{\gamma_1,\ldots,\gamma_n\}\subseteq \caS(E)$, let $\{\tilde{\gamma}_1,\ldots,\tilde{\gamma}_n\}$ denote the set of trails defined by removing all degree four vertices from $G_E$. I.e., if $v\in \caV_E^4$ and $v_j=v$ is the $j$-th vertex in $\caV_{\gamma_i}$, then
	\[
	\caV_{\tilde{\gamma}_i} = (\ldots, v_{j-1},v_{j+1},\ldots), \quad \caE_{\tilde{\gamma}_i}=(\ldots, (v_{j-1},v_{j+1}),\ldots).
	\]
	Notice that since $m\geq 1$, both $v_{j\pm 1}$ are decorated sites, and thus satisfy:
	\[
	\deg_{G_E}(v_{j\pm 1}) = \deg_{\gamma_i}(v_{j\pm 1}) = \deg_{\tilde{\gamma}_i}(v_{j\pm 1})=2.
	\]
	We note that the resulting trail $\tilde{\gamma}_i$ is not a trail in $\Gamma^{(m)}$ as $(v_j,v_{j+1})\notin \Gamma^{(m)}$. However, this is a trail in the complete graph on the vertices of $\Gamma^{(m)}$, and these new trails satisfy the following properties. First, $|\tilde{\gamma}_i|= |\gamma_i|-k$ where $k$ is the number of vertices removed from the sequence $\caV_{\gamma_i}$. As every $v\in \caV_4(E)$ appears twice in the collection $\caV_{\gamma_1},\ldots,\caV_{\gamma_n}$, it follows that
	\begin{equation}\label{degree_identity}
		\sum_{i=1}^n|\tilde{\gamma}_i| = \sum_{i=1}^n|\gamma_i|-2|\caV_E^4|\,.
	\end{equation}
	Second, the set $\{\tilde{\gamma}_1,\ldots\tilde{\gamma}_n\}$ is hard core, meaning that the vertex and edge sets of any two distinct $\tilde{\gamma}_i$, $\tilde{\gamma}_j$ are disjoint.
	
	Thus, considering \eqref{basic_ints} and the discussion that immediately follows, one finds that by first integrating \eqref{E_integral} over all degree four vertices $v\in\caV_E^4$,
	\[
	\int d\Omega^{\caV_E\setminus\partial\Lambda_{N,K}}\prod_{(x,y)\in E}(-\Omega_x\cdot\Omega_y) = \sum_{\{\gamma_1,\ldots, \gamma_n\}\in \caS(E)}\left(\frac{1}{15}\right)^{|\caV_E^4|}\int d\Omega^{\caV_E\setminus (\caV_E^4\cup\partial\Lambda_{N,K})}\prod_{i=1}^n\prod_{(x,y)\in \caE_{\tilde{\gamma_i}}}(-\Omega_x\cdot\Omega_y)
	\]
	Since $\deg_{\tilde{\gamma}_i}(v)\in\{0,2\}$ for all $v\in \caV_E\setminus (\caV_E^4\cup\partial\Lambda_{N,K})$, and $\{\tilde{\gamma}_i:i=1,\ldots, n\}$ is hard core, \eqref{basic_ints} implies
	\begin{equation}\label{tilde_weight}
		\int d\Omega^{\caV_E\setminus\partial\Lambda_{N,K}}\prod_{(x,y)\in E}(-\Omega_x\cdot\Omega_y)  = \sum_{\{\gamma_1,\ldots, \gamma_n\}\in \caS(E)}\left(\frac{1}{15}\right)^{|\caV_E^4|}\prod_{i=1}^n\left(-\frac{1}{3}\right)^{|\tilde{\gamma}_i|-1}\partial_{\gamma_i}(\Omega)
	\end{equation}
	where we use that $\caV_{\gamma_i}\cap\partial\Lambda_{N,K}=\caV_{\tilde{\gamma}_i}\cap\partial\Lambda_{N,K}$, which yields $\partial_{\gamma_i}(\Omega)$. In the case that $\gamma_i$ is a loop that contains single boundary site $v\in\partial\mathring{\Lambda}_K$, we also use
	\[
	-(\Omega_v\cdot\Omega_v)=-1=\partial_{\gamma_i}(\Omega)\,.
	\]
	
	Every $v\in\caV_E^4$ either appears in a single $\gamma_i$ twice, corresponding to a degree four vertex in $G_{\gamma_i}$, or once in two separate $\gamma_i$, $\gamma_j$, corresponding to the polymers intersecting at a vertex. Hence, inserting \eqref{degree_identity}, the definitions of the weight function \eqref{square_weight} and polymer interaction \eqref{square_interaction}, one finds
	\begin{equation}\label{E_int_m}
		\int d\Omega^{\caV_E\setminus\partial\Lambda_{N,K}}\prod_{(x,y)\in E}(-\Omega_x\cdot\Omega_y)  = \sum_{\{\gamma_1,\ldots, \gamma_n\}\in \caS(E)} \prod_{i=1}^n W(\gamma_i)\prod_{i<j}\delta(\gamma_i,\gamma_j)\,.
	\end{equation}
	Now, given any set $\{\gamma_1,\ldots, \gamma_n\}\subseteq \caP_{N,K}^{(m)}$ is any collection of polymers such that $\prod_{i<j}\delta(\gamma_i,\gamma_j)\neq 0$, one can easily check that the edge set $E$ for the graph $\cup_{i}G_{\gamma_i}$ satisfies \eqref{degree_requirement}. Hence, inserting \eqref{E_int_m} into \eqref{two_ints} and recalling \eqref{square_Zn} 
	\[
	\Phi_{N,K}^{(m)}({\bf\Omega})=2^{-|\caE_{\Lambda_{N,K}^{(m)}}|}\sum_{\{\gamma_1,...,\gamma_n\}\subseteq\cP_{N,K}^{(m)}}\prod_{i=1}^nW(\gamma_i)\prod_{1\leq i<j\leq n}\delta(\gamma_i,\gamma_j)\,.
	\]
	The result follows from noting that $|\caE_{\Lambda_{N,K}^{(m)}}|=(m+1)|\caE_{\Lambda_{N,K}}|$, invoking the bijection \eqref{poly_bijection} and inserting \eqref{mweight_square}.
	
	When $m=0$, it is possible that removing all vertices $v\in\caV_E^4$ can result in empty trails $\tilde{\gamma}_i=(\emptyset,\emptyset),$ and so the formula in \eqref{tilde_weight} no longer makes sense. However, letting 
	\[\caE_{{N,K}}^{(m),e}=\left\{E\subseteq\caE_{\Lambda_{N,K}^{(m)}}: \eqref{degree_requirement} \; \text{holds}\right\}\] 
	the operation of decorating any edge $e\in\bZ^2$ with $m\geq 1$ particles induces bijections
	\begin{align}
		\caE^{(0),e}_{N,K}\ni E&\mapsto E^{(m)}\in\caE^{(m),e}_{N,K}
	\end{align}
	where $|E^{(m)}|=(m+1)|E|.$ Here we use that the degree requirement in \eqref{degree_requirement} implies that if one decorated site $v$ along an undecorated edge $e\in\Gamma$ belongs to $\caV_{E^{(m)}}$, then all decorated sites on $e$ must belong to $\caV_{E^{(m)}}$. Hence, integrating the left hand side of \eqref{E_integral} first over the decorated sites, one finds,
	\begin{equation}\label{E_m_to_E_0}
		\int d\Omega^{\caV_{E^{(m)}}\setminus\partial\Lambda_{N,K}}\prod_{(x,y)\in E^{(m)}}(-\Omega_x\cdot\Omega_y) = \left(-\frac{1}{3}\right)^{m|E|} \int d\Omega^{\caV_{E}\setminus\partial\Lambda_{N,K}}\prod_{(x,y)\in E}(-\Omega_x\cdot\Omega_y) 
	\end{equation}
	The result follows from combining \eqref{E_int_m} and \eqref{E_m_to_E_0} with \eqref{mweight_square}.
\end{proof}
					We now turn to verifying that the polymer representation for the decorated square AKLT model satisfies the cluster expansion convergence criterion from Theorem~\ref{thm:KP_condition} in Section~\ref{sec:Cluster_convergence}. Here, we recall that \[\zeta(\gamma,\gamma')=1-\delta(\gamma,\gamma')\] 
					where $\delta$ is as in \eqref{square_interaction}. The result is only valid for $m\geq 1$ as the case $m=0$ does not satisfy the criterion below.
					
					\begin{thm}\label{thm:KP_condition_lieb}
						Fix $m\geq 1$, and let $\cP_{N,K}^{(m)}$, $W_m$ be defined as in \eqref{square_polymers}, \eqref{square_weight}. For any $N,K$ with $N>\max(K+4,8)$, and $\gamma\in\cP_{N,K}$:
					\begin{equation}\label{eqn:cluster_condition_lieb}
						\sum_{\gamma'\in\cP_{N,K}^{(m)}}|W_m(\gamma')|\cdot|\zeta(\gamma,\gamma')|\cdot e^{(a+\epsilon)m|\gamma'|}\leq am|\gamma|\
					\end{equation}  
						where $a=.085$ and $\epsilon=.046$. The same result also holds after replacing $\caP_{N,K}^{(m)}$ with $\caP_{N,K}^{{\rm bulk},(m)}$.
					\end{thm}
					\begin{proof}
							To ease notation, let $\caP_{N,K}=\caP_{N,K}^{(0)}$ and $W=W_0$ denote the polymer set and weight function associated with the undecorated square lattice. Recall the bijection in \eqref{poly_bijection} and weight function relation in \eqref{square_weight}. Since  $\mathring{S}_v(\gamma^{(m)})=\mathring{S}_v(\gamma^{(0)})$ for any $m\geq 0$ (which implies that $\zeta$ is invariant under the decoration $m$) the result follows from showing that for any $\gamma\in \caP_{N,K}$,
							\begin{equation}\label{equiv_cluster_square}
								\sum_{\gamma'\in\cP_{N,K}}\frac{|W(\gamma')|}{3^{m|\gamma'|}}\cdot|\zeta(\gamma,\gamma')|\cdot e^{(a+\epsilon)(m+1)|\gamma'|}\leq a(m+1)|\gamma|\,.
							\end{equation} 
							As in the hexagonal model case, given a fixed  $\gamma\in \caP_{N,K}$, we break up the sum in \eqref{equiv_cluster_square} based on $|\gamma'|$. The shortest such polymer has length $|\gamma'|=2$. As $|W(\gamma')|\leq \left(\frac{1}{3}\right)^{|\gamma'|-1}$ for all $\gamma'\in\caP_{N,K}$, \eqref{equiv_cluster_square} follows from determining constants $C_n\geq 0$ so that
							\begin{align}
								\sum_{n=2}^\infty C_nw_m(n)&\leq a \quad\text{where}\label{convergence_check}\\ \sum_{\substack{\gamma'\in\caP_{N,K}\\|\gamma'|=n}}\frac{|\zeta(\gamma,\gamma')|}{|\gamma|} \leq C_n, \quad&\text{and}\quad w_m(n) = \frac{e^{(a+\epsilon)(m+1)n}}{(m+1)3^{(m+1)n-1}}\,.
							\end{align}

							As before, we write $\gamma\nmid\gamma'$ if $G_\gamma\cap G_{\gamma'}\neq \emptyset$. As $|\zeta(\gamma,\gamma')|\leq 1$ for all $\gamma'\in \caP_{N,K}$, and $\zeta(\gamma,\gamma')\neq 0$ only if $\gamma\nmid\gamma'$, it follows that
							\[
							\sum_{\substack{\gamma'\in\caP_{N,K}\\|\gamma'|=n}}\frac{|\zeta(\gamma,\gamma')|}{|\gamma|} \leq\frac{|\{\gamma'\in\caP_{N,K}:|\gamma'|=n\wedge \gamma\nmid\gamma'\}|}{|\gamma|} ,
							\]
							and we determine the values of $C_n$ by appropriately bounding the cardinality of the set above.
							
							Consider first $n=2$. There are only four walks of length two, each of which straddles one of the corners on the outer boundary of $\Lambda_{N,K}$. As the closest pair of these walks are $2(N-1)$ edges away from one another, one finds
							\begin{equation}
								|\{\gamma'\in\caP_{N,K}:|\gamma'|=2\wedge \gamma\nmid\gamma'\}|\leq 1+\left\lfloor\frac{|\gamma|}{2(N-1)}\right\rfloor \leq \frac{|\gamma|}{2} 
							\end{equation}
							where the final bound above use $|\gamma|\geq 2$. Hence, we set $C_2=\frac{1}{2}$.
							
							For all remaining $n$, we use that
							\begin{equation}
								\left|\left\{\gamma'\in\caP_{N,K}:|\gamma'|=n\wedge \gamma\nmid\gamma'\right\}\right|\leq M_n|\gamma|\,.
							\end{equation}  
							where $M_n$ is the largest number of polymers of length $n$ that pass trough a fixed vertex:
							\begin{equation}
								M_n=\max_{v\in\caV_{\Lambda_{N,K}}}\left|\left\{\gamma'\in\caP_{N,K}:\;v\in\gamma', |\gamma'|=n\right\}\right|
							\end{equation}
							In the case $3\leq n\leq 7$ we bound $M_n\leq C_n:=W_n+L_n$ using the code provided in Appendix \ref{subsec:lieb} by calculating 
							\begin{equation}
							W_n:=\max_{v\in\caV_{\Lambda_{N,K}}}\left|\left\{\gamma'\in\caW_{N,K}:\;v\in\gamma', |\gamma'|=n\right\}\right|
														\end{equation}
							\begin{equation}
														L_n:=\max_{v\in\caV_{\Lambda_{N,K}}}\left|\left\{\gamma'\in\caL_{N,K}:\;v\in\gamma', |\gamma'|=n\right\}\right|
							\end{equation}
							 This produces
							$$C_3=4,\: C_4=9,\: C_5=13,\:,C_6=42,\: C_7=88.$$ 
							For $n\geq 8$, we use the bound
							\begin{equation}
								M_n\leq C_n:= 4(n+1)\cdot3^{n-1},
							\end{equation}
							which follows from the following argument. Fix $1\leq i \leq n+1$ and suppose that $v$ is the $i$-th vertex for a trail $\gamma$ of length $n$. Then, $\gamma$ can be thought of as the concatenation of the trail $(v_1,\ldots, v_i$) preceding $v$ and the trail $(v_i, \ldots, v_{n+1})$ succeeding $v$. The result follows from bounding the maximum number of preceding and succeeding trails using that (1) since no edge may be repeated in $\gamma$, there are four choices for the first edge emanating from $v$, and at most three choices for every other edge in the trail, and (2) there are $n+1$ choices for $i$. 
							
							Substituting $a=.085$, $\epsilon=.046$, and the values of $C_n$ into \eqref{convergence_check}, and using that $w_m(n)$ is decreasing in $m$, one finds:
							\begin{equation}
								\sum_{n=2}^\infty C_nw_m(n) \leq \sum_{n=2}^7 C_nw_1(n) + \sum_{n\geq 8}2(n+1)\left(\frac{e^{2(a+\epsilon)}}{3}\right)^n \leq 0.08425\,.
							\end{equation} 
							and so \eqref{convergence_check} is satisfied.
						\end{proof}
						We can now complete the proof of Theorem~\ref{thm:gs_indistinguishability} given that the decorated square lattice satisfies the cluster expansion convergence criterion from Theorem~\ref{thm:kp}. 
						
						\subsection{Proof of Theorem~\ref{thm:gs_indistinguishability} for the Decorated Square Lattice Models}
						\begin{proof}
							The statement and proof Lemma~\ref{lem:op_norm} hold identically in the decorated square model case as they do in the (un)decorated hexagonal model case. Fix $m\geq 1$, $K\geq 0$ and $N>K-4$. Since the cluster expansion converges for such $N$ and $K$ by Theorem~\ref{thm:KP_condition_lieb} and the weight function $W_m(\gamma)$ is a real-valued continuous function of the boundary variables $\Omega_x$, $x\in\partial\Lambda_{N,K}$, it follows from \eqref{cluster_expansion} that $\Phi_{N,K}^{{\rm bulk}, (m)}>0$. So Lemma~\ref{lem:op_norm} applies and we can proceed similarly to the proof of Case I in Section~\ref{sec:proof}.
							
							One main difference in the decorated square model case is that the uniqueness of the infinite-volume frustration-free ground state is not, a priori, established. However, bounding \eqref{bulk_cauchy_bound} as in Section~\ref{sec:proof} gives that $\omega_N^{{\rm bulk},(m)}(A)$ is Cauchy for each $A\in\caA_{\Gamma^{(m)}}^{\rm loc}$, see \eqref{Cauchy_bound}. As a result, there is an infinite volume, frustration free ground state $\omega^{(m)}:\caA_{\Gamma^{(m)}}\to\bC$ such that
							\[
							\omega^{(m)}(A) = \lim_{N\to\infty}\omega_N^{{\rm bulk},(m)}(A) \qquad \forall \, A\in \caA_{\Gamma^{(m)}}^{\rm loc}\,.
							\]
							We show that \eqref{eq:indistinguishability} is satisfied with this ground state as outlined in \eqref{two_bounds}. As every infinite volume, frustration-free ground state is the weak-$*$ limit of a sequence of finite volume ground states, it follows immediately from \eqref{eq:indistinguishability} that $\omega^{(m)}$ is the unique infinite volume frustration-free ground state.
							
							It is left to bound the right hand side of \eqref{bulk_cauchy_bound} and \eqref{bb_bulk_bound}, for which it is helpful to recall \eqref{Renyi_divergence}. The analogous formula in \eqref{log_quotient} again holds after replacing each quantity with the appropriate decorated square model counterpart, as do \eqref{set_diff_bound_bulk}-\eqref{set_diff_bound_bulk_bdy}. Hence, it is left to determine 
							\[
							\caD_{M,N,K}^{\rm bulk}\triangle \caD_{M,N,0}^{\rm bulk} \qquad \text{and} \qquad \caD_{N,K}\triangle\caD_{N,0}.\]
							
							Recalling the definition of a cluster in \eqref{cluster_def} and considering the polymer interacting in \eqref{square_interaction}, we see that a collection $C=\{\gamma_1,\ldots,\gamma_n\}\subseteq \caP_{N,K}^{(m)}$ is a cluster if the graph
							\[
							G_C=\bigcup_{i=1}^n G_{\gamma_i}
							\]
							is connected after removing any boundary vertices $\partial\Lambda_{N,K}$ as well as any edges connected to a boundary vertex. This is more restrictive than in the hexagonal model case, where one only needed that $G_C$ was connected. Nonetheless, the formulas in \eqref{set_diff_bound_bulk}-\eqref{set_diff_bound_bulk_bdy} still hold in this case. The rest of the proof runs analogously to the hexagonal case. Taking the same set definitions for the sets $A$ in \eqref{A_sets} and function $b$ in \eqref{b_hex}, using the expressions for $a$ and $\epsilon$ from Theorem~\ref{thm:KP_condition_lieb} and noting that
							\[
							|\gamma^{(K+1)}|=8K,\quad {\rm dist}(\partial\Lambda_N,\gamma^{(K+1)})=N-K\,.
							\]
						We find that for any $m\geq 1$
						\begin{equation}
							\max\left\{\sup_{\Omega^{\partial\Lambda_N}}D_\infty\left(\frac{\Phi_{N,K}^{(m)}}{Z_N^{(m)}}\right|\left|\frac{\Phi_{N,K}^{{\rm bulk},(m)}}{Z_N^{{\rm bulk},(m)}}\right),\;D_\infty\left(\frac{\Phi_{M,K}^{{\rm bulk},(m)}}{Z_M^{{\rm bulk},(m)}}\right|\left|\frac{\Phi_{N,K}^{{\rm bulk},(m)}}{Z_N^{{\rm bulk},(m)}}\right)\right\}\leq F_m(N,K)
						\end{equation}
						where
						\[
						F_m(N,K) = 1.36(m+1)Ke^{-\epsilon(m+1)(N-K)}
						\]
					The result is an immediate consequence of \eqref{two_bounds} combined with Lemma~\ref{lem:op_norm}, and noting that $\eta_\Gamma=\epsilon$ in this case.

						\end{proof}
\section*{Acknowledgements}

The authors acknowledge Angelo Lucia for helpful comments on the topic of this paper. 
This work was supported in part by the US National Science Foundation under grants DMS-2108390 and DMS-2510824. 

\appendix

\section{Exponential Decay of Correlations}\label{sec:correlation_decay}
As already discussed in the paper by Koteck\'{y} and Preiss \cite{kotecky:1986}, in a suitable context, convergent cluster expansions can be used to prove exponential decay of two- and multi-point correlation functions. To prove exponential decay of correlations in the ground states of the AKLT models on the hexagonal and Lieb lattices with the cluster expansion we are using in this paper, it turns out that the convergence criterion needs to be slightly adjusted. To prove LTQO for an observable $A$ supported in $\Lambda_K$, the basic geometry of finite volumes to study was given by the sets $\Lambda_{N,K}$, for sufficiently large $N$. To study correlation functions we will need to consider finite volumes with holes to fit the supports of each of the observables involved. This modifies the combinatorics required to check the convergence criterion for the cluster expansion. Since all of this is slightly non-standard, we provide a proof of exponential decay here.

We explicitly analyse the two-point correlations of observables that are translates of $A$ and $B$, supported in balls of radius $M/2$, for $M$ even, in the dual lattice, meaning there exist points $x_A,x_B$ in the dual lattice and balls $$\Lambda_A:=\Lambda_{M/2,x_A},\quad \Lambda_B:=\Lambda_{M/2,x_B}$$ such that $\supp(A)\subset\Lambda_A$ and $\supp(B)\subset\Lambda_B$. We use the lattice metric to define the distance between these balls:
$$
d(\Lambda_A,\Lambda_B):=\inf\{d(a,b)\::\:a\in\caV_{\Lambda_A},\;b\in\caV_{\Lambda_B}\}.
$$
Multi-point correlation functions can be estimated by iterating the two-point estimate.

It will be useful to have a notation for the smallest loops enclosing $\Lambda_A$ and $\Lambda_B$, namely $\gamma^A:=\partial\mathring{\Lambda}_A,$ and $\gamma^B:=\partial\mathring{\Lambda}_B$. 

\begin{thm} Let $A,B$ be observables supported in balls of radius $M/2\geq25$, $\Lambda_A$ and $\Lambda_B$, respectively, with $d(\Lambda_A,\Lambda_B)\geq 50$, and let $\omega$ be the unique infinite-volume ground-state for the AKLT model on either the hexagonal lattice or the Lieb lattice. Then, there exist a constant $\epsilon >0$, independent of $M$, for which we have
	$$
	|\omega(AB)-\omega(A)\omega(B)|\leq \Vert A\Vert \Vert B\Vert  M^{2.9} e^{-\epsilon d(\Lambda_A,\Lambda_B)}.
	$$
	For the $\epsilon$ we can take the same values we used in the proof of the LTQO property: $\epsilon_h=0.0086$ for the AKLT model on the hexagonal lattice and $\epsilon_l=0.046$ for the model on the Lieb lattice. With each further decoration $\epsilon$ can be increased from these values.
\end{thm}


\begin{proof}
	First we note that $\gamma^A,\gamma^B$ are both in $\caP^{\mathrm{bulk}}_N$ for large enough $N$. 
	Recalling the formulae for ground state expectations, we see the expression 
	\begin{align}
		|\omega^{\mathrm{bulk}}_N(A)\omega^{\mathrm{bulk}}_N(B)-\omega^{\mathrm{bulk}}_N(AB)|&=\left|\int d\Omega^{\Lambda_A}\rho_{\Lambda_A}
		\frac{\Phi^{\mathrm{bulk}}_{N,\Lambda_A}}{Z_N}A(\Omega)
		\cdot\int d\Omega^{\Lambda_B}\rho_{\Lambda_B}
		\frac{\Phi^{\mathrm{bulk}}_{N,\Lambda_B}}{Z_N}B(\Omega)\right.\\&\left.-\int d\Omega^{\Lambda_{A, B}}\rho_{\Lambda_{A, B}}
		\frac{\Phi^{\mathrm{bulk}}_{N,\Lambda_{A, B}}}{Z_N}(AB)(\Omega)\right|
	\end{align}
	As defined the symbols factorize when $A$ and $B$ have disjoint support so that 
	$$
	(AB)(\Omega) = A(\Omega)B(\Omega),
	\qquad\rho^{\Lambda_{A, B}}=\rho^{\Lambda_A}\rho^{\Lambda_B},
	\qquad d\Omega^{\Lambda_{A, B}}=d\Omega^{\Lambda_A}d\Omega^{\Lambda_B}
	$$ 
	so that
	\begin{align}
		|\omega^{\mathrm{bulk}}_N(A)\omega^{\mathrm{bulk}}_N(B)-\omega^{\mathrm{bulk}}_N(AB)|&=\left|\int d\Omega^{\Lambda_{A, B}}\rho_{\Lambda_{A, B}}\left[
		\frac{\Phi^{\mathrm{bulk}}_{N,\Lambda_A}}{Z_N}\frac{\Phi^{\mathrm{bulk}}_{N,\Lambda_B}}{Z_N}-\frac{\Phi^{\mathrm{bulk}}_{N,\Lambda_{A, B}}}{Z_N}\right](AB)(\Omega)
		\right| \\&\leq
		||AB||\cdot\left\|
		\frac{\Phi^{\mathrm{bulk}}_{N,\Lambda_A}}{Z_N}\frac{\Phi^{\mathrm{bulk}}_{N,\Lambda_B}}{Z_N}-\frac{\Phi^{\mathrm{bulk}}_{N,\Lambda_{A, B}}}{Z_N}\right\|\\&\leq \Vert A\Vert \Vert B\Vert  H\left(\frac{\Phi^{\mathrm{bulk}}_{N,\Lambda_A}}{Z_N}\frac{\Phi^{\mathrm{bulk}}_{N,\Lambda_B}}{Z_N},\frac{\Phi^{\mathrm{bulk}}_{N,\Lambda_{A, B}}}{Z_N}\right)
	\end{align}
	where as before $H(f,g)=D_{\infty}(f,g)e^{D_{\infty}(f,g)}$ with \begin{equation}D_{\infty}(f,g)=\left\|\log(f)-\log(g)\right\|_{L_{\infty}(\Lambda_{A, B})}\end{equation}

	We will use the same cluster expansion as in the proof of LTQO but there is one caveat. Namely, that the polymer sets are not identical to those used in the proofs of LTQO. In particular, for the correlation functions we need also need to count walks that connect the supports of the two observables. Therefore we need to verify separately that the convergence criterion is still satisfied. It turns out that the convergence criterion is satisfied for the expansions we need with the same parameters $a(l)$ and $\epsilon$.
	
	The expression for the partition function that we need for the correlation functions includes the set of paths that cross between volumes:
	$$
	\caW_{N}^{A, B}:=\{\gamma\in\caP_{\Lambda_N}\::\:\mathrm{ep}(\gamma)\cap\partial\Lambda_A\neq\emptyset,\:\mathrm{ep}(\gamma)\cap\partial\Lambda_B\neq\emptyset\}
	$$ 
	On the other hand,  loops and walks that cross into the interior of either $\Lambda_A$ or $\Lambda_B$ are not present. Our strategy is to show that the Koteck\'{y}-Preiss criterion is satisfied when we include the set of paths that begin and end on either volume $\caW_N^{A, B}$ into $\Phi_{N,\Lambda_A}$ under the assumption that $\Lambda_A$ and $\Lambda_B$ are sufficiently far apart. 
	
	For the hexagonal lattice, referring back to discussion and notations in Section \ref{sec:Cluster_convergence}, we split this up into the case that either $\gamma'$ or $\gamma$ are in $\caW^{A, B}$:
	\begin{enumerate}
		\item If $\gamma'\in\caW^{A, B}$ then assuming that $d(\Lambda_A,\Lambda_B)\geq 50$ then defining 
		$\tilde{\caW}_{l'}^{A,B}=\caW^A_{l}\cup\caW^{A, B}_{l}$
		we have that $$\left|\sum_{\substack{\gamma'\in\tilde{\caW}_{l'}^{A,B}}}w(l')\right|\leq l\cdot 2^{l'}\cdot\frac{1}{3^{l'-1}}$$ as before and so adding these walks does not change the numbers in the first $4$ columns of Table \ref{table:totals} provided $d(\Lambda_A,\Lambda_B)>6$.
		\item
		If $\gamma\in\tilde{\caP}_{A, B}$ then assuming $d(\Lambda_A,\Lambda_B)\geq 7$ implies that $\tilde{\caW}\subset\caW_{>6}$, so we restrict to this case and check that each of the bounds used in Lemmas 5.4 work.
		We need to check that the bounds on $M_{w,w},\;M_{\ell,w},\;M_{w,\ell}$ and $M_{\ell,\ell}$ still work when $l>6$.
		\begin{enumerate}
			\item The bounds on $M_{ww}(l,l')$ are the most intensive. We see that the arguments for $l'$ odd are intact since the distance condition $d(\Lambda_A,\Lambda_B)\geq 50$ will allow us to guarantee that $$\frac{1}{|\gamma'|}\sum_{\substack{\gamma'\nmid \gamma \\3\leq l'\leq 20\mathrm{\:odd}}}w(|\gamma'|)\leq\frac{1}{7}\sum_{\substack{3\leq l'\leq 20\\l'\mathrm{\:odd}}}Q(l')w(l')$$ 
			Similarly the estimate $$|M_{ww}(l,l')|\leq l(2l'+95)2^{l'-10}$$ still holds as long as $d(\Lambda_A,\Lambda_B)\geq 22$. Thus we only need to verify the bounds $M_{w,w}(l,l')$ for $l'$ even and $4\leq l'\leq 20$, which is the content of Lemma 5.4, which holds when $N-K\geq K\geq 25$ and $d(\Lambda_A,\Lambda_B)\geq 50$.
			\item The bounds on $M_{\ell, w}$ follow from the same logic as above.
			\item The upper bounds on $M_{w,\ell}$  remain unchanged as the presence of multiple interior volumes does not increase the number of loops.
			\item The upper bounds on $M_{\ell,\ell}$ remain unchanged as the presence of multiple interior volumes does not increase the number of loops.
		\end{enumerate}
	\end{enumerate}
	
	We conclude that the Koteck\'{y}-Preiss is satisfied with the same parameters as in Section \ref{sec:Cluster_convergence}. 
	
	Similar arguments can be repeated for the Lieb lattice. In fact, adding the polymer set $\caW_N^{A,B}$ does not change the upper bounds used in Theorem \ref{thm:KP_condition_lieb} as long as $d(\Lambda_A, \Lambda_B) \geq 8)$.
	
	In both cases, the exponential decay of correlations now follows by the following standard reasoning.
	
	For $\Lambda\subset\Lambda_N$ let $\caP^{\mathrm{bulk}}_{N,\Lambda}=\{\gamma\in\caP_{N}\::\:\mathrm{ep}(\gamma)\subset\partial\Lambda\}$ and recall from the definitions and the convergence of the cluster expansion of 
	\begin{equation}\log\left(\frac{\Phi_{N,\Lambda_A}^{\rm bulk}}{Z_N^{\rm bulk}}\right) =\log\left(2^{|\caE_{\Lambda_{A}}|}\right)+\sum_{C\in \caC(\caP_{N,\Lambda_A}^{\rm bulk})}W^T(C)-\sum_{C'\in \caC(\caP_{N,\emptyset}^{\rm bulk})}W^T(C')
	\end{equation} from which we get 
	\begin{align}
		D_{\infty}\left(\frac{\Phi^{\mathrm{bulk}}_{N,\Lambda_A}}{Z_N}\frac{\Phi^{\mathrm{bulk}}_{N,\Lambda_B}}{Z_N},\frac{\Phi^{\mathrm{bulk}}_{N,\Lambda_{A, B}}}{Z_N}\right)&=\left\|\log\left(2^{|\caE_{\Lambda_{A}}|}\right)+\sum_{C\in \caC(\caP_{N,\Lambda_A}^{\rm bulk})}W^T(C)-\sum_{C'\in \caC(\caP_{N,\emptyset}^{\rm bulk})}W^T(C')\right.\\&\left.+\log\left(2^{|\caE_{\Lambda_{B}}|}\right)+\sum_{C\in \caC(\caP_{N,\Lambda_B}^{\rm bulk})}W^T(C)-\sum_{C'\in \caC(\caP_{N,\emptyset}^{\rm bulk})}W^T(C')\right.\\&-\left.\log\left(2^{|\caE_{\Lambda_{A, B}}|}\right)-\!\!\!\!\sum_{C\in \caC(\caP_{N,\Lambda_{A, B}}^{\rm bulk})}\!\!\!\!W^T(C)+\sum_{C'\in \caC(\caP_{N,\emptyset}^{\rm bulk})}W^T(C')\right\|_{L_{\infty}(\Lambda_{A, B})}\\&=\left\|\sum_{\substack{C\in\caC(\caP^{\mathrm{bulk}}_{N,\Lambda_{A}})\\C\:\nmid\:\gamma^A,\:C\:\nmid\:\gamma^B}}W^T(C)+\sum_{\substack{C\in\caC(\caP^{\mathrm{bulk}}_{N,\Lambda_{B}})\\C\:\nmid\:\gamma^A, \:C\:\nmid\:\gamma^B}}W^T(C)-\sum_{\substack{C\in\caC(\caP^{\mathrm{bulk}}_{N,\Lambda_{A, B}})\\C\:\nmid\:\gamma^A,\:C\:\nmid\:\gamma^B}}W^T(C)\right.\\&\left.-\sum_{\substack{C\in\caC(\caP^{\mathrm{bulk}}_{N,\emptyset})\\C\:\nmid\:\gamma^A,\:C\:\nmid\:\gamma^B}}W^T(C)\right\|_{L_{\infty}(\Lambda_{A, B})}\\&\leq4 a(\gamma^A\cup\gamma^B)e^{-\epsilon\cdot d(\gamma^A,\gamma^B)}
	\end{align}
	by application of a Lemma \ref{lem:kplemma}.
	
	Now keeping $A,B$ be fixed, we may take the limit $N\to\infty$ and thus we have 
	\begin{align}
		|\omega(AB)- \omega(A)\omega(B)|&\leq0.9\Vert A\Vert \Vert B\Vert  M e^{-\epsilon d(\Lambda_A,\Lambda_B)}e^{0.9 M e^{-\epsilon d(\Lambda_A,\Lambda_B)}}\\&
		\leq  \Vert A\Vert \Vert B\Vert M^{2.9} e^{-\epsilon d(\Lambda_A,\Lambda_B)}.
	\end{align}
\end{proof}

 \section{Code for Generating Paths in the Hexagonal Lattice}\label{sec:computer_code}
 This section provides the Python code for generating the hexagonal lattice paths which are enumerated and used in Tables 1-5.
 \subsection{Preliminary functions}\label{subsec:prelim}
 	This block of code contains functions which are called later to generate the numbers in Tables 1-5.
 	The variable \texttt{TOL} is the tolerance for floating point comparisons, which we use to compare if two paths are the same.
 	The function \texttt{is\_close(p1, p2, tol=TOL)} decides whether two coordinates \texttt{p1} and \texttt{p2} are within a tolerance for error of each other.
 	The function \texttt{in\_path\_with\_tol(p, path, tol=TOL)} decides whether a point \texttt{p} is in a path \texttt{path} with a given tolerance for error.
 	The function \texttt{starts2(n)} generates a line of \texttt{n} successive points of hexagonal lattice boundary.
 	The function \texttt{ends2(n,o)} generates a line of \texttt{n} successive points on the inner boundary if $\texttt{o}=1$ and on the outer boundary if $\texttt{o}=0$. The function \texttt{alls(n,o)} returns the points in \texttt{end2(n,o)} and \texttt{starts2(n)} in a single array.
 	The function \texttt{hits(a,b)} returns whether two paths \texttt{a} and \texttt{b} intersect. The functions \texttt{firsts(paths)} and \texttt{lasts(paths)} return the first and last coordinates respectively for each path in a list of paths labelled \texttt{paths}.
 	The function \texttt{generate\_saws(length, start\_paths, end\_paths, intersect\_path)} generates the walks on the boundary of length \texttt{length}, beginning at \texttt{start\_paths}, ending at \texttt{end\_paths}, and intersecting \texttt{intersect\_path}.
 	The function \texttt{generate\_loops(length, start\_paths, end\_paths, intersect\_path)} generates the loops on the boundary of length \texttt{length}, beginning at \texttt{start\_paths},  ending at \texttt{end\_paths}, and intersecting \texttt{intersect\_path}.
 \begin{lstlisting}[language=Python,breaklines=True]
import matplotlib.pyplot as plt
import numpy as np
import time
TOL = 1e-6
def is_close(p1, p2, tol=TOL):
    return abs(p1[0] - p2[0]) < tol and abs(p1[1] - p2[1]) < tol

def in_path_with_tol(p, path, tol=TOL): 
    return any(is_close(p, q, tol) for q in path)
def starts2(n):
    s=[]
    for i in range(0,n+1):
        s.append([[-0.5-i*1.5,3**0.5/2*(1-i)],[-1-i*1.5,3**0.5/2*(2-i)]])
    return s
def ends2(n, o):
    e=[]
    for i in range(0,n+1):
        if o==1:
            e.append([[0.5+i*1.5,3**0.5/2*(1-i)],[1+i*1.5,3**0.5/2*(2-i)]])
        elif o==0:
            val=2*(i+1.5)*3**0.5/2
            e.append([[0.5, val],[-0.5, val]])
    return e
def alls(length,o):
    return(starts2(length)+ends2(length,o))
def hits(a, b):
    return any(in_path_with_tol(p, b) for p in a)
def firsts(paths): 
    return [p[0] for p in paths]
def lasts(paths): 
    return [p[-1] for p in paths]
def generate_saws(length, start_paths, end_paths, intersect_path): 
    results = [] 
    end_firsts = firsts(end_paths)
    start_firsts = firsts(start_paths)
    all_end_start = start_firsts + end_firsts
    def recurse(path):
        if len(path) == length + 1:
            if in_path_with_tol(path[-1], end_firsts):
                if hits(path, intersect_path):
                    if not hits(path[1:-1], all_end_start):
                        if not any(all(is_close(p1, p2) for p1, p2 in zip(path[::-1], r)) for r in results):
                            results.append(path.copy())
            return
        last=path[-1]
        second_last=path[-2]
        dx = last[0]-second_last[0]
        dy = last[1]-second_last[1]
        left=[last[0]+dx/2-(3**0.5/2)*dy,last[1]+(3**0.5/2)*dx+dy/2]
        right=[last[0]+dx/2+(3**0.5/2)*dy,last[1]-(3**0.5/2)*dx+dy/2]
        for next_point in [left, right]:
            if not in_path_with_tol(next_point, path):
                if in_path_with_tol(last, end_firsts) and len(path)!=length:
                    continue
                path.append(next_point)
                recurse(path)
                path.pop()
    for sp in start_paths:
        recurse(sp[:])
    return results
def generate_loops(length, start_paths, end_paths, intersect_path):
    results = []
    end_firsts = firsts(end_paths)
    start_firsts = firsts(start_paths)
    all_end_start = start_firsts + end_firsts
    def recurse(path):
        if len(path) == length + 1:
            if in_path_with_tol(path[-1], end_firsts):
                if hits(path, intersect_path):
                    if not hits(path[1:-1], all_end_start):
                        if not any(all(is_close(p1, p2) for p1, p2 in zip(path[::-1], r)) for r in results):
                            results.append(path.copy())
            return
        last=path[-1]
        second_last=path[-2]
        dx = last[0]-second_last[0]
        dy = last[1]-second_last[1]
        left=[last[0]+dx/2-(3**0.5/2)*dy,last[1]+(3**0.5/2)*dx+dy/2]
        right=[last[0]+dx/2+(3**0.5/2)*dy,last[1]-(3**0.5/2)*dx+dy/2]
        for next_point in [left, right]:
            if not in_path_with_tol(next_point, path[1:]):
                if in_path_with_tol(last, end_firsts) and len(path)!=length:
                    continue
                path.append(next_point)
                recurse(path)
                path.pop()
    for sp in start_paths:
        recurse(sp[:])
    return results
\end{lstlisting}
\subsection{The values for Table \ref{table:loops}}\label{subsec:N}
The function \texttt{N(length)} counts all loops passing through a fixed edge of even length up to \texttt{length} using simple calls to the \texttt{generate\_loops} function above where the starting and ending points are the same edge. The values of this function for $\texttt{length}\leq 28$ appear in Table \ref{table:loops}.
\begin{lstlisting}[language=Python,breaklines=True]
def N(length):
    a=[]
    for i in range(6,length+2,2):
        a.append(len(generate_loops(i,starts2(0),starts2(0),starts2(0)[0])))
    return(a)
\end{lstlisting}
\subsection{The values for Table \ref{table:N}}\label{subsec:Pb}
The function \texttt{layered\_loop\_translates(mm,l)}  generates closest \texttt{mm} loops from the corner and \texttt{l} from the boundary.
The functions \texttt{get\_points(list\_of\_paths)}  and  \texttt{get\_edges(list\_of\_paths)} return a list of all vertices, respectively edges, in a list of paths called \texttt{list\_of\_paths}.
The function \texttt{Pb(length,o)} counts the upper bound on the number of walks of length up to \texttt{length} on the inner boundary if $\texttt{o}=1$ and outer boundary if $\texttt{o}=0$. The function \texttt{P(length)} returns the maximum of the two outputs. The values of this function for $\texttt{length}\leq 10$ appear in Table \ref{table:N}.
The function \texttt{plot2(p)} plots the paths.
\begin{lstlisting}[language=Python,breaklines=True]
def layered_loop_translates(mm,l):
    z=loop_translates(mm)
    for i in loop_translates(mm):
        for j in range(1,l+1):
            z.append(shift(i,0,3**.5*j))
    return(z)
def get_edges(list_of_paths):
    edges=[]
    for path in list_of_paths:
        for i in range(len(path)-1):
            edges.append([path[i],path[(i+1)%len(path)]])
    return(edges)
def get_points(list_of_paths):
    points=[]
    for path in list_of_paths:
        for i in range(len(path)-1):
            points.append(path[i])
    return(points)
def Pb(length,o):
    results=[]
    for i in range(1,length+1):
        running_max=[]
        for point in get_points(get_edges(layered_loop_translates(length,length))):
            a=0
            for edge in get_edges(layered_loop_translates(length,length)):
                if in_path_with_tol(point,edge) and is_close(edge[0],point):
                    a+=(len(saw_with_avoidance(i,[edge],alls(length,o),firsts(alls(length,o)),firsts(alls(length,o)))))
            running_max.append(a) 
        results.append(max(running_max))
    return(results)
def P(length):
    return(np.maximum(Pb(length,0),Pb(length,1)))
def plot2(p):
    plt.plot(*np.array(p).T)
\end{lstlisting}
\subsection{The values of columns 1-4 of Table \ref{table:saw}}\label{subsec:saws_col1-4}
The function \texttt{sww(length\_gamma,upper\_bound)} generates the values of $S_{ww}(l,l')$ where $l=\texttt{length\_gamma}$ and $l'\leq\texttt{upper\_bound}$. These numbers for $\texttt{length\_gamma}=3,4,5,6$ and $\texttt{upper\_bound}\leq 20$ appear in columns 1-4 of Table \ref{table:saw}.
\begin{lstlisting}[language=Python,breaklines=True]
def sww(length_gamma,upper_bound): 
    a=[[] for _ in range(upper_bound-2)] 
    gammas_outer=generate_saws(length_gamma,alls(upper_bound,0),alls(upper_bound,0),lasts(alls(upper_bound,1)))
    gammas_inner=generate_saws(length_gamma,alls(upper_bound,1),alls(upper_bound,1),lasts(alls(upper_bound,0))) 
    for self_avoiding_walk_inner in gammas_inner:
        for length_gamma_prime in range(3,upper_bound+1):
            a[length_gamma_prime-3].append(len(generate_saws(length_gamma_prime,alls(upper_bound+2,1),alls(u+2,1),self_avoiding_walk_inner))) 
    for self_avoiding_walk_outer in gammas_outer:
        for length_gamma_prime in range(3,upper_bound+1):
            a[length_gamma_prime-3].append(len(generate_saws(length_gamma_prime,alls(upper_bound+2,0),alls(upper_bound+2,0),self_avoiding_walk_outer))) 
    swws=[0]
    sumlist=np.cumsum(a, axis=0)
    for i in range(len(sumlist)):
        swws.append(max(sumlist[i])) 
    return(np.diff(swws))
\end{lstlisting}
\subsection{The values for column $5$ of Table \ref{table:saw}}\label{subsec:saws_col5}
The function \texttt{shift(dx,dy)}  shifts a path in the $x$ direction by \texttt{dx} and $y$ direction by \texttt{dy}.
The function \texttt{loop\_translates(m)} generates \texttt{m} translates of the loop which sits at the corner of the inner boundary along a line parallel to the edge.
The function \texttt{sl6w(upper\_bound)} calculates all values of $S_{lw}(l,l')$ for $l=6$ and $l'=\texttt{upper\_bound}$. These values for $\texttt{upper\_bound}\leq 20$ appear in column $5$ of Table \ref{table:saw}.
 \begin{lstlisting}[language=Python,breaklines=True]
def shift(paths,dx,dy):
    result=[]
    for path in paths:
        result.append([path[0]+dx,path[1]+dy])
    return(result)
def loop_translates(m):
    toploop=[[-0.5, 1.5*(3**.5)], [-1.0, 2*(3**.5)], [-0.5, 2.5*(3**.5)], 
 [0.5,2.5*(3**.5)], [1, 2*(3**.5)],
 [0.5, 1.5*(3**.5)], [-0.5, 1.5*(3**.5)]] 
    w=[]
    for ii in range(0,m):
        a=toploop
        w.append(shift(a,-ii*3/2,ii*(-3**.5/2)))
    return(w)
def sl6w(upper_bound):
    a=[[] for _ in range(upper_bound-2)] 
    for loop1 in loop_translates(upper_bound):
        for length_gamma_prime in range(3,upper_bound+1): 
            a[length_gamma_prime-3].append(len(generate_saws(length_gamma_prime,alls(upper_bound+2,0),alls(upper_bound+2,0),loop1))) 
    for loop2 in loop_translates(upper_bound):
        for length_gamma_prime in range(3,upper_bound+1):
            a[length_gamma_prime-3].append(len(generate_saws(length_gamma_prime,alls(upper_bound+2,1),alls(upper_bound+2,1),loop2))) 
    slw=[0]
    sumlist=np.cumsum(a, axis=0)
    for i in range(len(sumlist)):
        slw.append(max(sumlist[i]))
    return(np.diff(slw)) 
\end{lstlisting}

\subsection{The values for Table \ref{table:R}}\label{subsec:R}
The function \texttt{saw\_with\_avoidance} is similar to the function \texttt{generate\_saws}, that generates the walks on the boundary of length \texttt{n}, beginning at \texttt{start\_paths}, ending at \texttt{end\_paths}, intersecting \texttt{intersect\_path}, and avoiding \texttt{avoid\_path}. It does this in the same manner as the function \texttt{generate\_saws} above.
The function \texttt{R(length)} returns number of lines of length \texttt{length} hitting a fixed bond in the boundary and returning to the boundary to the left of the fixed bond. These numbers for $\texttt{length}\leq20$ are compiled in Table \ref{table:R}.
\begin{lstlisting}[language=Python,breaklines=True]
def saw_with_avoidance(length, start_paths, end_paths, intersect_path, avoid_path):
    results = []
    end_firsts = firsts(end_paths)
    start_firsts = firsts(start_paths)
    def recurse(path):
        if len(path) == length + 1:
            if in_path_with_tol(path[-1], end_firsts):
                if hits(path, intersect_path):
                    if not hits(path[1:-1], avoid_path):
                        if not any(all(is_close(p1, p2) for p1, p2 in zip(path[::-1], r)) for r in results):
                            results.append(path.copy())
            return
        last=path[-1]
        second_last=path[-2]
        dx = last[0]-second_last[0]
        dy = last[1]-second_last[1]
        left=[last[0]+dx/2-(3**0.5/2)*dy,last[1]+(3**0.5/2)*dx+dy/2]
        right=[last[0]+dx/2+(3**0.5/2)*dy,last[1]-(3**0.5/2)*dx+dy/2]
        for next_point in [left, right]:
            if not in_path_with_tol(next_point, path):
                if in_path_with_tol(last, end_firsts) and len(path)!=length:
                    continue
                path.append(next_point)
                recurse(path)
                path.pop()
    for sp in start_paths:
        recurse(sp[:])
    return results
def R(length):
    a=[]
    z=len(starts2(length))//2
    for i in range(4,length+2,2):
        a.append(len(saw_with_avoidance(i,[starts2(length)[z]],starts2(length)[:z],firsts(starts2(length)),firsts(starts2(length)[z:]))))
    return(a)
\end{lstlisting}
\subsection{The values for Table \ref{table:Q}}\label{subsec:Q}
The function \texttt{Qb(length,o)} counts the number of walks of odd length up to \texttt{length} on the inner boundary if $\texttt{o}=0$ and outer boundary if $\texttt{o}=1$. The function \texttt{Q(length)} returns the maximum of the two outputs. The values of this function for $\texttt{length}\leq19$ appear in Table \ref{table:Q}.
\begin{lstlisting}[language=Python,breaklines=True]
def Qb(length,o): 
    a=[0]*int((length-3)//2+1)
    for i in range(3,length+2,2):
        a[i//2-1]+=len(generate_saws(i,starts2(length),ends2(length,o),firsts(ends2(length,o))))
    return(a)
def Q(length):
    return(np.maximum(Qb(length,1),Qb(length,0)))
\end{lstlisting}

\section{Code for Generating Edge Self-Avoiding Walks on the Square Lattice}\label{subsec:lieb}
This section provides Python code for the numbers used in Section \ref{sec:lieb}.
The variable \texttt{TOL} is the tolerance for floating point comparisons, which we use to compare if two paths are the same.
The function \texttt{is\_close(p1, p2, tol=TOL)} decides whether two coordinates \texttt{p1} and \texttt{p2} are within a tolerance for error of each other.
The function \texttt{in\_path\_with\_tol(p, path, tol=TOL)} decides whether a point \texttt{p} is in a path \texttt{path} with a given tolerance for error.
 The function \texttt{starts2(n)} generates a line of \texttt{n} successive points of square lattice boundary.
 The function \texttt{ends2(n,o)} generates a line of \texttt{n} successive points on the inner boundary if $\texttt{o}=1$ and on the outer boundary if $\texttt{o}=0$. The function \texttt{alls(n,o)} returns the points in \texttt{end2(n,o)} and \texttt{starts2(n)} in a single array.
 The function \texttt{hits(a,b)} returns whether two paths \texttt{a} and \texttt{b} intersect. The function \texttt{firsts(paths)} returns the first coordinates for each path in a list of paths labelled \texttt{paths}.
 The function \texttt{esaw\_with\_avoidance} is similar to the function \texttt{saw\_with\_avoidance} above but on the square lattice; it generates the edge self-avoiding walks on the boundary of length \texttt{n}, beginning at \texttt{start\_paths}, ending at \texttt{end\_paths}, intersecting \texttt{intersect\_path}, and avoiding \texttt{avoid\_path}. It does this in the same manner as the function \texttt{generate\_saws} above.
 The function \texttt{vertices(n)} generates a list of all the vertices within distance \texttt{n} of the origin.
 The function \texttt{maxes(i)} finds the maximum number of paths beginning and ending on either the inner or outer boundary of length \texttt{i} that intersect a single vertex. 
\begin{lstlisting}[language=Python,breaklines=True]
import matplotlib.pyplot as plt
import numpy as np
TOL = 1e-6
def is_close(p1, p2, tol=TOL):
	return abs(p1[0] - p2[0]) < tol and abs(p1[1] - p2[1]) < tol
def in_path_with_tol(p, path, tol=TOL): 
	return any(is_close(p, q, tol) for q in path)
def starts2(n):
	s=[]
	for i in range(1,n+1):
		s.append([[0,i],[1,i]])
	return s
def ends2(n, o):
    e=[]
    for i in range(1,n+1):
        if o==1:
            e.append([[1-i,1],[1-i,0]])
        elif o==0:
            e.append([[i,0],[i, 1]])
    return e
def alls(length,o):
	return(starts2(length)+ends2(length,o)
def firsts(paths):
	return [p[0] for p in paths]
def esaw_with_avoidance(length, start_paths, end_paths, intersect_path, avoid_path):
    results=[]
    end_firsts=firsts(end_paths)
    start_firsts=firsts(start_paths)
    def recurse(path,visited_edges):
        if len(path)==length + 1:
            if in_path_with_tol(path[-1],end_firsts):
                if hits(path,intersect_path):
                    if not hits(path[1:-1],avoid_path):
                        if not any(all(is_close(p1, p2) for p1,p2 in zip(path[::-1],r)) for r in results):
                            results.append(path.copy())
            return
        last=path[-1]
        second_last=path[-2]
        dx=last[0]-second_last[0]
        dy=last[1]-second_last[1]
        left=[last[0]-dy,last[1]+dx]
        right=[last[0]+dy,last[1]-dx]
        middle=[last[0]+dx,last[1]+dy]

        for next_point in [left, right, middle]:
            edge=frozenset((tuple(last), tuple(next_point)))
            if edge in visited_edges:
                continue
            if in_path_with_tol(last, end_firsts) and len(path)!=length:
                continue
            visited_edges.add(edge)
            path.append(next_point)
            recurse(path, visited_edges)
            path.pop()
            visited_edges.remove(edge)
    for sp in start_paths:
        initial_path=sp[:]
        initial_edges=set()
        for i in range(1,len(initial_path)):
            edge=frozenset((tuple(initial_path[i - 1]),tuple(initial_path[i])))
            initial_edges.add(edge)
        recurse(initial_path,initial_edges)
    return results
def vertices(n):
    e=[]
    for i in range(n):
        for j in range(n):
            e.append([[i,j]])
    return(e)
def maxes(i):
    a=[]
    for e in vertices(n):
        a.append(len(saw_with_avoidance(i,starts2(2*i),alls(2*i,0),e,firsts(alls(2*i,0))))+len(saw_with_avoidance(i,ends2(2*i,0),ends2(2*i,0),e,firsts(alls(2*i,0)))))
        a.append(len(saw_with_avoidance(i,starts2(2*i),alls(2*i,1),e,firsts(alls(2*i,1))))+len(saw_with_avoidance(i,ends2(2*i,1),ends2(2*i,1),e,firsts(alls(2*i,1)))))
    return(max(a))
\end{lstlisting}
	
\section*{Statements and Declarations}	

\subsection*{Data availability statement}

All data associated with this manuscript are included in the manuscript as well as the code used to produce the data (Appendix \ref{sec:computer_code}).

\subsection*{Conflict of Interest statement}

The authors have no relevant financial or non-financial interests to disclose.


\end{document}